\documentclass[manuscript, nonacm, dvipsnames]{acmart}

\usepackage{color}
\usepackage{xspace}
\usepackage{latexsym}
\usepackage{array}
\usepackage{todonotes}
\usepackage{amsmath}
\usepackage{amsthm}
\usepackage{amsfonts}
\usepackage{url}
\usepackage{graphicx}
\usepackage{colortbl}
\usepackage{multirow}

\usepackage{tikz,tikz-qtree,tikz-qtree-compat,tikz-3dplot}
\usetikzlibrary{shapes,decorations.markings,arrows.meta,fit}
\pdfpageattr{/Group <</S /Transparency /I true /CS /DeviceRGB>>}

\usepackage{algorithmicx}
\usepackage{algorithm}
\usepackage[noend]{algpseudocode}

\usepackage{hhline}

\colorlet{DarkGreen}{green!50!black}

\renewcommand{\vec}[1]{\ensuremath\boldsymbol{#1}}
\newcommand{\cd}{\ \leftarrow \ }
\newcommand{\bcq}{\text{\sf BCQ}}
\newcommand{\set}[1]{\left\{ #1 \right\}}

\newcommand{\setof}[2]{\left\{ #1 \suchthat #2 \right\}}

\newcommand{\proofseq}{\text{\sf ProofSeq}}
\newcommand{\myskip}{\hskip .6em}

\newcommand{\zy}{\text{\sf ZY}}
\newcommand{\filter}{\text{\sf filter}}
\newcommand{\minimaxw}{\text{\sf Minimaxwidth}}
\newcommand{\maximinw}{\text{\sf Maximinwidth}}

\newcommand{\vertexbound}{\text{\sf VB}}
\newcommand{\daentropic}{\text{\sf DAEB}}
\newcommand{\outputbound}{\text{\sf LogSizeBound}}
\newcommand{\outputsize}{|\text{\sf output}|}
\newcommand{\dapolymatroid}{\text{\sf DAPB}}

\newcommand{\csma}{\textsf{CSMA}}

\newcommand{\sa}{\text{\sf SA}}
\newcommand{\ed}{\text{\sf ED}}
\newcommand{\ec}{\text{\sf ECP}}
\newcommand{\vd}{\text{\sf VD}}

\newcommand{\panda}{\textsf{PANDA}}
\newcommand{\Dom}{\textsf{Dom}}
\newcommand{\faqcs}{\text{\sf FAQ-SS}}
\newcommand{\faq}{\text{\sf FAQ}}

\newcommand{\AGM}{\textsf{AGM}}
\newcommand{\agm}{\textsf{AGM}}

\newcommand{\csp}{\textsf{CSP}}
\newcommand{\FD}{\textsf{FD}}
\newcommand{\dc}{\textsf{DC}}
\newcommand{\cc}{\textsf{CC}}
\newcommand{\td}{\textsf{TD}}
\newcommand{\hdc}{\textsf{HDC}}
\newcommand{\hcc}{\textsf{HCC}}
\newcommand{\hfd}{\textsf{HFD}}

\newcommand{\image}{\text{\sf image}}
\newcommand{\tw}{\text{\sf tw}}
\newcommand{\ghtw}{\text{\sf ghtw}}

\newcommand{\fhtw}{\text{\sf fhtw}}

\newcommand{\adw}{\text{\sf adw}}
\newcommand{\subw}{\text{\sf subw}}
\newcommand{\dasubw}{\text{\sf da-subw}}
\newcommand{\dafhtw}{\text{\sf da-fhtw}}
\newcommand{\edafhtw}{\text{\sf eda-fhtw}}
\newcommand{\edasubw}{\text{\sf eda-subw}}

\newcommand{\calH}{\mathcal H}
\newcommand{\calA}{\mathcal A}
\newcommand{\calB}{\mathcal B}
\newcommand{\calK}{\mathcal K}
\newcommand{\calC}{\mathcal C}

\newcommand{\calP}{\mathcal P}
\newcommand{\calV}{{[n]}}
\newcommand{\mcalV}{\mathcal V}
\newcommand{\calE}{\mathcal E}
\newcommand{\calF}{\mathcal F}
\newcommand{\calT}{\mathcal T}
\newcommand{\calM}{\mathcal M}
\newcommand{\calR}{\mathcal R}

\newcommand{\calG}{\mathcal G}

\newcommand{\bB}{{\mathbf B}}

\newcommand{\incomp}{\perp}

\newcommand{\eat}[1]{}

\newcommand{\flow}{{\sf inflow}}
\newcommand{\flowsa}{{\sf flow}}
\newcommand{\outflow}{{\sf outflow}}

\newcommand{\opt}{{\textsf{OPT}}}
\newcommand{\obj}{{\textsf{OBJ}}}
\newcommand{\Deg}{\deg}

\newcommand{\poly}{\mathrm{poly}}

\newcommand{\F}{\mathbb{F}}

\newcommand{\R}{\mathbb{R}}
\newcommand{\Mod}{\text{\sf M}}
\newcommand{\Q}{\mathbb{Q}}
\newcommand{\N}{\mathbb{N}}

\newcommand{\argmax}{\mathop{\text{argmax}}}

\newcommand{\np}{\mathsf{NP}}

\newcommand{\ptime}{\mathsf{PTIME}}
\newcommand{\fpt}{\mathsf{FPT}}

\newcommand{\pr}{\mathop{\textnormal{Pr}}}

\newcommand{\be}{\begin{enumerate}}
\newcommand{\ee}{\end{enumerate}}
\newcommand{\bi}{\begin{itemize}}
\newcommand{\ei}{\end{itemize}}
\newcommand{\beq}{\begin{equation}}
\newcommand{\eeq}{\end{equation}}

\newcommand{\bp}{\begin{proof}}
\newcommand{\ep}{\end{proof}}
\newcommand{\bcor}{\begin{cor}}
\newcommand{\ecor}{\end{cor}}
\newcommand{\bthm}{\begin{thm}}
\newcommand{\ethm}{\end{thm}}
\newcommand{\blmm}{\begin{lmm}}
\newcommand{\elmm}{\end{lmm}}
\newcommand{\bdefn}{\begin{defn}}
\newcommand{\edefn}{\end{defn}}
\newcommand{\bprop}{\begin{prop}}
\newcommand{\eprop}{\end{prop}}
\newcommand{\bconj}{\begin{conj}}
\newcommand{\econj}{\end{conj}}
\newcommand{\bopm}{\begin{opm}}
\newcommand{\eopm}{\end{opm}}
\newcommand{\brmk}{\begin{rmk}}
\newcommand{\ermk}{\end{rmk}}

\newcommand{\norm}[1]{\|#1\|}

\newcommand{\suchthat}{\ | \ }
\newcommand{\inner}[1]{\langle #1 \rangle}

\newcommand{\mv}[1]{\mathbf{#1}}

\theoremstyle{plain}                   
\newtheorem{thm}{Theorem}[section]
\newtheorem{lmm}[thm]{Lemma}
\newtheorem{prop}[thm]{Proposition}
\newtheorem{cor}[thm]{Corollary}

\theoremstyle{definition}              

\newtheorem{opm}{Question}
\newtheorem{conj}[thm]{Conjecture}
\newtheorem{ex}[thm]{Example}

\newtheorem{defn}[thm]{Definition}

\newtheorem{rmk}[thm]{Remark}
\newtheorem{claim}{Claim}

\newcommand{\bbox}{
\begin{center}
\begin{tabular}{|c|}
\hline
}
\newcommand{\ebox}{
\\
\hline
\end{tabular}
\end{center}
}

\newlength{\toppush}
\def\subjname{Probabilistically Checkable Proofs and Inapproximability}
\def\doheading#1#2#3{\vfill\eject\vspace*{-\toppush}%
  \vbox{\hbox to\textwidth{{\bf}
  \subjnum: \subjname
  \hfil Lecturer: Hung Q. Ngo}%
  \hbox to\textwidth{{\bf} SUNY at Buffalo, Fall 2004\hfil#3\strut}%
  \hrule}
}

\newcommand{\defeq}{\stackrel{\mathrm{def}}{=}}

\algrenewcommand\algorithmicrequire{\textbf{Input:}}
\algrenewcommand\algorithmicensure{\textbf{Output:}}
\algrenewcommand\algorithmicwhile{\textbf{While}}
\algrenewcommand\algorithmicfor{\textbf{For}}
\algrenewcommand\algorithmicreturn{\textbf{Return}}
\algrenewcommand\algorithmicif{\textbf{If}}

\newcommand{\X}{\mathcal{X}}

\allowdisplaybreaks

\usepackage{booktabs} 

\usepackage{listings} 

\usepackage{xpatch}
\usepackage{textcase}
\makeatletter
\xpatchcmd{\@sect}{\uppercase}{\MakeTextUppercase}{}{}
\xpatchcmd{\@sect}{\uppercase}{\MakeTextUppercase}{}{}
\makeatother

\title[Shannon-flow inequalities, submodular width, and disjunctive datalog]{What do Shannon-type Inequalities, 
Submodular Width, and Disjunctive Datalog have to do with one another?}

\author{Mahmoud Abo Khamis}
\orcid{0000-0003-3894-6494}
\affiliation{%
  \institution{Relational\underline{AI}}
  \city{Berkeley}
  \state{CA}
  \country{USA}
}
\author{Hung Q. Ngo}
\affiliation{%
  \institution{Relational\underline{AI}}
  \city{Berkeley}
  \state{CA}
  \country{USA}
}
\author{Dan Suciu}
\affiliation{%
   \institution{University of Washington}
   \city{Seattle}
   \state{WA}
   \country{USA}
}

\thanks{An extended abstract of this manuscript appeared in the Proceedings of the 36th ACM Symposium on Principles of Database Systems {\em (PODS '17)} \cite{panda-pods}.\\
This work is partly supported by NSF grant 1535565.}

\begin{document}

\begin{abstract}
Recent works on bounding the output size of a conjunctive query with
functional dependencies and degree constraints have shown a deep connection between
fundamental questions in information theory and database theory.
We prove analogous output bounds for {\em disjunctive datalog rules}, and
answer several open questions regarding the tightness and looseness of these
bounds along the way.
Our bounds are intimately related to Shannon-type information inequalities.
We devise the notion of a ``proof sequence'' of a specific class of
Shannon-type information inequalities called ``Shannon flow inequalities''. 
We then show how such a proof sequence can be interpreted as symbolic instructions 
guiding an algorithm called $\panda$, which
answers disjunctive datalog rules within the time that the size bound predicted.
We show that $\panda$ can be used as a black-box to devise algorithms matching
precisely the fractional hypertree width and the submodular width runtimes
for aggregate and conjunctive queries {\em with} functional dependencies and 
degree constraints.

Our results improve upon known results in three ways. First, our bounds and
algorithms are for the much more general class of disjunctive datalog rules, of 
which conjunctive queries are a special case. 
Second, the runtime of $\panda$ matches precisely the submodular width bound, 
while the previous algorithm by Marx has a runtime that is polynomial in this 
bound. Third, our bounds and algorithms work for queries with input 
cardinality bounds, functional dependencies, {\em and} degree constraints.

Overall, our results show a deep connection between three seemingly
unrelated lines of research; and, our results on proof sequences for Shannon
flow inequalities might be of independent interest.
\end{abstract}

\maketitle


\section{Introduction}
\label{sec:the:problems}

This paper answers four major questions that resulted from four
different research threads, and establishes new connections between those threads.

\subsection{Output-Size Bound for Full Conjunctive Queries}

Grohe and Marx~\cite{DBLP:journals/talg/GroheM14},
Atserias, Grohe, and Marx~\cite{AGM}, and Gottlob, Lee, Valiant and Valiant \cite{GLVV}
developed a deep connection between the output size bound of a
conjunctive query with (or without) functional dependencies (FD) and
information theory. Our first problem is to extend this bound to
degree constraints, and to study whether the bound is tight.

We associate a full conjunctive query $Q$ to a hypergraph
$\calH\defeq ([n],\calE)$, $\calE \subseteq 2^{[n]}$ (where $[n]=\{1,\ldots,n\}$). The query's variables are
$A_i$, $i \in [n]$. Its atoms are $R_F, F \in \calE$. The query is:
\begin{equation}
   Q(\mv A_{[n]}) \cd \bigwedge_{F\in\calE} R_F(\mv A_F), \label{eqn:conjunctive:query}
\end{equation}
where $\mv A_J$ denotes the tuple $(A_j)_{j\in J}$, for
any $J \subseteq[n]$. Our goal is to compute an upper bound on the output
size, when the input database satisfies a set of {\em degree constraints}.

\bdefn[Degree, cardinality, and FD constraints]
For $X\subset Y\subseteq F \in \calE$,
define
\begin{equation}
   \Deg_F(\mv A_Y | \mv A_X) \defeq \max_{\mv t} |\Pi_{\mv A_Y}(\sigma_{\mv
   A_X=\mv t}(R_F))|,
\end{equation}
where $\Pi$ and $\sigma$ are the projection and selection operators in relational algebra
respectively.
Then, a {\em degree constraint} is an assertion of the form
$\Deg_F(\mv A_Y | \mv A_X) \leq N_{Y|X}$, where $N_{Y|X}$ is a natural number.
A {\em cardinality constraint} is an
assertion of the form $|R_F| \leq N_F$, for some $F \in \calE$; it is exactly
the degree constraint $\Deg_F(\mv A_F | \emptyset) \leq N_{F|\emptyset} \defeq N_F$.
A {\em functional dependency} $\mv A_X \rightarrow \mv A_Y$ is a
degree constraint with $N_{X\cup Y|X}=1$.
In particular, degree constraints strictly generalize both cardinality
constraints and FDs.
\edefn

Handling queries with degree constraints has a strong practical motivation.
For example, Armbrust et al.~\cite{DBLP:conf/cidr/ArmbrustFPLTTO09,
DBLP:journals/pvldb/ArmbrustCKFFP11,
DBLP:conf/sigmod/ArmbrustLKFFP13}
described a new approach to query evaluation, called {\em scale-independent
query processing}, which guarantees a fixed runtime even
when the size of the database increases without bound; this guarantee
is provided by asking developers to write explicit degree constraints,
then using heuristics to derive upper bounds on the query output.
Thus, improved upper bounds on the size of the query answer have
immediate applications to scale-independent query processing. Some complexity
results on the associated decision problem (``is the output size of the query
bounded?'') were considered
in~\cite{DBLP:journals/tods/BenediktCT16,DBLP:journals/pvldb/BenediktLT15,DBLP:journals/pvldb/CaoFWY14}.

Built in part on earlier work by~Friedgut and Kahn~\cite{MR1639767} and Chung et
al.~\cite{MR859293},
the first output-size upper bound for full conjunctive queries was established
in~\cite{DBLP:journals/talg/GroheM14,AGM}. Their bound,
known today as the {\em $\agm$-bound} (see Section~\ref{sec:background}),
was tight, but only for queries with cardinality constraints.
Extensions of the bound to handle FDs and degree constraints were discussed in
Gottlob et al.~\cite{GLVV} and Abo Khamis et al.~\cite{csma}, respectively, who
left open the question of whether these bounds are tight.
Our first question is whether the upper bounds in~\cite{GLVV, csma} for queries
with FDs or degree constraints (in addition to cardinality constraints) are
tight.

To set the technical context for this question, we briefly describe how the
bounds were derived.  Fix an input database $\mv D$ and consider the joint
distribution on random variables $\mv A_{[n]}$ where each output tuple
$\mv t \in Q(\mv D)$ is selected uniformly with probability $1/|Q(\mv D)|$. For
any $S\subseteq \calV$, let
$H(\mv A_S)$ denote the marginal entropy on the variables $\mv A_S$.
Then,\footnote{All logs are in base 2, unless otherwise stated.} by uniformity
$H(\mv A_{[n]})=\log |Q(\mv D)|$,
and $H(\mv A_Y| \mv A_X) \leq \log N_{Y|X}$ for every degree constraint.
A function $h : 2^{\mv A_{[n]}} \to \R_+$ is said to be {\em entropic} if there
is a joint distribution on $\mv A_{[n]}$ such that $h(\mv A_S)$ is the marginal
entropy on $\mv A_S$, $S\subseteq [n]$.
We just proved the {\em entropic bound} of a query, which states that
$\log |Q| \leq \max_h h(\mv A_{[n]})$, where $h$ ranges over all entropic
functions satisfying the given degree constraints.
Recently, Gogacz and Toru{\'n}czyk~\cite{szymon-2015} showed
that the entropic bound is tight given cardinality and FD constraints.
However, they did not address general degree constraints.

The problem with the entropic bound is that we do not know how to compute it
(except for the special case when all degree constraints are cardinality
constraints), partly because the entropic cone is characterized by infinitely
many {\em non-Shannon-type inequalities}~\cite{Yeung:2008:ITN:1457455,matus2007infinitely}.  To
overcome this limitation, Gottlob et al.~\cite{GLVV} replace
entropic functions (which are difficult) by polymatroids (which are
easier).  A {\em polymatroid} is a set function $h : 2^{[n]} \to \R_+$
that is non-negative,
{\em monotone} (i.e. satisfies $h(X) \leq h(Y)$ for all $X \subseteq Y\subseteq [n]$),
and {\em submodular} (i.e. satisfies $h(X \cup Y) + h(X \cap Y) \leq h(X) + h(Y)$ for all $X, Y \subseteq [n]$),
with $h(\emptyset)=0$.
Every entropic function $h$ is also a polymatroid,
if we write $h(S)$ for $h(\mv A_S)$ (see Section~\ref{sec:background}).
Linear inequalities satisfied by all polymatroids are called
{\em Shannon-type inequalities}~\cite{Yeung:2008:ITN:1457455}.
The {\em polymatroid bound} of a full conjunctive query is
$\log |Q| \leq \max_h h(\mv A_{[n]})$, where $h$ ranges over all
polymatroids satisfying the given constraints.
%
The polymatroid bound, while at least as large as the entropic bound,
can be shown to be tight for cardinality constraints, because
the AGM bound is exactly the polymatroid bound for cardinality constraints
(see Proposition~\ref{prop:unifying:bound}) and it is tight~\cite{AGM}.
The polymatroid bound is also tight for cardinality constraints with
certain sets of FDs~\cite{csma}.
We ask whether it is tight in more general settings:

\bopm Is the polymatroid bound tight for general degree constraints?  Or, at
least for queries with both cardinality and FD constraints?
\eopm

\begin{ex} \label{ex:intro:1}
  Consider the 4-cycle query:
  \begin{align}
     Q(A_1,A_2,A_3,A_4) & \cd R_{12}(A_1,A_2),R_{23}(A_2,A_3),
                              R_{34}(A_3,A_4),R_{41}(A_4,A_1) \label{eq:c4}
  \end{align}
  Assuming all input relations have size $\leq N$, then (a) the AGM
  bound is $|Q|\leq N^2$, (b) if we add the degree constraints
   $\Deg_{12}(A_1A_2 | A_1) \leq D$ and
   $\Deg_{{12}}(A_1A_2 | {A_2}) \leq D$ for some integer $D\leq \sqrt{N}$ then
  $|Q| \leq D\cdot N^{3/2}$, and (c) if we replace the degree constraints
  with FDs $A_1 \rightarrow A_2$ and $A_2\rightarrow A_1$ the bound
  reduces further to $|Q| \leq N^{3/2}$.  These bounds can be proven using
  only Shannon-type inequalities, thus they are polymatroid bounds.
  They are also asymptotically tight (see Appendix~\ref{app:sec:the:problems}).
\end{ex}

\noindent
\textbf{Answer 1.} The polymatroid bound is {\em not} tight for queries with
cardinality and FD constraints! By adding a variable to a non-Shannon
inequality by Zhang-Yeung~\cite{DBLP:journals/tit/ZhangY98} and constructing
accordingly a database instance, we prove in
Section~\ref{subsec:glvv:bound:is:not:tight} the following theorem.

\begin{thm}
  For any integer $s>0$, there exists a query $Q$ with output size $\Theta(s)$
  and cardinality and FD constraints, such that the ratio between the
  polymatroid bound and the entropic bound is $\geq N^s$, where $N$ is the
  size of the database.
\label{thm:glvv:not:tight}
\end{thm}

\subsection{Size Bound for Disjunctive Datalog Rules}
Disjunctive datalog~\cite{DBLP:journals/tods/EiterGM97,DBLP:conf/datalog/AlvianoFLPPT10} is a powerful extension
of datalog.
In this paper we are interested in a single disjunctive-datalog rule:
\begin{equation}
   P: \ \ \ \bigvee_{B \in \calB}T_B(\mv A_B) \cd \bigwedge_{F\in\calE} R_F(\mv A_F)
   \label{eqn:disjunctive:datalog:query}
\end{equation}
The body is similar to that of a conjunctive query, while the head is
a disjunction of output relations $T_B$, which we call {\em
targets}. Given a database instance $\mv D$, a {\em model} of $P$ is a tuple
$\mv T=(T_B)_{B\in\calB}$ of relations, one for each target,
such that the logical implication indicated by the rule holds.
More precisely, for any
tuple $\mv t$, if $\Pi_F(\mv t) \in R_F$ for every input relation
$R_F$, then there exists a target $T_B \in \mv T$ such that
$\Pi_B(\mv t) \in T_B$.  We write $\mv T \models P$ to denote the fact
that $\mv T$ is a model.  Define the size of a model to be
$\max_B |T_B|$, and define the {\em output size} of $P$ to be the minimum size
over all models:
\begin{align}
|P(\mv D)| \defeq \min_{\mv T: \mv T \models P} \max_{B \in \calB} |T_B|\label{eq:def:size}
\end{align}

Our second question is to find an output size upper bound for a disjunctive
datalog rule whose input database $\mv D$ satisfies the given degree
constraints. If the rule has a single target then it becomes a
conjunctive query: a model is any superset of the answer, and
the output size is the standard size of the query's answer.  We thus
expect the upper bound to come in two flavors, entropic and
polymatroid, as is the case for full conjunctive queries.

\bopm Find the entropic and polymatroid output size bounds of a
disjunctive datalog rule, under general degree constraints.
Determine if it is tight.
\eopm

\begin{ex} \label{ex:intro:2}
  Consider the disjunctive datalog rule, where input relations have sizes $\leq
   N$:
    \begin{align*}
      P:\ \ \ & T_{123}(A_1,A_2,A_3)\vee T_{234}(A_2,A_3,A_4) \cd
            R_{12}(A_1,A_2),R_{23}(A_2,A_3),R_{34}(A_3,A_4).
    \end{align*}
    Intuitively, for every tuple $\mv t = (a_1,a_2,a_3,a_4)$ in
    $R_{12} \Join R_{23} \Join R_{34}$ we want to have either $(a_1,a_2,a_3)$ in $T_{123}$ or $(a_2,a_3,a_4)$ in $T_{234}$ or both. A model of size $N^3$ can be obtained
    trivially by populating the target $T_{123}$ with all triples
    obtained from the active domain, but we show below that
    $|P(\mv D)| \leq N^{3/2}$, for all $\mv D$ whose relation sizes are $\leq
    N$.
\end{ex}

\noindent
\textbf{Answer 2.} To describe the answer to the second question, we recall
some standard notations~\cite{Yeung:2008:ITN:1457455}.
We identify set-functions
$h : 2^{[n]} \to \R_+$ with vectors in $\R_+^{2^n}$, and we use both
$h(\mv A_S)$ and $h(S)$ to denote $h_S$, where $A_1, A_2, \ldots$
are (random) variables. (We will use $h(S)$ and $h(\mv A_S)$ interchangeably in
this paper, depending on context.
The reason is that, $h(\mv A_S)$ is more apt for marginal entropies,
and $h(S)$ is more apt for polymatroids.)
The sets
$\Gamma^*_n \subset \overline \Gamma^*_n \subset \Gamma_n \subset
\R_+^{2^n}$
denote the set of entropic functions, its topological closure, and the
set of polymatroids. (See Figure~\ref{fig:set:functions} and Definition~\ref{def:set-functions}.)
We encode degree constraints by a set $\dc$ of triples
$(X,Y,N_{Y|X})$, specifying $\Deg_F(\mv A_Y | \mv A_X) \leq N_{ Y| X}$. Define
\begin{equation}
   \hspace{-8pt}
   \hdc \defeq \setof{h : 2^\calV \to \R_+}{\bigwedge_{(X,Y,N_{Y|X}) \in \dc}h( Y |  X) \leq \log N_{
      Y| X}}
    \label{eqn:hdc}
\end{equation}
to be the collection of set functions $h$ satisfying the constraints $\dc$, where
$h( Y |  X) \defeq h(Y) - h( X)$.\footnote{If $h$ was entropic, then
$h(Y | X) = h(Y)-h(X)$ is the conditional entropy. Recall also $X \subseteq
Y$ whenever $(X,Y,N_{Y|X}) \in \dc$.} Fix a closed
subset $\calF \subseteq \R_+^{2^n}$.  Define the {\em
$\log$-size-bound} with respect to $\calF$ of a disjunctive datalog rule $P$
to be the quantity:
\begin{equation}
   \outputbound_\calF(P) \defeq \max_{h\in\calF} \min_{B\in\calB} h(B).
   \label{eqn:intro:output:bound}
\end{equation}
The following is our second result, whose proof can be found in
Section~\ref{sec:size:bound:ddl}.

\bthm\label{thm:intro:size:bound:disjunctive}
Let $P$ be any disjunctive datalog rule~\eqref{eqn:disjunctive:datalog:query},
and $\dc$ be given degree constraints.
\bi
 \item[(i)] For any database instance $\mv D$ satisfying all constraints in
    $\dc$, the following holds:
\begin{eqnarray}
  \log |P(\mv D)|
  &\leq& \underbrace{\outputbound_{\overline\Gamma^*_n \cap \hdc}(P)}_{\text{entropic bound}}
         \label{eqn:intro:entropic:disjunctive:2}\\
  &\leq& \underbrace{\outputbound_{\Gamma_n \cap \hdc}(P)}_{\text{polymatroid bound}}
         \label{eqn:intro:polymatroid:disjunctive:2}
\end{eqnarray}
 \item[(ii)] The entropic bound above is asymptotically tight.
 \item[(iii)] The polymatroid bound is not tight, even if the constraints
   are all cardinality constraints, and even if all cardinality upperbounds are
   identical. Furthermore, the gap between the two bounds can be arbitrarily
   large.
\ei
\ethm


Inequalities~\eqref{eqn:intro:entropic:disjunctive:2} and
\eqref{eqn:intro:polymatroid:disjunctive:2} generalize the entropic
and polymatroid bounds from full conjunctive queries
(Proposition~\ref{prop:unifying:bound}) to arbitrary disjunctive datalog rules.
The tightness result $(ii)$ generalizes the main result in~\cite{szymon-2015},
which states that the entropic bound is
asymptotically tight for full conjunctive queries under FDs. Note that the
non-tightness result $(iii)$ is incomparable to the non-tightness
result in Theorem~\ref{thm:glvv:not:tight}.

\begin{ex} From~\eqref{eqn:intro:polymatroid:disjunctive:2},
the bound $|P(\mv D)| \leq N^{3/2}$ in Example~\ref{ex:intro:2} follows
by applying twice the submodularity law for polymatroids:
\begin{align*}
   3\log N
   & \geq h(A_1A_2) + h(A_2A_3) + h(A_3A_4) &\text{(cardinality constraints)} \\
   & \geq h(A_1A_2A_3) + h(A_2) + h(A_3A_4) &\text{(submodularity)} \\
   & \geq h(A_1A_2A_3) + h(A_2A_3A_4) &\text{(submodularity and $h(\emptyset) = 0$)} \\
   & \geq 2 \min(h(A_1A_2A_3),h(A_2A_3A_4)) &\\
   & \geq 2\log |P|. &\text{(From~\eqref{eqn:intro:polymatroid:disjunctive:2})}
\end{align*}
\label{ex:intro:2b}
\end{ex}

Table~\ref{tab:bounds} summarizes size bounds for full conjunctive queries and disjunctive datalog rules and states their tightness properties.
\colorlet{clr:entropy}{blue}
\colorlet{clr:polymatroid}{red}
\begin{table}[th!]
   \begin{tabular}{|c|c|c|c|}
      \hline
      \rowcolor{gray!30}
      &
      Bound &
      {\color{clr:entropy}Entropic} Bound&
      {\color{clr:polymatroid}Polymatroid} Bound\\
      \hline\hline
      \multirow{4}{*}{\rotatebox{90}{$\begin{array}{c}\text{Full Conjunctive}\\ \text{Query $Q$}\end{array}$}}&
      \begin{tabular}{c}
         \\
         Definition\\
         ~
      \end{tabular}&
      \begin{tabular}{c}
      $\displaystyle{\log|Q| \leq \max_{h \in
         {\color{clr:entropy}\overline\Gamma^*_n} \cap \hdc}h([n])}$\\
         (See~\cite{GLVV,csma})
      \end{tabular}&
      \begin{tabular}{c}
      $\displaystyle{\log|Q| \leq\max_{h \in {\color{clr:polymatroid}\Gamma_n} \cap \hdc}h([n])}$\\
      (See~\cite{GLVV,csma})
      \end{tabular}\\
      \cline{2-4}
      &
      \begin{tabular}{c}
         Cardinality\\
         Constraints ($\cc$)
      \end{tabular}&
      \begin{tabular}{c}
         $\agm$ bound~\cite{DBLP:journals/talg/GroheM14,AGM}\\
         (Tight~\cite{AGM})
      \end{tabular}&
      \begin{tabular}{c}
         $\agm$ bound~\cite{DBLP:journals/talg/GroheM14,AGM}\\
         (Tight~\cite{AGM})
      \end{tabular}\\
      \cline{2-4}
      &
      $\cc$ and $\FD$&
      \begin{tabular}{c}
         {\color{clr:entropy}Entropic} Bound for $\FD$~\cite{GLVV}\\
         ({\color{clr:entropy}Tight}~\cite{szymon-2015})
      \end{tabular}&
      \begin{tabular}{c}
         {\color{clr:polymatroid}Polymatroid} Bound for $\FD$~\cite{GLVV}\\
         \textbf{({\color{clr:polymatroid}Not tight} [Thm.~\ref{thm:glvv:not:tight}])}
      \end{tabular}\\
      \cline{2-4}
      &
      \begin{tabular}{c}
         Degree\\
         Constraints ($\dc$)
      \end{tabular}&
      \begin{tabular}{c}
         {\color{clr:entropy}Entropic} Bound for $\dc$~\cite{csma}\\
         \textbf{({\color{clr:entropy}Tight} [Thm.~\ref{thm:intro:size:bound:disjunctive}, (ii)])}
      \end{tabular}&
      \begin{tabular}{c}
         {\color{clr:polymatroid}Polymatroid} Bound for $\dc$~\cite{csma}\\
         \textbf{({\color{clr:polymatroid}Not tight} [Thm.~\ref{thm:glvv:not:tight}])}
      \end{tabular}\\
      \hline\hline
      \multirow{3}{*}{\rotatebox{90}{$\begin{array}{c}\text{Disjunctive Datalog}\\ \text{Rule $P$}\end{array}$}}&
      \begin{tabular}{c}
         \\
         Definition\\
         ~
      \end{tabular}&
      \begin{tabular}{c}
      $\displaystyle{\log |P(\mv D)| \leq \max_{h \in
         {\color{clr:entropy}\overline\Gamma^*_n} \cap \hdc}\min_{B\in\calB}
         h(B)}$\\
         (\textbf{[Thm.~\ref{thm:intro:size:bound:disjunctive}, (i)]})
      \end{tabular}&
      \begin{tabular}{c}
      $\displaystyle{\log |P(\mv D)| \leq\max_{h \in {\color{clr:polymatroid}\Gamma_n} \cap \hdc}\min_{B\in\calB} h(B)}$\\
         (\textbf{[Thm.~\ref{thm:intro:size:bound:disjunctive}, (i)]})
      \end{tabular}\\
      \cline{2-4}
      &
      \begin{tabular}{c}
         Cardinality\\
         Constraints ($\cc$)
      \end{tabular}&
      \begin{tabular}{c}
         {\color{clr:entropy}Entropic} Bound for\\
         Disjunctive Datalog with $\cc$\\
         \textbf{({\color{clr:entropy}Tight}
         [Thm.~\ref{thm:intro:size:bound:disjunctive}, (ii)])}
      \end{tabular}&
      \begin{tabular}{c}
         {\color{clr:polymatroid}Polymatroid} Bound for\\
         Disjunctive Datalog with $\cc$\\
         \textbf{({\color{clr:polymatroid}Not tight}
         [Thm.~\ref{thm:intro:size:bound:disjunctive}, (iii)])}
      \end{tabular}\\
      \cline{2-4}
            &
      \begin{tabular}{c}
         $\cc$ and $\FD$
      \end{tabular}&
      (same as below)&
      (same as below)\\
      \cline{2-4}
      &
      \begin{tabular}{c}
         Degree\\
         Constraints ($\dc$)
      \end{tabular}&
      \begin{tabular}{c}
         {\color{clr:entropy}Entropic} Bound for\\
         Disjunctive Datalog with $\dc$\\
         \textbf{({\color{clr:entropy}Tight}
         [Thm.~\ref{thm:intro:size:bound:disjunctive}, (ii)])}
      \end{tabular}&
      \begin{tabular}{c}
         {\color{clr:polymatroid}Polymatroid} Bound for\\
         Disjunctive Datalog with $\dc$\\
         \textbf{({\color{clr:polymatroid}Not tight}
         [Thm.~\ref{thm:intro:size:bound:disjunctive}, (iii)}\\
         \hspace{1.7cm}\textbf{or Thm.~\ref{thm:glvv:not:tight}])}
      \end{tabular}\\
      \hline
   \end{tabular}
\caption{Summary of entropic and polymatroid size bounds for full conjunctive queries and for disjunctive datalog rules along with their tightness properties.
The top half of the table depicts bounds for full conjunctive queries while
the bottom half depicts bounds for disjunctive datalog rules.
The ``Definition'' row shows definitions of both the entropic and polymatroid bounds.
The ``Cardinality Constraints ($\cc$)'' row shows the special cases of those definitions when only cardinality constraints are given.
The ``$\cc$ and $\FD$'' row shows the cases when both cardinality constraints and functional dependencies are given.
The ``Degree Constraints ($\dc$)'' row shows the most general cases.
New results due to this work are marked in \textbf{bold}.
}
\label{tab:bounds}
\end{table}

\subsection{Algorithm evaluating disjunctive datalog rules}
\label{subsec:intro:algo-disjunctive}

A {\em worst-case optimal  algorithm} is an algorithm for computing a query in
time within a poly-log factor of a tight worst-case output size bound. Such
algorithms are known for full conjunctive queries under cardinality
constraints~\cite{NPRR,LFTJ,skew,tetris} and FDs~\cite{csma}.
Our next problem is finding a (worst-case optimal?) algorithm for a disjunctive
datalog rule $P$, under degree constraints. More precisely, given an input
database $\mv D$ satisfying given degree constraints, compute a model $\mv T$ in time no
larger than the worst-case bound for $|P(\mv D)|$ under those constraints.

Notice that we allow the algorithm to compute {\em any} model, and not
necessarily a minimal model. This is unavoidable in order to guarantee
the runtime proportional to a good upperbound on $|P(\mv D)|$, as we explain
next.
A conjunctive query $Q$ is a
single-target disjunctive datalog rule $P_Q$. If $Q$ is full, then
from {\em any} model $\mv T$ of $P_Q$ we can answer $Q$ by semijoin-reducing
$\mv T$ with each input relation. Thus, any algorithm evaluating
disjunctive datalog rules can also be used to answer (i.e. compute a minimal
model for) a full conjunctive query. However, this does not hold for non-full
conjunctive queries. For example, if $Q$ is Boolean, then $P_Q$ has a single
target $T_\emptyset()$, and its size is trivially bounded by $|P(\mv D)| \leq
1$; for trivial information theoretic reasons, we simply cannot answer an
arbitrary Boolean conjunctive query in $O(1)$-time.
By allowing the algorithm to compute any model, we can answer a Boolean query
trivially by returning $T_\emptyset = \set{()}$, since this is always a model.
Our third problem is:

\bopm Design an algorithm to compute a model for a given disjunctive datalog
rule, under given degree constraints, with runtime matching the
polymatroid bound~\eqref{eqn:intro:polymatroid:disjunctive:2} above.
\eopm


\noindent
\textbf{Answer 3.}
Details are presented in Sections~\ref{sec:ps} and~\ref{sec:panda}.
We summarize the ideas here.
We present an algorithm
called $\panda$ ({\bf P}roof-{\bf A}ssisted e{\bf N}tropic {\bf D}egree-{\bf A}ware),
which computes a model of a disjunctive datalog rule $P$
within the runtime predicted by the bound~\eqref{eqn:intro:polymatroid:disjunctive:2}.
$\panda$ is derived using a novel principle that we introduced in
\cite{csma}. First, one has to provide ``evidence'', called {\em proof
sequence}, that the polymatroid bound is correct.
Second, each step in the sequence is interpreted as a
relational operator (one of: join, horizontal partition, union),
leading to a model of $P$.

We elaborate a bit more on how a proof sequence arises from the
polymatroid bound~\eqref{eqn:intro:polymatroid:disjunctive:2}.
This bound seems difficult to handle at first glance:
while the feasible region $\Gamma_n \cap \hdc$ is polyhedral,
the objective of~\eqref{eqn:intro:output:bound} is non-linear. We start
by proving in Lemma~\ref{lmm:lambda:1:reformulation}
that it is equivalent to a linear program: there exist constants
$\lambda_B \geq 0$, for $B \in \calB$,
for which:
\begin{equation}
  \max_{h\in \Gamma_n \cap \hdc} \min_{B\in\calB} h(B) =
   \max_{h\in \Gamma_n \cap \hdc} \sum_{B\in\calB} \lambda_B h(B)
\label{eq:intro:polymatroid:bound:lp}
\end{equation}
The right hand side of~\eqref{eq:intro:polymatroid:bound:lp} is simpler to deal with.
In particular, from linear programming duality and Farkas's lemma we show in
Proposition~\ref{prop:sfi:ddl:target} that one can compute non-negative coefficients
$\delta_{Y|X}$ for which the following hold:
\begin{align}
   \max_{h\in \Gamma_n \cap \hdc} \sum_{B\in\calB} \lambda_B h(B)
   &= \sum_{(X,Y,N_{Y|X}) \in \dc}
   \delta_{Y|X} \cdot \log N_{Y|X}, \myskip{ \text{ and}}\\
   \sum_{B\in \calB} \lambda_{B} \cdot h(B)
   &\leq \sum_{(X,Y,N_{Y|X}) \in \dc}
   \delta_{Y|X} \cdot h(Y|X), \myskip \forall h \in \Gamma_n.
\label{eqn:intro:sfi}
\end{align}
Inequality~\eqref{eqn:intro:sfi}, which holds for any polymatroid $h \in
\Gamma_n$, is called a {\em Shannon-flow inequality};
it is a (vast) generalization of Shearer's lemma~\cite{MR859293}.
Note that~\eqref{eqn:intro:sfi} implies
$\log |P| \leq \sum {\delta_{Y|X}} \log N_{Y|X}$.
Thus, the first task is to prove (i.e. provide evidence for) the inequality
(\ref{eqn:intro:sfi}).

A key technical result in the paper is
Theorem~\ref{thm:ps:construction:1} which, stated informally, says
that inequality~\eqref{eqn:intro:sfi} can be proved using a sequence of
rules of one of the following four types, where
$ X \subseteq  Y$:
\begin{align}
   \text{\em Submodularity} &&h( Y| X) \rightarrow h( Y \cup Z| X \cup Z),\label{eq:ps:submod}\\
   \text{\em Monotonicity} &&h( Y) \rightarrow h( X),\label{eq:ps:mono}\\
   \text{\em Composition} &&h( X) + h( Y| X) \rightarrow h( Y),\label{eq:ps:comp}\\
   \text{\em Decomposition} &&h( Y) \rightarrow h( X) + h( Y| X)\label{eq:ps:decomp}.
\end{align}
We think of the above rules as rewrite rules that transform the terms on the LHS into the
terms on the RHS. Moreover for each one of the above rules, the RHS is guaranteed to be
smaller than or equal to the LHS for all $h \in \Gamma_n$.
To explain the theorem, assume for the sake of discussion that all
coefficients in (\ref{eqn:intro:sfi}) are natural numbers.  Then both
sides of~\eqref{eqn:intro:sfi} can be seen as bags of terms, and the theorem says that there
exists a sequence of rewritings using the four rules above,
transforming the term bag on the RHS of~\eqref{eqn:intro:sfi} to the term bag on the LHS.  Obviously, if such a
sequence exists, then inequality~\eqref{eqn:intro:sfi} holds, because each
rewriting replaces a term (or sum of two terms) with a smaller or equal term (or
sum of two terms). The converse statement is non-obvious.   For example in
our prior work~\cite{csma} we found that, without the decomposition
rule~\eqref{eq:ps:decomp}, the remaining three rules along with the additional submodularity rule $h(A)+h(B)\rightarrow h(A\cup B)+h(A\cap B)$ are {\em not} a complete proof system: there exists
a Shannon-flow inequality without a proof sequence consisting only of those rules.

Finally, $\panda$ consists of interpreting each step in the proof
sequence as a relational operation, leading to:

\bthm $\panda$ computes a model of a disjunctive datalog rule $P$
under degree constraints $\dc$ in time
\footnote{In this paper, the big-$O$ notation is in data-complexity, hiding a factor
   that is query-dependent and data-independent.
   The big-$\tilde O$ additionally hides a single $\log$-factor in data-complexity.}
$$\tilde O\left(N + \poly(\log N)\cdot \prod_{(X,Y,N_{Y|X})\in\dc}N_{
  Y| X}^{\delta_{ Y| X}}\right),$$
where $\displaystyle{\sum_{(X,Y,N_{Y|X})\in\dc} \delta_{Y|X}\log N_{Y|X} =
\outputbound_{\Gamma_n\cap \hdc}(P)}$.
\label{thm:main:panda:disjunctive}
\ethm

We now illustrate how the
inequality in Example~\ref{ex:intro:2b} can be proved using the above set of
complete rules, and then use the proof to compute the query $P$ from
Example~\ref{ex:intro:2} in time $O(N^{3/2})$.

\begin{ex} \label{ex:intro:panda} Consider the disjunctive rule $P$ in
   Example~\ref{ex:intro:2}, which we repeat here:
   \begin{align*}
      P:\ \ \ & T_{123}(A_1,A_2,A_3)\vee T_{234}(A_2,A_3,A_4) \cd
      R_{12}(A_1,A_2),R_{23}(A_2,A_3),R_{34}(A_3,A_4).
   \end{align*}
   Assume all input relations have cardinality $\leq N$.  We
   illustrate $\panda$ by showing how to compute a model for this rule in time
   $O(N^{3/2})$.  To do this, the first step for us is to
   provide a proof sequence, showing that the output size of $P$ satisfies $|P| \leq N^{3/2}$.  While we have done this
   already in Example~\ref{ex:intro:2b}, we provide here an
   alternative proof sequence, using only the proof steps
   \eqref{eq:ps:submod}\ldots\eqref{eq:ps:decomp}.  As before we start by noticing
   $\log |P| \leq \min(h(A_1A_2A_3),h(A_2A_3A_4)) \leq
   \frac{1}{2}(h(A_1A_2A_3)+h(A_2A_3A_4))$,
   then we prove the Shannon-flow inequality:
   \[
      \frac 1 2 \bigl( h(A_1A_2A_3)+h(A_2A_3A_4) \bigr) \leq
      \frac 1 2 \bigl( h(A_1A_2)+h(A_2A_3)+h(A_3A_4) \bigr)
   \]
   using a proof sequence consisting of steps~\eqref{eq:ps:submod}\ldots\eqref{eq:ps:decomp}:
   \begin{eqnarray*}
      h(A_1A_2)+h(A_2A_3)+h(A_3A_4) &\stackrel{(1)}{\rightarrow}& \\
      h(A_1A_2|A_3) + h(A_2A_3) + h(A_3A_4) &\stackrel{(2)}{\rightarrow}& \\
      \left[h(A_1A_2|A_3) + h(A_3)\right] + \left[h(A_2A_3) + h(A_4|A_3)\right] &\stackrel{(3)}{\rightarrow}& \\
      \left[h(A_1A_2|A_3) + h(A_3)\right] + \left[h(A_2A_3) + h(A_4|A_2A_3)\right] &\stackrel{(4)}{\rightarrow}& \\
      h(A_1A_2A_3) + h(A_2A_3A_4) & &
   \end{eqnarray*}
(And since each one of the terms $h(A_1A_2)$, $h(A_2A_3)$ and $h(A_3A_4)$ is bounded by $\log N$, this proves that $\log |P| \leq 3/2 \log N$.)
   $\panda$ associates each term in the above proof sequence to some relation,
   and interprets each proof step as a
   relational operator on the relations associated to the terms in that proof step.
   (See Figure~\ref{fig:panda1}.)
   Initially, the terms $h(A_1A_2)$, $h(A_2A_3)$ and $h(A_3A_4)$ are associated with the input relations $R_{12}$, $R_{23}$ and $R_{34}$ respectively.
   Step (1) is a submodularity step: $\panda$ does nothing but keeps track that the
   new term $h(A_1A_2|A_3)$ is associated with the relation $R_{12}$.
   Step (2) is a
   decomposition step:  $\panda$ partitions $R_{34}(A_3,A_4)$
   horizontally into $R_3'(A_3)$ and $R_{34}'(A_3,A_4)$, where $R_3'$
   contains all values $a_3$ that are ``heavy hitters'', meaning that
   $|\sigma_{A_3=a_3}(R_{34})| \geq \sqrt N$,
   and $R_{34}'$ consists of all pairs $(a_3,a_4)$ with $a_3$ being
   ``light hitters''. $R_3'$ becomes associated with the term $h(A_3)$
   while $R_{34}'$ becomes associated with $h(A_4|A_3)$.
   Step (3) is another submodularity where $\panda$ associates the new term $h(A_4|A_2A_3)$ with the relation $R_{34}'$.
   Step (4) are two compositions,
   interpreted as joins.  $\panda$ computes the first target
   $T_{123}(A_1,A_2,A_3)= R_{12}(A_1,A_2) \Join R_3'(A_3)$, and the
   second target
   $T_{234}(A_2,A_3,A_4) = R_{23}(A_2,A_3) \Join R_{34}'(A_3,A_4)$.
   Both joins take time $O(N^{3/2})$, because
   $|R_3'|\leq |R_{34}|/ \sqrt N \leq \sqrt N$ and $\Deg_{R_{34}'}(A_3A_4 | A_3) < \sqrt N$.
   (Note that without the horizontal partitioning into heavy and light hitters in step (3), both $|R_3'|$ and $\Deg_{R_{34}'}(A_3A_4 | A_3)$ could have been as large as $N$, hence both joins could have taken up to $N^2$ time, which would have exceeded our budget of $N^{3/2}$.)
\end{ex}

\begin{figure}[ht!]
\tikzstyle{panda-fig:node} = [draw, ellipse, inner sep = .05cm]
\tikzstyle{panda-fig:edge} = [->]
\centering{
\begin{tikzpicture}[scale=.75, every node/.style={transform shape}]
\colorlet{clr-t123}{red!50!blue}
\colorlet{clr-t234}{DarkGreen!50!blue}
\begin{scope}[shift={(0,0)}]
   \node[panda-fig:node] at(5,0) (n1) {${\color{red}h(A_1A_2)}+{\color{DarkGreen}h(A_2A_3)}+{\color{blue}h(A_3A_4)}$};
   \node[panda-fig:node] at(5,-2) (n11) {${\color{red}h(A_1A_2|A_3)}+{\color{DarkGreen}h(A_2A_3)}+{\color{blue}h(A_3A_4)}$};
   \node[panda-fig:node] at(2.5,-4) (n111) {${\color{red}h(A_1A_2|A_3)}+{\color{blue}h(A_3)}$};
   \node[panda-fig:node] at(7.5,-4) (n112) {${\color{DarkGreen}h(A_2A_3)}+{\color{blue}h(A_4|A_3)}$};
   \node[panda-fig:node] at(2.5,-6) (n1111) {$\color{clr-t123}h(A_1A_2A_3)$};
   \node[panda-fig:node] at(7.5,-6) (n1121) {${\color{DarkGreen}h(A_2A_3)}+{\color{blue}h(A_4|A_2A_3)}$};
   \node[panda-fig:node] at(7.5,-8) (n11211) {$\color{clr-t234}h(A_2A_3A_4)$};
   \draw[panda-fig:edge,red] (n1)--(n11) node[midway,left] {\color{red}submodularity};
   \draw[panda-fig:edge,blue] (n11)--(n111) node[midway,right=.1cm] {\color{blue}decomposition};
   \draw[panda-fig:edge,blue] (n11)--(n112);
   
   \draw[panda-fig:edge,color={clr-t123}] (n111)--(n1111) node[midway, left] {\color{clr-t123}composition};
   \draw[panda-fig:edge,blue] (n112)--(n1121) node[midway, right] {\color{blue}submodularity};
   \draw[panda-fig:edge,color={clr-t234}] (n1121)--(n11211)node[midway, right] {\color{clr-t234}composition};
   \node at (5,-9) {\Large Proof steps};
\end{scope}

\draw[gray] (10,.5)--(10,-8.5);

\begin{scope}[shift={(9.7cm,0)}]
   \node[panda-fig:node] at(5,0) (n1) {${\color{red}R_{12}(A_1,A_2)}, {\color{DarkGreen}R_{23}(A_2,A_3)}, {\color{blue}R_{34}(A_3,A_4)}$};
   \node[panda-fig:node] at(5,-2) (n11) {${\color{red}R_{12}(A_1,A_2)}, {\color{DarkGreen}R_{23}(A_2,A_3)}, {\color{blue}R_{34}(A_3,A_4)}$};
   \node[panda-fig:node] at(2.5,-4) (n111) {${\color{red}R_{12}(A_1,A_2)}, {\color{blue}R'_3(A_3)}$};
   \node[panda-fig:node] at(7.5,-4) (n112) {${\color{DarkGreen}R_{23}(A_2,A_3)}, {\color{blue}R_{34}'(A_3,A_4)}$};
   \node[panda-fig:node] at(2.5,-6) (n1111) {$\color{clr-t123}T_{123}(A_1,A_2,A_3)$};
   \node[panda-fig:node] at(7.5,-6) (n1121) {${\color{DarkGreen}R_{23}(A_2,A_3)}, {\color{blue}R_{34}'(A_3,A_4)}$};
   \node[panda-fig:node] at(7.5,-8) (n11211) {$\color{clr-t234}T_{234}(A_2,A_3,A_4)$};
   
   \draw[panda-fig:edge,red] (n1)--(n11) node[midway,left] {\color{red}no OP};
   \draw[panda-fig:edge,blue] (n11)--(n111) node[midway,right=.3cm, align=center] {\color{blue}data\\partitioning};
   \draw[panda-fig:edge,blue] (n11)--(n112);
   \draw[panda-fig:edge,color={clr-t123}] (n111)--(n1111) node[midway, left] {\color{clr-t123}\large $\Join$};
   \draw[panda-fig:edge,blue] (n112)--(n1121) node[midway, right] {\color{blue}no OP};
   \draw[panda-fig:edge,color={clr-t234}] (n1121)--(n11211)node[midway, right] {\color{clr-t234}\large $\Join$};
   \node at (5,-9) {\Large Algorithmic steps};
\end{scope}
\end{tikzpicture}
}
\caption{Illustration of the proof sequence and the corresponding $\panda$ algorithm from Example~\ref{ex:intro:panda}.}
\Description{Illustration of the proof sequence and the corresponding $\panda$ algorithm from Example~\ref{ex:intro:panda}.}
\label{fig:panda1}
\end{figure}
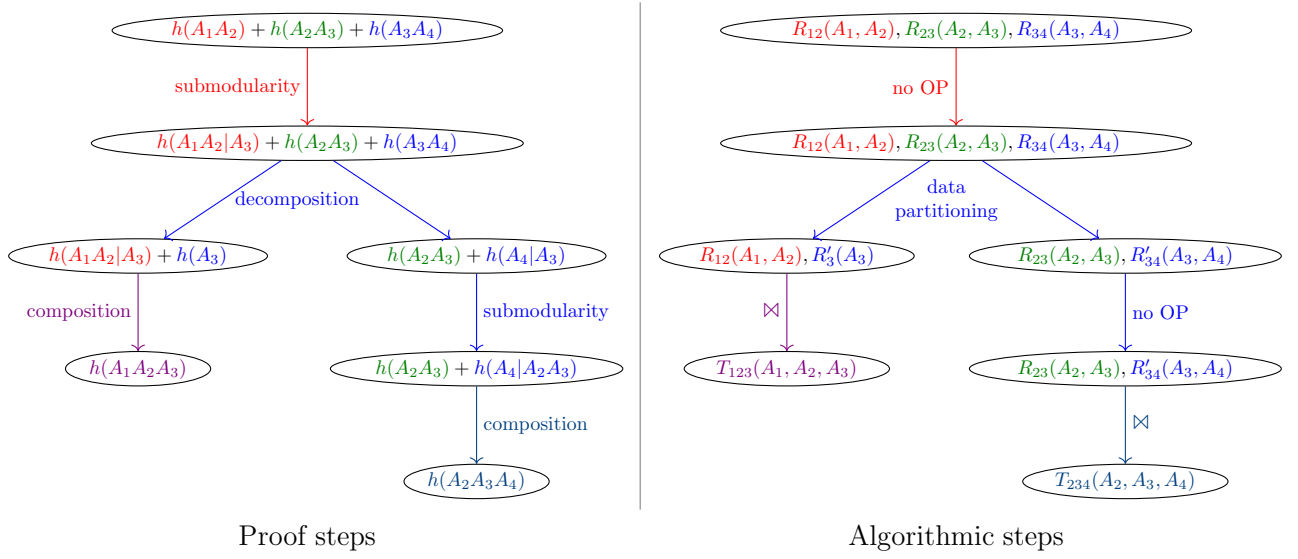

Example~\ref{ex:intro:panda} has the nice property that the two terms $h(A_4|A_3)$ and
$h(A_3)$ resulting from the decomposition step (2) {\em diverged},
i.e. were used in different targets.  This allowed $\panda$ to place
each tuple from $R_{34}$ in either $R_3'$ or $R_{34}'$: no need to
place in both, since these relations are not joined later.  However,
we could neither prove nor disprove the divergence property in
general. Instead, $\panda$ in general conservatively places each tuple in {\em
   both} relations, yet it must ensure
$|R'_3(A_3)| \cdot \Deg_{R'_{34}}(A_3A_4 | A_3) \leq |R_{34}|$.  For that
it creates $\log N$ bins, with tuples whose degree is in
$[2^i, 2^{i+1})$, for $i=0,\ldots,\lfloor \log N \rfloor$, and processes
each bin separately.  This needs to be repeated at each non-divergent
decomposition step, hence the additional $\poly(\log N)$ factor in the
runtime.

\subsection{Towards Optimal Algorithms for Conjunctive Queries}
\label{subsec:intro:algo-conjunctive}
What is an optimal runtime to compute a given conjunctive query?
A common belief is that its cost is of the form
$\tilde O(N^d + \outputsize)$, where $N$ is the size of the input
database, $N^d$ represents the ``intrinsic'' cost of the query,
and $\outputsize$ is the unavoidable cost of reporting the output.
Worst-case optimal algorithms are not optimal in this sense.
They are only good for inputs whose intrinsic cost is about the same as the
worst-case output size.
As described below, there are algorithms whose runtimes are more output-sensitive.
If the query is Boolean, then
the output size is always $1$, and the cost is totally dominated by
the intrinsic cost of the query; for simplicity we discuss here only
Boolean queries, but our discussion extends to other conjunctive and aggregate
queries~\cite{faq} (see Section~\ref{sec:conclusion}).
Thus, an optimal algorithm should compute a
Boolean query in time $\tilde O(N^d)$, with an exponent $d$ as small
as possible. Generalizing to degree constraints, it should
compute the query in time
$\tilde O(\prod N_{Y|X}^{\delta_{Y|X}})$, where
$N_{Y|X}$ are the degree bounds, and the product is minimized.

In search of a yardstick for optimality, we borrow from the long
history of research on fixed-parameter tractability.  To a recursively
enumerable class
$\calC$ of Boolean conjunctive queries (equivalently, hypergraphs of $\csp$
problems) we associate the following decision problem, denoted by
$\bcq(\calC)$: given a query $Q\in \calC$ and an instance $\mv D$,
check if the query is true on the instance.  $\bcq(\calC)$ is {\em
polynomial-time solvable} if there exists an algorithm that runs in
polynomial time in $|Q|, |\mv D|$ (combined
complexity~\cite{DBLP:conf/stoc/Vardi82}).  The problem is {\em
fixed-parameter tractable} ($\fpt$, with parameter $|Q|$, the
query's size) if there is an algorithm solving every
$\calC$-instance in time $f(|Q|) \cdot |\mv D|^d$ for some fixed
constant $d$, where $f$ is any computable function.
It was known very early on that if $\calC$ is the class of
bounded tree-width then $\bcq(\calC)$ is in
$\ptime$~\cite{DBLP:conf/aaai/Freuder90}, and this was extended to
query width~\cite{DBLP:journals/tcs/ChekuriR00}, then (generalized)
hypertree width~\cite{DBLP:journals/jcss/GottlobLS02}, and fractional
hypertree width~\cite{DBLP:journals/talg/GroheM14}: boundedness of any
of these parameters implies $\bcq(\calC)$ is in $\ptime$.

Grohe's now classic result~\cite{DBLP:conf/focs/Grohe03} states that,
assuming $\textsf{W}[1] \neq \fpt$,
if the arities of the relations are bounded, then the converse also holds:
$\bcq(\calC) \in \fpt$, $\bcq(\calC) \in \ptime$,
and queries in $\calC$ having bounded width are three equivalent statements, for
any width notion above. In a beautiful paper, Marx~\cite{MR3144912} extended
this result to unbounded arity case, showing that $\bcq(\calC)$ is $\fpt$ iff
every $Q \in\calC$ has bounded submodular width, denoted by $\subw(Q)$.
His results suggest to us using
the submodular width as a yardstick for optimality. In
order to prove that bounded submodular width implies $\fpt$-membership, Marx
described a query evaluation algorithm that runs in time
$O(\poly(N^{\subw(Q)}))$.\footnote{It is not clear what the exact
  runtime of Marx's algorithm is. His theorem states that it is
  $O(\poly(N^{\subw(Q)}))$. Our best interpretation of Lemma 4.3 and
  Lemma 4.5 of~\cite{MR3144912} is that Marx's algorithm runs in time
  at least $O(N^{2\cdot\subw(Q)})$.}  We define an algorithm to be
{\em optimal} if its runtime is $\tilde O(N^{\subw(Q)})$.  While no
lower bounds are known to date to rule out faster algorithms for a
specific query, Marx's dichotomy theorem ruled out faster algorithms
for any recursively enumerable class of queries (see~\cite{MR3144912}
and Section~\ref{sec:background}).  Our fourth problem is:

\bopm Design an algorithm evaluating a Boolean conjunctive query $Q$
in $\tilde O(N^{\subw(Q)})$-time. Extend the notion of submodular width, and
the algorithm, to handle arbitrary degree constraints, to arbitrary conjunctive
and aggregate queries.
\eopm

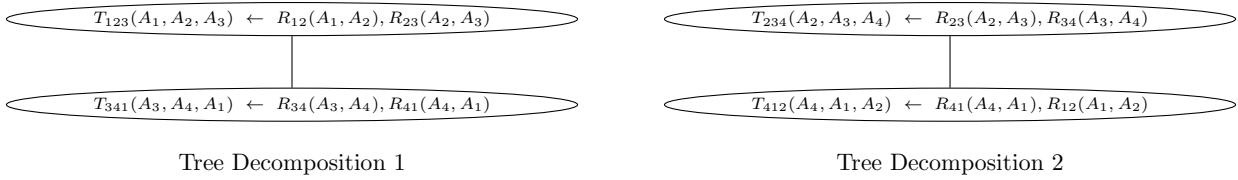
\begin{figure}
\begin{tikzpicture}
\begin{scope}[color=black]
\node[draw,ellipse] (123) {\scriptsize $T_{123}(A_1,A_2,A_3) \cd R_{12}(A_1, A_2), R_{23}(A_2, A_3)$};
\node[draw,ellipse, below = (.3in) and of 123] (134) {\scriptsize $T_{341}(A_3,A_4,A_1) \cd R_{34}(A_3, A_4), R_{41}(A_4, A_1)$};
\draw (123) -- (134) node (123134) {};
\node[below = (.3in) of 123134, anchor=center] {Tree Decomposition 1};
\end{scope}
\begin{scope}[color=black]
\node[draw,ellipse, right = .25in of 123] (124) {\scriptsize $T_{234}(A_2,A_3,A_4) \cd R_{23}(A_2, A_3), R_{34}(A_3, A_4)$};
\node[draw,ellipse, below = (.3in) and of 124] (234) {\scriptsize $T_{412}(A_4,A_1,A_2) \cd R_{41}(A_4, A_1), R_{12}(A_1, A_2)$};
\draw (124) -- (234) node (124234) {};
\node[below = (.3in) of 124234, anchor=center] {Tree Decomposition 2};
\end{scope}
\end{tikzpicture}
\caption{Two tree decompositions for the query in
  Example~\ref{ex:intro:1}.  Normally, each tree node is labeled with
  a set of variables, e.g. $\chi(t) = \set{A_1,A_2,A_3}$; for
  convenience we also show the atoms contained in those variables,
  i.e. $R_{12}(A_1,A_2),R_{23}(A_2,A_3)$, and also give a name to the
  associated full conjunctive query, e.g. $T_{123}(A_1,A_2,A_3)$.}
\Description{Two tree decompositions for the query in
Example~\ref{ex:intro:1}.}
  \label{fig:tree:decomposition}
\end{figure}

We briefly review the notion of submodular width, and its
relationship to other width parameters. Note that all known
width parameters considered only cardinality constraints.
A polymatroid $h$
is {\em edge-dominated} if $h(F) \leq 1, \forall F \in \calE$.
Edge domination is a normalized version of cardinality
constraints.
The submodular width is defined to be
$\subw(Q) \defeq \max_h \min_{(T, \chi)} \max_{t\in V(T)}
h(\chi(t)),$
where $h$ ranges over edge-dominated polymatroids, and $(T,\chi)$ over
tree decompositions of $Q$ (see Definition~\ref{defn:TD}).  Prior width
parameters such as tree-width~\cite{DBLP:conf/aaai/Freuder90},
generalized-~\cite{DBLP:journals/jcss/GottlobLS02} and fractional-
hypertree width~\cite{DBLP:journals/talg/GroheM14}
are defined by first defining the width of a tree decomposition, then choosing the
decomposition that minimizes this width
(see~\cite{DBLP:conf/pods/GottlobGLS16} for a nice survey).
Thus, there is always a {\em
best} tree decomposition $(T,\chi)$, and a query evaluation algorithm running
on that $(T,\chi)$; e.g., the fractional hypertree
width is $\fhtw(Q) \defeq \min_{(T, \chi)} \max_{t\in V(T)} \rho^*(\chi(t)),$
where $\rho^*$ is the fractional edge cover number of the set
$\chi(t)$.
In the submodular width, we are allowed to choose the tree decomposition $(T,
\chi)$ {\em after} we see the polymatroid $h$.
Marx showed that $\subw(Q) \leq \fhtw(Q)$, for all $Q$,
and there are classes of queries for which the gap is unbounded
(see also Example~\ref{ex:gap1}).

\noindent
\textbf{Answer 4.}  Our answer is presented in Section~\ref{sec:it}.
Briefly, we generalize the notions of fractional hypertree width $\fhtw(Q)$
and submodular width $\subw(Q)$ to account for arbitrary degree constraints, and
call them
{\em degree-aware fractional hypertree width / submodular width},
denoted by $\dafhtw(Q)$ and $\dasubw(Q)$, respectively.  In fact, we
describe a very general framework that captures virtually all previously
defined width-parameters, {\em and} extends them to degree constraints.
We then show how to use $\panda$ to compute a query in time whose
exponent is $\dasubw(Q)$ (which is bounded by $\dafhtw(Q)$), using the same
earlier principle: from a
proof of the bound of $\dasubw(Q)$, we derive an algorithm that
computes $Q$ in that bound.

\bthm $\panda$ computes any full or Boolean conjunctive query
$Q$ in time
$\tilde O(N + \poly(\log N)\cdot 2^{\dasubw(Q)}+\outputsize)$.
\label{thm:main:panda:submodular}
\ethm

We have chosen to present only results on the full and Boolean conjunctive query
cases in order to keep the paper more accessible and focused. Our results extend
straightforwardly to proper conjunctive queries and to aggregate queries
(in the sense of $\faq$-queries over one semiring~\cite{faq,AM00}) as well.
Section~\ref{sec:conclusion} briefly describes how these more general queries are
handled with $\panda$.

\begin{ex} \label{ex:intro:subw} Consider the Boolean variant of the
  4-cycle query $Q$ in Example~\ref{ex:intro:1}.  More precisely we
  ask the question: {\em does there exists a cycle of length 4?}; Alon,
  Yuster, and Zwick~\cite{AYZ97} described an algorithm solving this
  problem in time $O(N^{3/2})$, where $N$ is the number of edges.
  However, every traditional tree-decomposition-based evaluation
  algorithm takes time $N^2$.  Indeed, the query has only two
  non-trivial tree decompositions, shown in
  Figure~\ref{fig:tree:decomposition} and, for each tree, there exists a
  worst-case input on which the intermediate tables for that tree have
  size $N^2$. For example, given the instance
  $R_{12} = R_{34} = [N] \times [1]$,
  $R_{23}=R_{41} = [1] \times [N]$, both intermediate tables of the
  tree on the left have size $N^2$, and a similar worst-case instance
  exists for the tree on the right.  In fact, the fractional hypertree
  width of $Q$ is $\fhtw(Q) = 2$, because both trees have $\fhtw=2$.

  In contrast to the fractional hypertree width, the submodular width
  is adaptive: it chooses the tree based on $h$.  More precisely,
  $\subw(Q)$ is the maximum over edge-dominated polymatroids $h$ of
  the quantity
  \begin{equation}
      \min\bigl(\max(h(A_1A_2A_3),h(A_3A_4A_1)),\max(h(A_2A_3A_4),h(A_4A_1A_2))\bigr)
      \label{eq:4-cycle-subw}
  \end{equation}
  (where ``edge-dominated'' in this example means $h(A_1A_2), h(A_2A_3), h(A_3A_4)$, and $h(A_4A_1)$ are all $\leq 1$).
  Intuitively, $\max(h(A_1A_2A_3),h(A_3A_4A_1))$ is the complexity of the tree on
  the left, and the other $\max$ is the complexity of the tree on the right.
  $\panda$ starts by proving $\subw(Q) \leq 3/2$.  To do
  that, it applies the distributivity law of $\min$ over $\max$ on
  \eqref{eq:4-cycle-subw}:
  \begin{align}
\mbox{Eq.~\eqref{eq:4-cycle-subw}} =
     \max\bigl(  \nonumber
    & \min(h(A_1A_2A_3),h(A_2A_3A_4)),
      \min(h(A_1A_2A_3),h(A_4A_1A_2)), \nonumber \\
    & \min(h(A_3A_4A_1),h(A_2A_3A_4)),
      \min(h(A_3A_4A_1),h(A_4A_1A_2))
      \bigr) \label{eq:min:max:dist}
  \end{align}
then it proves one inequality for each term under $\max$:
\begin{eqnarray}
   \min(h(A_1A_2A_3),h(A_2A_3A_4)) &\leq & 1/2\bigl(h(A_1A_2)+h(A_2A_3)+h(A_3A_4)\bigr),\label{eq:4-cycle-subw-1}\\
   \min(h(A_1A_2A_3),h(A_4A_1A_2)) &\leq & 1/2\bigl(h(A_4A_1)+h(A_1A_2)+h(A_2A_3)\bigr),\label{eq:4-cycle-subw-2}\\
   \min(h(A_3A_4A_1),h(A_2A_3A_4)) &\leq & 1/2\bigl(h(A_2A_3)+h(A_3A_4)+h(A_4A_1)\bigr),\label{eq:4-cycle-subw-3}\\
   \min(h(A_3A_4A_1),h(A_4A_1A_2)) &\leq & 1/2\bigl(h(A_3A_4)+h(A_4A_1)+h(A_1A_2)\bigr).\label{eq:4-cycle-subw-4}
\end{eqnarray}
Example~\ref{ex:intro:2b} showed the first inequality, the other three
are similar.  Since $h$ is edge-dominated, every RHS is $\leq 3/2$,
implying that Eq.~\eqref{eq:4-cycle-subw} is $\leq 3/2$, which proves
the claim that $\subw(Q) \leq 3/2$.
%
$\panda$ computes the query in time $\tilde O(N^{\subw(Q)})$ as follows.
First, it interprets each inequality above as the output size bound of a
disjunctive datalog rule:
\begin{eqnarray*}
   P_1:\ \ \ & T_{123}(A_1,A_2,A_3)\vee T_{234}(A_2,A_3,A_4) \cd
   R_{12}(A_1,A_2),R_{23}(A_2,A_3),R_{34}(A_3,A_4).\\
   P_2:\ \ \ & T_{123}(A_1,A_2,A_3)\vee T_{412}(A_4,A_1,A_2) \cd
   R_{41}(A_4,A_1),R_{12}(A_1,A_2),R_{23}(A_2,A_3).\\
   P_3:\ \ \ & T_{341}(A_3,A_4,A_1)\vee T_{234}(A_2,A_3,A_4) \cd
   R_{23}(A_2,A_3),R_{34}(A_3,A_4),R_{41}(A_4,A_1).\\
   P_4:\ \ \ & T_{341}(A_3,A_4,A_1)\vee T_{412}(A_4,A_1,A_2) \cd
   R_{34}(A_3,A_4),R_{41}(A_4,A_1),R_{12}(A_1,A_2).
\end{eqnarray*}
Next, $\panda$ evaluates each rule using the algorithm mentioned
earlier in Section~\ref{subsec:intro:algo-disjunctive}; the first rule
above was shown in Example~\ref{ex:intro:2} and the corresponding algorithm with runtime $N^{3/2}$ was explained in Example~\ref{ex:intro:panda}.
Each one of the rules $P_1$ and $P_2$ computes a different table $T_{123}$, and the union of both is taken as a single intermediate table $T_{123}$.
Similarly, the four rules result in the intermediate tables $T_{341}, T_{234}, T_{412}$,
which -along with $T_{123}$-
correspond to all tree nodes in Figure~\ref{fig:tree:decomposition}.
The runtime so far is $N^{3/2}$.
Now, $\panda$ semi-join reduces each one of those intermediate tables with all input relations,
i.e.~it semi-join reduces $T_{123}(A_1,A_2,A_3)$ with $R_{12}(A_1,A_2)$ and with $R_{23}(A_2,A_3)$, and so on: This step is needed to remove spurious tuples.
Finally, $\panda$ computes separately each one of the two trees from Figure~\ref{fig:tree:decomposition} (viewed as an acyclic
Boolean query and using Yannakakis
algorithm~\cite{DBLP:conf/vldb/Yannakakis81}), in time $N^{3/2}$
(since each one of the intermediate tables has size at most $N^{3/2}$), then
returns the logical $\bigvee$ of the two results.

One way to get the intuition behind the four rules above is to combine them
into a single rule, where the head is the conjunction of the four
heads, then apply the (reverse) distributivity law to the head
\begin{multline}
  (T_{123} \vee T_{234}) \wedge (T_{123} \vee T_{412}) \wedge (T_{341} \vee T_{234}) \wedge (T_{341} \vee T_{412}) = \\
  \bigl(T_{123}(A_1,A_2,A_3) \wedge T_{341}(A_3, A_4, A_1)\bigr) \vee
  \bigl(T_{234}(A_2,A_3,A_4) \wedge T_{412}(A_4,A_1,A_2)\bigr)
  \cd\\
  R_{12}(A_1,A_2),R_{23}(A_2,A_3),R_{34}(A_3,A_4),R_{41}(A_4,A_1).
   \label{eqn:4:cycle:distributivity}
\end{multline}
This captures the intuition of the fact that $\panda$ places each four-cycle $(a_1,a_2,a_3,a_4)$ either in both intermediate tables of the first tree
$T_{123}(A_1,A_2,A_3) \wedge T_{341}(A_3, A_4, A_1)$, or in both intermediate tables of the second
tree $T_{234}(A_2,A_3,A_4) \wedge T_{412}(A_4,A_1,A_2)$.

As we explained, no single tree is sufficient to compute $Q$ in time
$O(N^{3/2})$.  Instead, each of the four rules $P_1,\dots,P_4$ directs
tuples to either one tree or the other. For example, on the worst-case
instance above (where $R_{12} = R_{34} = [N] \times [1]$,
$R_{23}=R_{41} = [1] \times [N]$), at most $N^{3/2}$ of the $N^2$ four-cycles are inserted by
$P_1$ into $T_{123}$ (in the left tree), the others are spilled over
to $T_{234}$ (in the right tree).  However, we prove that the two
trees {\em together} are sufficient to compute $Q$,
%
%
%
by showing that every four-cycle $\mv a = (a_1,a_2,a_3,a_4)$ is
inserted in both nodes of either tree, thus, by taking the $\bigvee$ of the two trees,
$\panda$ computes the query correctly.  Indeed, suppose
otherwise: Since $\mv a$ is missing from the left tree, assume
w.l.o.g. that $(a_1,a_2,a_3)$ is missing from node $T_{123}$; similarly,
since $\mv a$ is missing from the right tree, assume that $(a_2,a_3,a_4)$ is missing from
$T_{234}$.
%
But that implies that $T_{123} \vee T_{234}$ is not a model of the
disjunctive datalog rule $P_1$ above, which is a contradiction.  Thus,
each tuple $\mv a$ is fully included in some tree.
\end{ex}

\section{Background and Related Work}
\label{sec:background}

Throughout the paper, we use the following convention.  The
non-negative reals, rationals, and integers are denoted by
$\R_+,\Q_+$, and $\N$ respectively. For a positive integer $n$, $[n]$ denotes
the set $\{1,\dots,n\}$.

The function $\log$ without a base specified is
base-$2$, i.e. $\log = \log_2$. Uppercase $A_i$ denotes a
variable/attribute, and lowercase $a_i$ denotes a value in the
discrete domain $\Dom(A_i)$ of the variable.  For any subset
$S\subseteq [n]$, define $\mv A_S = (A_i)_{i\in S}$,
$\mv a_S = (a_i)_{i\in S} \in \prod_{i\in S}\Dom(A_i)$.  In
particular, $\mv A_S$ is a tuple of variables and $\mv a_S$ is a tuple
of specific values with support $S$.  Occasionally we use $\mv t_S$ to
denote a tuple with support $S$. For any two finite sets $S$ and $T$,
let $S^T$ denote the collection of all maps $f : T \to S$. Such a map
$f$ is also viewed as a vector whose coordinates are indexed by
members of $T$ and whose coordinate values are members of
$S$.\footnote{This is standard combinatorics notation. There are
  $|S|^{|T|}$ of those maps.}

\bdefn
Let $n$ be a positive integer.
A function $f : 2^{\calV} \to \R_+$ is called a (non-negative)
{\em set function} on $\calV$.
A set function $f$ on $\calV$ is {\em modular} if
$f(S) = \sum_{v\in S} f(\{v\})$ for all $S\subseteq \calV$,
is {\em monotone} if $f(X) \leq f(Y)$ whenever $X \subseteq Y$,
is {\em subadditive} if $f(X\cup Y)\leq f(X)+f(Y)$
for all $X,Y\subseteq \calV$,
and is {\em submodular} if $f(X\cup Y)+f(X\cap Y)\leq f(X)+f(Y)$
for all $X,Y\subseteq \calV$.
A function $h : 2^{\mv A_{[n]}} \to \R_+$ is said to be {\em entropic} if there
is a joint distribution on $\mv A_{[n]}$ such that $h(\mv A_S)$ is the marginal
entropy on $\mv A_S$, $S\subseteq [n]$. We also write $h(S)$ for $h(\mv A_S)$,
and thus entropic functions are also set functions.
\edefn

Unless specified otherwise,
we will only consider {\em non-negative} and {\em monotone} set functions
$f$ for which $f(\emptyset) = 0$; this assumption will be implicit in the
entire paper.

\bdefn
\label{def:set-functions}
Let $\Mod_n$, $\sa_n$, and $\Gamma_n$ denote the set of all (non-negative and
monotone) modular, subadditive, and submodular set functions on
$\calV$, respectively.
Let $\Gamma^*_n$ denote the set of all entropic functions on $n$ variables, and
$\overline\Gamma^*_n$ denote its topological closure.\footnote{Note that any non-negative modular function is
monotone.  The notations $\Gamma_n, \Gamma^*_n,\overline \Gamma^*_n$ are
standard in information theory~\cite{Yeung:2008:ITN:1457455}.}
\edefn

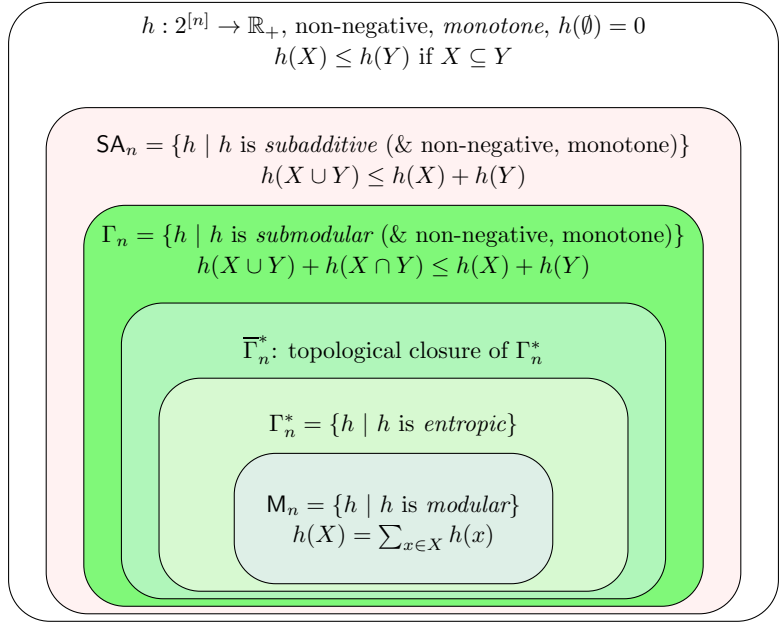
\begin{figure*}[t!]
\centering    \begin{tikzpicture}
      \node[fill opacity=0.5, rounded corners=.55cm, draw, fit={(-1,0) (9,8)}] (all) { };
      \node[] at (4, 7.6) {\shortstack{$h : 2^{[n]} \to \R_+$, non-negative,
      {\em monotone}, $h(\emptyset) =0$\\ $h(X)\leq h(Y)$ if $X \subseteq Y$}};
      \node[fill opacity=0.5, rounded corners=.55cm, draw, fit={(-.5,0.1)
      (8.5,6.6)}, fill=red!10] (sa) { };
      \node[] at (4, 6.0) {\shortstack{ $\sa_n = \{ h \suchthat h \text{ is
      {\em subadditive} (\& non-negative, monotone)} \}$ \\ $h(X\cup Y) \leq
      h(X)+h(Y)$}};
      \node[fill opacity=0.5, rounded corners=.55cm, draw, fit={(0,0.2)
      (8,5.3)}, fill=green] (gamma-n) { };
      \node[] at (4, 4.8) {\shortstack{$\Gamma_n = \{ h
      \suchthat h \text{ is {\em submodular} (\& non-negative, monotone)} \}$\\ $h(X\cup Y)
      +h(X\cap Y) \leq h(X)+h(Y)$}};
      \node[fill opacity=0.5, rounded corners=.55cm, draw, fit={(0.5,0.3)
      (7.5,4)}, fill=cyan!10] (bar-gamma-star) { };
      \node[] at (4, 3.5) {$\overline\Gamma^*_n$: topological closure of $\Gamma^*_n$};
      \node[fill opacity=0.5, rounded corners=.55cm, draw, fit={(1,0.4) (7,3)},
      fill=yellow!10] (gamma-star) { };
      \node[] at (4, 2.5) {$\Gamma^*_n = \{ h \suchthat h \text{ is {\em entropic}}\}$};
      \node[fill opacity=0.5, rounded corners=.55cm, draw, fit={(2,0.5) (6,2)},
      fill=blue!10] (mn) { };
      \node[] at (4, 1.2) {\shortstack{$\Mod_n = \{ h \suchthat h \text{ is {\em modular}}\}$ \\ $h(X) = \sum_{x\in
      X}h(x)$}};
   \end{tikzpicture}
   \caption{Hierarchy of set functions}
   \Description{Hierarchy of set functions}
\label{fig:set:functions}
\end{figure*}
The following chain of inclusion is either
known~\cite{Yeung:2008:ITN:1457455}, or straightforward to show (see also
Figure~\ref{fig:set:functions}):
\begin{prop}
For any positive integer $n$, we have
\begin{equation}
   \Mod_n \subseteq \Gamma^*_n \subseteq \overline \Gamma^*_n
   \subseteq \Gamma_n \subseteq \sa_n.
   \label{eq:chain:M-SA-}
\end{equation}
Furthermore, when $n \geq 4$, all inclusions are strict:
$\Mod_n \subsetneq \Gamma^*_n \subsetneq \overline \Gamma^*_n
\subsetneq \Gamma_n \subsetneq \sa_n$.
\end{prop}

Throughout the paper, we will work on multi-hypergraphs $\calH=(\calV,\calE)$
(i.e. a hyperedge may occur multiple times in $\calE$).

\bdefn
\label{def:ed}
Given a hypergraph $\calH=(\calV, \calE)$, define the following two set functions:
\begin{eqnarray}
\ed \defeq& \{ h \suchthat h : 2^{\calV} \to \R_+, h(F) \leq 1,
\forall F \in \calE\}\label{eqn:ed}\\
\vd \defeq& \{ h \suchthat h : 2^{\calV} \to \R_+, h(\{v\}) \leq 1,
\forall v \in \calV\}\label{eqn:vd}
\end{eqnarray}
$\ed$ stands for {\em edge-dominated} and $\vd$ stands for {\em
vertex-dominated}.
\edefn

Given a set function $h$ and a scalar $s$, we will use $s\cdot h$ to denote $h$
scaled by $s$. If $\calC$ is a class of set functions, then
$s \cdot \calC \defeq \{s\cdot h \suchthat h \in \calC\}$. Specifically, we will be
interested in $\log N\cdot \ed$ and $\log N \cdot \vd$.

To each hypergraph $\calH$ we associate a {\em full conjunctive
query}, or {\em natural join query},
$Q(\mv A_{[n]}) \leftarrow \bigwedge_{F\in \calE} R_F(\mv A_F)$, as discussed
in Section~\ref{sec:the:problems}, where for each hyperedge
$F\in \calE$ there is an input relation $R_F$ with attributes
$\mv A_F$; the set $F$ is called the {\em support} of relation $R_F$.
The hypergraph is a multi-hypergraph because for the same hyperedge
$F$ there can be multiple relations $R_F$ with the same support $F$.
Similarly, we associate the {\em Boolean query},
$Q() \leftarrow \exists \mv A_{[n]} \bigwedge_{F\in \calE} R_F(\mv A_F)$, and
drop the existential quantifiers, writing
$Q() \leftarrow \bigwedge_{F\in \calE} R_F(\mv A_F)$.

For each $F\in \calE$, let $N_F \defeq |R_F|$, where
$N_F \in \N \cup \{\infty\}$. We set $N_F=\infty$ if $R_F$ is not a
materialized relation, a negation of a relation, or if its
size is not known.
Throughout this paper, let
\begin{equation}
   N = \max_{F\in\calE, |R_F|<\infty} N_F.
   \label{eqn:N}
\end{equation}

\subsection{Bounds for queries with cardinality constraints}
\label{subsec:no-FD-HDC}

\subsubsection{Bounds on the worst-case output size}
In the case of a full conjunctive query, the runtime is as least at
large as the size of its answer, hence there is interest in finding
upper bounds on the query answers.  Bounding the worst-case output
size $|Q|$ of a natural join query $Q$ is a well-studied problem.
There is a hierarchy of such bounds: the vertex bound, integral edge
cover bound, and the fractional edge cover bound (also called the
$\AGM$-bound).

\begin{align}
   |Q| & \leq \vertexbound(Q) \defeq N^{n} & \text{Vertex bound} \\
   |Q| & \leq 2^{\rho(Q, (N_F)_{F\in \calE})}
   & \text{Integral edge cover bound} \label{eqn:integral:ecb}\\
   |Q| & \leq \AGM(Q, (N_F)_{F\in\calE}) \defeq 2^{\rho^*(Q, (N_F)_{F\in \calE})}
   & \text{Fractional edge cover bound } \label{eqn:agm:bound}\\
   && \text{(also called {\em AGM bound})}\nonumber
\end{align}

The vertex bound is trivial. The quantities
$\rho(Q, (N_F)_{F\in\calE})$ and
$\rho^*(Q, (N_F)_{F\in\calE})$ are defined via the edge cover polytope:
\begin{align}
   \ec &\defeq \left\{ \vec\lambda \suchthat
\vec\lambda\in \R_+^\calE,
\sum_{\substack{F\in\calE\\ v \in F}}\lambda_F \geq 1, \forall v \in \calV \right\}
   & \text{({\em edge-cover polytope})}.
\end{align}
A vector $\vec\lambda \in \ec$ is called a {\em fractional edge cover}
of the hypergraph, while a vector $\vec\lambda \in \ec \cap \{0, 1\}^\calE$
is called an {\em integral edge cover}, or just an {\em edge cover} of
$\calH$. Then,
\begin{align}
\rho(Q, (N_F)_{F\in\calE})
   &\defeq \min \setof{\log \left(\prod_{F\in\calE}N_F^{\lambda_F}
     \right)}{\vec\lambda \in \ec\cap \{0, 1\}^\calE} \label{eqn:rho:Q:NF}\\
\rho^*(Q, (N_F)_{F\in\calE})
   &\defeq \min \setof{\log \left(\prod_{F\in\calE}N_F^{\lambda_F}
     \right)}{\vec\lambda \in \ec}
\end{align}

The edge cover bounds are dependent on the input relations' sizes.
Often, to state a bound that is independent of the input size,
researchers use cruder approximations of the bound:
\begin{align}
   \rho(Q) &\defeq \frac{1}{\log N} \cdot \rho(Q, (N_F=N)_{F\in\calE}) & \text{the {\em integral edge cover
   number} of } \calH,\\
   \rho^*(Q) &\defeq\frac{1}{\log N} \cdot  \rho^*(Q, (N_F=N)_{F\in\calE}) & \text{the {\em fractional edge
   cover number} of } \calH,\label{eq:defn:rho-star-Q}
\end{align}

(assuming $N>1$.) In particular, note that $2^{\rho(Q, (N_F=N)_{F\in\calE})} = N^{\rho(Q)}$, and
similarly for the $\rho^*$ case.
The AGM-bound by Atserias, Grohe, and Marx~\cite{AGM}
built on earlier works \cite{MR859293,MR1639767,MR599482,GM06}.
One remarkable property of the $\agm$-bound is that it is asymptotically tight.
In addition, there are known algorithms~\cite{NPRR,skew,LFTJ,tetris} with runtime $\tilde
O(2^{\rho^*(Q, (N_F)_{F\in\calE})})$: they are {\em worst-case optimal}.
(Recall that in this paper, the big-$O$ notation is in data-complexity, hiding a factor
that is query-dependent and data-independent.
The big-$\tilde O$ additionally hides a single $\log$-factor in data-complexity.)

\subsubsection{Islands of tractability for Boolean conjunctive queries (and $\csp$ problems)}
We have briefly discussed various width parameters of a Boolean conjunctive query
in Section~\ref{subsec:intro:algo-conjunctive}. In particular, bounded widths often imply
tractability. For unbounded-arity inputs,
Marx~\cite{MR3144912} showed that the submodular width is the best one can hope
for in terms of being in $\fpt$,
unless the exponential time hypothesis fails.
Figure~\ref{fig:hierarchy} explains
the important classes of widths known thus far.\footnote{Redrawn from Marx's
slides at \url{http://www.cs.bme.hu/~dmarx/papers/marx-weizmann-hypergraph.pdf}}
The hierarchy is strict in the sense that there are (infinite) classes of
queries which have bounded submodular width but unbounded fractional hypertree
width, bounded fractional hypertree width but unbounded (generalized)
hypertree width, and so forth.
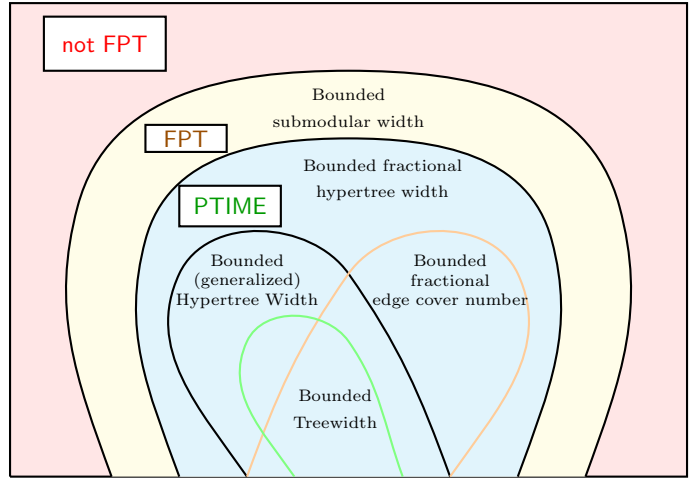
\begin{figure*}[t!]
\centering 
\begin{tikzpicture}[scale = .9]
\def\nextAngle{0}
\tikzset{
    next angle/.style={
        in=#1+180,
        out=\nextAngle,
        prefix after command= {\pgfextra{\def\nextAngle{#1}}}
    },
    start angle/.style={
        out=#1,
        nangle=#1,
    },
    nangle/.code={
        \def\nextAngle{#1}
    }
}

   \draw [fill=red!10,thick] (0, 0) -- (10, 0) -- (10, 7) -- (0, 7) -- (0, 0);
   \draw [fill=yellow!10,thick] (1.5,0) to[start angle=110,next angle=75] (1,4) to[next angle=0] (5,6)
   to[next angle=-75] (9,4) to[next angle=-110] (8.5,0);
   \draw [fill=cyan!10,thick] (2.5,0) to[start angle=110,next angle=75] (2,3.5) to[next angle=0] (5,5)
   to[next angle=-75] (8,3.5) to[next angle=-110] (7.5,0);
   \draw [thick] (3.5,0) to[start angle=180-55,next angle=65] (2.5,3) 
								 to[next angle=-55] (5,3)
   							 to[next angle=-70] (6.5,0);
   \draw [thick,color=orange!40] (6.5,0) to[start angle=55,next angle=180-65] (7.5,3) 
								 to[next angle=180+55] (5,3)
   							 to[next angle=180+70] (3.5,0);
   \draw [thick,color=green!50] (4.2,0) to[start angle=180-55,next angle=65] (3.5,2) 
								 to[next angle=-55] (5,2)
   							 to[next angle=-70] (5.8,0);
	\node [] at (4.8,1.2) {\tiny Bounded};
	\node [] at (4.8,0.8) {\tiny Treewidth};
	\node [] at (3.5,3.2) {\tiny Bounded};
	\node [] at (3.5,2.9) {\tiny (generalized)};
	\node [] at (3.5,2.6) {\tiny Hypertree Width};
	\node [] at (6.5,3.2) {\tiny Bounded};
	\node [] at (6.5,2.9) {\tiny fractional};
	\node [] at (6.5,2.6) {\tiny edge cover number};
	\node [] at (5.5,4.6) {\tiny Bounded fractional};
	\node [] at (5.5,4.2) {\tiny hypertree width};
	\node [] at (5,5.65) {\tiny Bounded};
	\node [] at (5,5.25) {\tiny submodular width};

	\draw[thick,fill=white] (2.5, 3.7) rectangle (4, 4.3) node[pos=.5,color=green!60!black] {\sf PTIME};
	\draw[thick,fill=white] (2.0, 4.8) rectangle (3.2, 5.2) node[pos=.5,color=orange!60!black] {\sf FPT};
	\draw[thick,fill=white] (0.5, 6) rectangle (2.3, 6.8) node[pos=.5,color=red] {\sf not FPT};
\end{tikzpicture}
\caption{Islands of tractability for conjunctive queries and constraint satisfaction problems.}
\Description{Islands of tractability for conjunctive queries and constraint satisfaction problems.}
\label{fig:hierarchy}
\end{figure*}

\subsubsection{Tree decompositions and their widths}
This section defines more precisely commonly used width parameters, which are
parameters of the tree decompositions of the input query.
Tree decompositions capture conditional independence among variables in a query,
facilitating
dynamic-programming.  We refer the reader to the recent survey by
Gottlob et al.~\cite{DBLP:conf/pods/GottlobGLS16} for more details on
historical contexts, technical descriptions, and open problems thereof.

\bdefn
\label{defn:TD}
A {\em tree decomposition} of a hypergraph $\calH=(\calV,\calE)$
is a pair
$(T,\chi)$, where $T$ is a tree and $\chi: V(T) \to 2^{\calV}$ maps each node
$t$ of the tree to a subset $\chi(t)$ of vertices such that
(1) Every hyperedge $F\in \calE$ is a subset of some $\chi(t)$, $t\in V(T)$,
(2) For every vertex $v \in \calV$,
    the set $\{t \suchthat v \in \chi(t)\}$ is a non-empty
    (connected) sub-tree of $T$.
Somewhat confusingly, the sets $\chi(t)$ are often called the {\em bags}
of the tree decomposition.
\edefn
A compact method of defining width parameters is the
framework introduced by
Adler \cite{adler:dissertation}.

\bdefn
Let $\calH=(\calV,\calE)$ be a hypergraph, and
$g : 2^\calV \to \mathbb R_+$ be a set function.
The {\em $g$-width} of a tree decomposition $(T, \chi)$ is
$\max_{t\in V(T)} g(\chi(t))$.
The {\em $g$-width of $\calH$} is the {\em minimum} $g$-width
over all tree decompositions of $\calH$.
Note that the $g$-width of a hypergraph is a {\em minimax} function.
\edefn

\bdefn
For any subset of vertices $B\subseteq \calV$, denote by
$\calH_B = (B, \setof{F\cap B}{F \in \calE})$ the hypergraph
restricted to $B$.  Define $s(B) = |B|-1$, $\rho(B) = \rho(\calH_B)$
the integral edge cover number of $\calH_B$, and
$\rho^*(B) = \rho^*(\calH_B)$, its fractional edge cover number. Then,
the {\em treewidth} of $\calH$, denoted by $\tw(\calH)$, is the
$s$-width of $\calH$.  The {\em generalized hypertree width} of
$\calH$, denoted by $\ghtw(\calH)$ is the $\rho$-width of $\calH$.
And, the {\em fractional hypertree width} of $\calH$, denoted by
$\fhtw(\calH)$, is the $\rho^*$-width of $\calH$.
\label{defn:traditional:widths}
\edefn

Very recently, Fischl et al.~\cite{2016arXiv161101090F} showed that, checking whether
a given hypergraph has a fractional hypertree width or a generalized
hypertree width at most $2$ is $\np$-hard, settling two important open
questions.

The common approach to compute a Boolean query $Q$
using the tree decomposition is to compute the full conjunctive query
associated to the hypergraph $\calH_{\chi(t)}$ at each tree node $t$,
then run Yannakakis' algorithm~\cite{DBLP:conf/vldb/Yannakakis81} on
the acyclic conjunctive query consisting of their results.
It is known~\cite{faq,
DBLP:journals/pvldb/BakibayevKOZ13,
AM00,
DBLP:journals/ai/Dechter99,
DBLP:journals/ai/KohlasW08}
that a vast number of problems in graphical model inference,
database query computation, constraint satisfaction,
and logic can be solved using this strategy.
However, this tree-decomposition-first strategy
has a drawback that once we stick with a tree decomposition we are forced to suffer
the worst-case instance for that tree decomposition.
Marx~\cite{MR3144912,DBLP:journals/mst/Marx11} had a wonderful
observation: if we partition the data first, and then use a different
tree decomposition for each part of the data, then we can in some cases
significantly improve the runtime.
This idea leads to the notions of {\em adaptive width} and {\em submodular width}
of a query, where in essence data partitioning and query decomposition are used interleavingly.

\bdefn
The {\em adaptive width} $\adw(\calH)$ of $\calH$ and
the {\em submodular width} $\subw(\calH)$ of $\calH$ are defined by
\begin{eqnarray}
   \adw(\calH) &\defeq& \max_{h \in \ed \cap \Mod_n} \min_{(T,\chi)} \max_{t \in V(T)} h(\chi(t)),\label{eqn:adw}\\
   \subw(\calH) &\defeq& \max_{h \in \ed \cap \Gamma_n} \min_{(T,\chi)} \max_{t \in V(T)} h(\chi(t))\label{eqn:subw},
\end{eqnarray}
where  $\ed$ is the set of  edge-dominated functions (Definition~\ref{def:ed}).
\edefn

Marx~\cite{MR3144912} proved that
$\adw(\calH) \leq \subw(\calH)\leq\fhtw(\calH)$, and that
$\subw(\calH) = O(\adw^4(\calH))$, for any $\calH$.
There are classes of queries for which $\subw(\calH)$ is bounded
by a constant while $\fhtw(\calH)$ grows with the query size
(see also Example~\ref{ex:gap1}).
Marx also designed an algorithm showing that Boolean conjunctive queries can be
solved in data-complexity time $O(\poly(N^{\subw(\calH)}))$.

The definitions of the traditional width parameters and the adaptive and
submodular width parameters are somewhat disconnected from one another.
In Section~\ref{sec:it}, we will explain how all these width parameters can
be unified using the same information-theoretic framework which can deal with
much more general constraints on the input queries. We will also present
a simple example showing the unbounded gap between $\subw(\calH)$ and $\fhtw(\calH)$
for some class of graphs.

A tree decomposition is {\em non-redundant} if no bag is a
subset of another.
A tree decomposition $(T_1,\chi_1)$ is {\em dominated} by another tree
decomposition $(T_2,\chi_2)$ if every bag of $(T_1,\chi_1)$ is a subset of some
bag of $(T_2,\chi_2)$.
Let $\td(\calH)$ denote the set of all non-redundant tree
decompositions of $\calH$ such that no tree decomposition in $\td(\calH)$ is dominated
by another.
We will also use $\td(Q)$ to denote $\td(\calH)$ where $\calH$ is $Q$'s
hypergraph. In many cases, we drop the qualifier $Q$ or $\calH$ for brevity.
We shall use the following crude estimate in Section~\ref{sec:it}.

\bprop\label{prop:td:properties}(Bounding the number of tree decompositions and the size of each tree decomposition)
Given a hypergraph with $n$ vertices, $|\td|\leq n!$ and
the number of bags of a non-redundant tree decomposition
is at most $n$.
In particular, $|\td| = O(1)$ in data complexity.
\eprop
\bp
Every non-dominated non-redundant tree decomposition can be constructed from a
variable ordering~\cite{faq}, and there are $n!$ variable orderings.
When we run GYO-elimination on such a tree decomposition
with an arbitrary root-bag, every bag contains a variable that does not belong
to the parent bag.
\ep

The above bound is certainly an over-estimate even for graphs with an
exponential number of minimal non-redundant tree decompositions.
For example, if the graph is an
$n$-cycle, then every minimal non-redundant tree decomposition corresponds
precisely to a triangularization of an $n$-gon, whose number is the $(n-1)$th Catalan
number $\frac{1}{n-1}\binom{2n-1}{n-1} = o(4^n)$.

\subsection{Bounds for queries with FD  and degree constraints}
\label{subsec:FD-HDC}

The series of bounds and width parameters from Section~\ref{subsec:no-FD-HDC} were based on a single
class of statistics on the input relations: their sizes. In practice
we very often encounter queries with degree bounds, or functional
dependencies (FD) which can be seen as special cases of degree bounds
equal to 1.  The FDs come from two main sources: primary keys and
builtin predicates (such as $A_1 + A_2 = A_3$).  The degree
constraints come from more refined statistics of the input
(materialized) relations, or from some user-defined predicates, for
example the relation $\text{\sf edit-distance}(A_1,A_2) \leq 2$ has a bounded degree
that depends on the maximum length of the strings $A_1, A_2$.

\bdefn[Degree constraints and their guards]
A {\em degree constraint} is a triple $(X, Y, N_{Y|X})$ where
$X \subset Y \subseteq \calV$ and $N_{Y|X} \in \N \cup \{\infty\}$.
A relation/predicate $R_F$ is said to {\em guard} the degree constraint $(X, Y, N_{Y|X})$
if $X\subset Y \subseteq F$ and for every tuple $\mv t_X$ we have
\begin{equation}\label{eqn:degree}
   \deg_{R_F}(Y | \mv t_X) \defeq
   |\Pi_Y(\sigma_{\mv A_X = \mv t_X}(R_F))|
  \leq N_{Y|X}.
\end{equation}
The quantity on the left-hand side is called the {\em degree of $\mv t_X$}
with respect to $Y$ in relation $R_F$.
Note that a relation may guard multiple degree constraints.

To avoid writing $\log_2$ in many places,
define $n_{Y|X} \defeq \log_2 N_{Y|X}$.
We use $\dc$ to denote a set of degree constraints.
A {\em cardinality constraint} is a degree constraint of the form $(\emptyset, F, N_F)$.
We use $\cc$ to denote a set of cardinality constraints.
An FD $X\to Y$ is a degree constraint of the form $(X, Y, 1)$.
Thus, degree constraints strictly generalize both cardinality constraints and
FDs.
Similar to $\hdc$ defined by~\eqref{eqn:hdc}, we use $\hcc$ (and $\hfd$) to denote
the collections of set functions $h$ satisfying cardinality constraints $\cc$ (and FDs respectively).
\label{defn:guard}
\edefn

As an example on guards of degree constraints, consider an input relation $R(A_1,A_2,A_3)$ satisfying the
following conditions: for every value $a_1$ in the active domain of $A_1$,
there are at most $D$ different values $a_2 \in \Dom(A_2)$ such that
$(a_1,a_2)\in \Pi_{A_1A_2}(R)$. Then, $R$ guards the degree constraint
$(\{A_1\},\{A_1,A_2\}, D)$.

As explained earlier, and as shown in~\cite{csma},
the output size of $Q(\mv A_{\calV})$ can be bounded by
\begin{equation}
   \log_2 |Q| \leq
   \underbrace{\max_{h \in \hdc \cap \bar \Gamma^*_n}h(\calV)}_{\daentropic(Q)}
\leq
   \underbrace{\max_{h \in \hdc \cap \Gamma_n}h(\calV)}_{\dapolymatroid(Q)}
   \label{eqn:daeb:dbpb}
\end{equation}
$\daentropic$ and $\dapolymatroid$ stand for ``degree-aware'' entropic and
polymatroid bounds, respectively.
Note that \eqref{eqn:daeb:dbpb} is a special case of \eqref{eqn:intro:entropic:disjunctive:2} and \eqref{eqn:intro:polymatroid:disjunctive:2}.
The $\csma$ algorithm from~\cite{csma} can solve a Boolean query
$Q$ with known degree constraints
in time $\tilde O(N+\poly(\log N)\cdot 2^{\dapolymatroid(Q)})$.
We shall use the quantities $\daentropic(Q)$ and $\dapolymatroid(Q)$ in
Figure~\ref{fig:diagram}.

\section{Size bounds for full conjunctive queries}
\label{sec:size:bound:fcq}

\subsection{Basic observations}
\label{subsec:basic:bounds}

We revisit known output size bounds for conjunctive queries without FDs nor
degree constraints, and reformulate them from the perspective of the information
theoretic framework.
We first prove a simple lemma, which shows that, sometimes we can
``modularize'' an optimal submodular function under cardinality constraints.
Recall that $\cc$ denotes a set of cardinality constraints and $\hcc$ denotes
the collections of set functions $h$ satisfying the cardinality constraints.

\blmm[Modularization Lemma]
Given a
(not-necessarily full) conjunctive query $Q$ with hypergraph $\calH =
(\calV,\calE)$ and cardinality constraints $\cc$, let $B\subseteq \calV$ be arbitrary. Then, we have
\begin{equation}
   \max\{ h(B) \suchthat h \in \Mod_n \cap \hcc\}
   =
   \max\{ h(B) \suchthat h \in \Gamma^*_n \cap \hcc\}
   =
   \max\{ h(B) \suchthat h \in \Gamma_n \cap \hcc\}.
   \label{eqn:Mod:Gamma}
\end{equation}
\label{lmm:modularization}
\elmm
\bp
Up to renumbering, we can assume $B = [k]$ for some positive integer $k$.
Set $[0]=\emptyset$ by convention.
Let $h^* = \argmax\{ h(B) \suchthat h \in \Gamma_n \cap \hcc \}$.
Define a set function $\bar h$ as follows:
$\bar h(F) = \sum_{i\in F} (h^*([i])-h^*([i-1]))$, for any $\emptyset \neq
F \subseteq \calV$,
and $\bar h(\emptyset) = 0$.
Clearly $\bar h \in \Mod_n$ and $\bar h(B) = h^*(B)$.
Next, we prove $\bar h \in \hcc$ by proving
$\bar h(F) \leq h^*(F)$ for all $F\subseteq \calV$, by induction
on $|F|$. The base case when $|F|=0$ is trivial.
For the inductive step, let $j$ be the maximum integer in $F$, then
by noting that $|F\cap [j-1]|< |F|$, we have
\begin{eqnarray*}
   \bar h(F) &=& h^*([j]) - h^*([j-1]) + \sum_{i \in F-\{j\}} (h^*([i])-h^*([i-1]))\\
   &=& h^*([j]) - h^*([j-1]) + \bar h(F \cap [j-1])\\
   &=& h^*(F \cup [j-1]) - h^*([j-1]) + \bar h(F \cap [j-1])\\
   (\text{induction hypothesis}) &\leq & h^*(F \cup [j-1]) - h^*([j-1]) + h^*(F \cap [j-1])\\
   (\text{submodularity of $h^*$}) &\leq & h^*(F).
\end{eqnarray*}
Consequently, $\bar h \in \Mod_n \cap \hcc$ and thus
$\max\{ h(B) \suchthat h \in \Mod_n \cap \hcc\} \geq
\max\{ h(B) \suchthat h \in \Gamma_n \cap \hcc\}$. The reverse inequality is
trivial because $\Mod_n\subseteq \Gamma_n$.
\ep

Recall that a full conjunctive query is a special case of a disjunctive datalog
rule: It is a disjunctive datalog rule with only one target $B=\calV$.
Hence for a full conjunctive query $Q$, the log-size-bound defined
by~\eqref{eqn:intro:output:bound} simplifies to
\begin{equation}
   \outputbound_{\calF}(Q)\defeq \max_{h\in\calF} h(\calV).
   \label{eq:defn-logsizebound}
\end{equation}
The following simple proposition recaps major known bounds under one
umbrella.  Since the proposition is restricted to full conjunctive
queries rather than the more general disjunctive datalog rules, it makes
stronger claims about size bounds: In particular, some
bounds for full conjunctive queries collapse part of the hierarchy of function classes:
$\Mod_n \subseteq \overline\Gamma^*_n \subseteq \Gamma_n \subseteq \sa_n$.
In the following, we use notations defined in Section~\ref{sec:background}.
\bprop\label{prop:unifying:bound}
Let $Q$ be a full conjunctive query with no FDs whose hypergraph
is $\calH=(\calV,\calE)$.  Then the followings hold:
\begin{align}
\log \vertexbound(Q)
   &= \outputbound_{\calF \cap (\log N \cdot \vd)}(Q),
   &\forall \calF \in \{\Mod_n, \overline\Gamma^*_n, \Gamma_n, \sa_n\}
   \label{eqn:u1}\\
   \rho(Q, (N_F)_{F\in\calE}) &= \outputbound_{\sa_n \cap \hcc}(Q)\label{eqn:u2} \\
\rho^*(Q) \cdot \log N
   &= \outputbound_{\calF \cap (\log N \cdot \ed)}(Q),
   &\forall \calF \in \{\Mod_n, \overline\Gamma^*_n, \Gamma_n\}
   \label{eqn:u3}\\
   \log \agm(Q)
   &= \outputbound_{\calF \cap \hcc}(Q),
   &\forall \calF \in \{\Mod_n, \overline\Gamma^*_n, \Gamma_n\}
   \label{eqn:u4}
\end{align}
\eprop
\bp
By comparing~\eqref{eq:defn:rho-star-Q} to~\eqref{eqn:agm:bound}, we notice
that~\eqref{eqn:u3} is the special case of~\eqref{eqn:u4} where $N_F= N$ for all $F \in \calE$.
Hence~\eqref{eqn:u4} implies~\eqref{eqn:u3}.
Next, we prove~\eqref{eqn:u4}.
Note that $\log \AGM(Q)$ and $\outputbound_{\Mod_n\cap \hcc}(Q)$ are the solutions to dual
linear programs. In particular, from~\eqref{eqn:agm:bound}, we have
\begin{align*}
   \log \AGM(Q) = \min&\sum_{F \in \calE}\lambda_F \cdot \log N_F\\
   \text{such that }&
   \sum_{\substack{F\in\calE\\ v \in F}}\lambda_F \geq 1,& \forall v \in \calV\\
   &\lambda_F \geq 0,& \forall F \in \calE.
\end{align*}
And from~\eqref{eq:defn-logsizebound} and the definitions of $\Mod_n$ and $\hcc$, we have
\begin{align*}
   \outputbound_{\Mod_n\cap \hcc}(Q) = \max & \sum_{v \in [n]} h(v)\\
   \text{such that }& \sum_{v \in F} h(v) \leq \log N_F,& \forall F \in \calE,\\
   & h(v) \geq 0, & \forall v \in [n].
\end{align*}
By strong
duality of linear programming, we have
$\log \AGM(Q) =\outputbound_{\Mod_n\cap \hcc}(Q)$.
Lemma~\ref{lmm:modularization} implies
$\outputbound_{\Mod_n\cap \hcc}(Q) =
\outputbound_{\Gamma_n \cap \hcc}(Q)$.

To prove~\eqref{eqn:u1}, note that $\log \vertexbound(Q) = \outputbound_{\Mod_n
\cap (\log N \cdot \vd)}(Q)$ is trivial. Since $\Mod_n \subseteq \sa_n$ it is
sufficient to show that
$\outputbound_{\sa_n \cap (\log N \cdot \vd)}(Q) \leq n\log N =
\log\vertexbound(Q)$.
Let $h \in \sa_n \cap (\log N \cdot \vd)$ be arbitrary, then from sub-additivity
we have $h(\calV) \leq \sum_{v \in \calV}h(v) \leq \sum_{v\in \calV}\log N=n\log
N$, which completes the proof.

Lastly, we prove equality~\eqref{eqn:u2} which is the only non-trivial
statement in this proposition.
Recall the definition of $\rho(Q, (N_F)_{F\in\calE})$ from~\eqref{eqn:rho:Q:NF}.
Let $(\lambda_F)_{F\in\calE} \in \ec \cap \{0, 1\}^{\calE}$ be any integral edge cover
of $\calH$, and $h \in \sa_n \cap \hcc$ be arbitrary.
Let $\{F_1,\dots,F_k\}$ be the collection of all hyperedges in $\calE$
with $\lambda_{F_i}=1, \forall i\in [k]$.
Then,
\begin{align*}
   h(\calV)
   &\leq \sum_{i=1}^k h\left(F_i \setminus \bigcup_{j=1}^{i-1}F_j\right)
   &(\text{sub-additivity}) \\
   &\leq \sum_{i=1}^k h\left(F_i\right)
   &(\text{monotonicity})     \\
   &= \sum_{F \in\calE} \lambda_F \cdot h(F)\\
   &\leq \sum_{F \in \calE} \lambda_F \log N_F
   &(\text{cardinality constraints})                          \\
   &= \log \left( \prod_{F\in\calE} N_F^{\lambda_F} \right).
\end{align*}
This proves $\outputbound_{\sa_n \cap \hcc}(Q) \leq \rho(Q,(N_F)_{F\in\calE})$.
Conversely, suppose $h^*$ is an optimal solution to the linear program
$\max \setof{h(\calV)}{h \in \sa_n\cap\hcc}$, which is explicitly written as
\begin{align}
   \max \myskip h(\calV)\label{eqn:LP:sa}\\
   \text{s.t.} \myskip h(I\cup J) - h(I) - h(J)&\leq 0, & I \incomp J
   &&\text{ (i.e. } I \not\subseteq J \text{ and } J \not\subseteq I)\nonumber\\
   h(F) &\leq \log N_F, & F\in\calE\nonumber\\
   h(X) - h(Y) &\leq 0, & X \subset Y \subseteq \calV\nonumber\\
     h(Z) &\geq 0 & \emptyset\neq Z\subseteq \calV\nonumber.
\end{align}
We are done if we can prove the following claim:

\begin{claim}
There is an integral edge cover $(\lambda_F)_{F\in\calE}$ such that
$\sum_{F\in\calE} \lambda_F \log N_F \leq h^*(\calV)$.
\label{clm:unifying:bound:1}
\end{claim}

To prove this claim, we need the dual of the LP~\eqref{eqn:LP:sa}.
Associate a dual variable $\delta_F$ to each
cardinality constraint $h(F)\leq \log N_F$, a dual variable $\sigma_{I,J}$
to each sub-additivity constraint $h(I\cup J) - h(I) - h(J)\leq 0$,
and a dual variable $\mu_{X, Y}$ to each monotonicity constraint $h(X) - h(Y) \leq 0$.
For convenience, let $\delta_Z\defeq 0$ for any $Z\in  2^\calV\setminus \calE$.
For any $Z\subseteq \calV$, define $\flowsa(Z)$ to be
\begin{equation}
\flowsa(Z)\defeq\delta_Z + \sum_{I\incomp J: I\cup J = Z}\sigma_{I,J} - \sum_{J:J \incomp
   Z} \sigma_{Z,J}- \sum_{X:X\subset Z} \mu_{X, Z} +\sum_{Y:Z\subset Y} \mu_{Z, Y}.
   \label{eqn:flowsa}
   \end{equation}
Then, the dual LP is
\begin{eqnarray}
   \min \myskip &\sum_{F\in\calE} (\log N_F) \cdot \delta_F\label{eqn:dual:LP:sa}\\
   \text{s.t.}\myskip  &\flowsa(Z)& \geq  0, \myskip \emptyset \neq Z\subseteq \calV\nonumber\\
  & \flowsa(\calV)& \geq  1\nonumber\\
   &(\vec\delta, \vec\sigma, \vec\mu) &\geq \mv 0 \nonumber.
\end{eqnarray}

We need the following auxiliary claim.

\begin{claim}
Given a rational feasible dual solution $(\vec\delta,\vec\sigma,\vec\mu)$,
let $D$ be a common denominator for all entries in $(\vec\delta,\vec\sigma,\vec\mu)$.
Then, there exist $D$ integral edge covers
$\vec\lambda^{(i)} = (\lambda_F^{(i)})_{F\in\calE}$, $i \in [D]$,
such that for any $F\in\calE$, $\sum_{i=1}^D\lambda_F^{(i)}\leq D\cdot \delta_F$.
\label{clm:unifying:bound:2}
\end{claim}

Assuming Claim~\ref{clm:unifying:bound:2} holds, we prove Claim~\ref{clm:unifying:bound:1}.
Let $(\vec\delta^*,\vec\sigma^*,\vec\mu^*)$ be a rational optimal dual solution, then
$h^*(\calV)=\sum_{F\in\calE}(\log N_F)\cdot \delta^*_F$ by strong duality.
Applying Claim~\ref{clm:unifying:bound:2} on $(\vec\delta^*,\vec\sigma^*,\vec\mu^*)$, we get
\[ h^*(\calV)
\geq \sum_{F\in\calE} \sum_{i \in [D]} \frac{\lambda_F^{(i)}}{D}\log N_F
= \frac 1 D \sum_{i\in [D]} \sum_{F\in\calE} \lambda_F^{(i)}\log N_F
   \geq \min_{i\in [D]} \sum_{F\in\calE} \lambda_F^{(i)} \log N_F.
\]
This proves Claim~\ref{clm:unifying:bound:1}.

{\it Proof of Claim~\ref{clm:unifying:bound:2}:}
Note the following: if $\flowsa(B)>\delta_B$ for some $B\neq\emptyset$,
then from~\eqref{eqn:flowsa} and the fact that $(\vec\delta,\vec\sigma,\vec\mu)$
is feasible to the dual LP, there must be either
(1) $I\incomp J$ with $I\cup J=B$ and $\sigma_{I, J}>0$,
or (2) $Y\supset B$ with $\mu_{B, Y}>0$.
Let $w\defeq1/D$.
Initially, let $\calB=\{\calV\}$.
While there is $B\in \calB$ with $\flowsa(B)>\delta_B$, we do the following:
\bi
\item If Case (1) holds,
then we reduce $\sigma_{I, J}$ by $w$, remove $B$ from $\calB$, and add both $I$ and $J$ to $\calB$.
\item If Case (2) holds, then we reduce $\mu_{B, Y}$ by $w$, remove $B$ from $\calB$, and add $Y$ to $\calB$.
\ei
Either way, we are maintaining the following invariants (which initially hold):
\bi
\item $\flowsa(B)>0$ for any $B\in\calB$.
\item $\flowsa(Z)\geq 0$ for any $\emptyset\neq Z\subseteq \calV$.
\item For any $v\in \calV$, there is some $B\in\calB$ such that $v\in B$.
\ei
By the assumption of Claim~\ref{clm:unifying:bound:2}, any non zero value in
$(\vec\delta,\vec\sigma,\vec\mu)$ must be $\geq w$, hence
the above process terminates because at each iteration we are reducing
$\norm{\vec\sigma}_1+\norm{\vec\mu}_1$ by $w$.
When it does terminate,  we have $\delta_B\geq\flowsa(B)>0$ for all $B\in\calB$.
Note that $\delta_Z=0$ for all $Z \in 2^\calV - \calE$.
Hence, $(\lambda_F^{(1)})_{F\in \calE}$ is an integral edge cover where
\[
\lambda_F^{(1)}\defeq
\begin{cases}
1 \myskip\text{ if $F \in \calB$},\\
0 \myskip\text{ otherwise}.
\end{cases}
\]
This way, we have constructed the first integral edge cover $\vec\lambda^{(1)}$
out of the promised $D$ integral edge covers.
Now for every $B\in\calB$, we reduce $\delta_B$ (hence $\flowsa(B)$) by $w$
and remove $B$ from $\calB$.
Notice that except for the initial $\calV$ in $\calB$, the following was always maintained.
Before we added any $B$ to $\calB$ we increased $\flowsa(B)$ by $w$,
and before we removed any $B'$ from $\calB$ we decreased $\flowsa(B')$ by $w$.
Therefore, at the end $\flowsa(Z)$ is unchanged for all $Z\subset \calV$,
and $\flowsa(\calV)$ is reduced by $w$.
If the final $\flowsa(\calV)=0$, then we are done.
Otherwise, $\frac{1}{1-w}\cdot(\vec\delta,\vec\sigma,\vec\mu)$ is a feasible
dual solution with $(D-1)$ as a common denominator, and we can repeat the
process to construct the next integral edge cover.
\ep

The top part of Figure~\ref{fig:diagram} (which we shall elaborate more on
later)
depicts all the bounds in Proposition~\ref{prop:unifying:bound}.
Regarding the $\agm$-bound, the fact that the size bound on $\Gamma_n$ is equal
to the size bound on $\Mod_n$ allows us to compute the bound efficiently
(polynomial time in query complexity). By contrast, computing
the integral edge cover bound is $\np$-hard in query complexity, despite
being the worse bound, i.e. the bound over the more relaxed function class $\sa_n$.
From the above proposition and the chain of inclusions
$\Mod_n \subseteq \overline\Gamma^*_n \subseteq \Gamma_n \subseteq \sa_n$
and $\hcc \subseteq \vd\cdot \log N$,
we have
$\vertexbound(Q) \geq 2^{\rho(Q, (N_F)_{F\in\calE})} \geq \agm(Q)$ for any $Q$.

\subsection{The polymatroid size bound is not tight!}
\label{subsec:glvv:bound:is:not:tight}

Next, we prove Theorem~\ref{thm:glvv:not:tight}, which states that the
polymatroid bound is not tight for queries with
cardinality and FD constraints.

\colorlet{myGreen}{green!60!black}
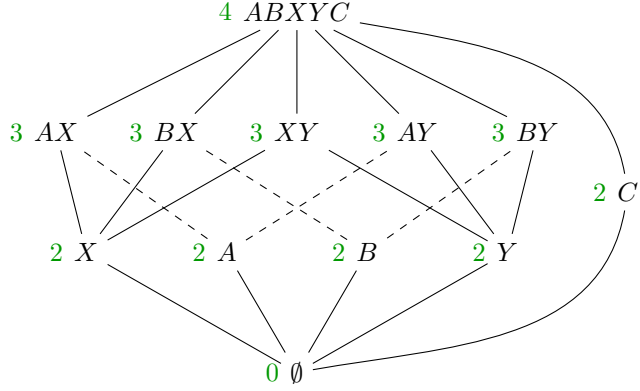
\begin{figure}[th!]
\centering
   \begin{tikzpicture}[scale=0.7]
\node[] at (4, 0) (ABXYC){$ABXYC$}; \node[myGreen, left=-.1 of ABXYC] {4};

\node[] at (0, -2) (AX){$AX$}; \node[myGreen, left=-.1 of AX] {3};
\node[] at (2, -2) (BX){$BX$}; \node[myGreen, left=-.1 of BX] {3};
\node[] at (4, -2) (XY){$XY$}; \node[myGreen, left=-.1 of XY] {3};
\node[] at (6, -2) (AY){$AY$}; \node[myGreen, left=-.1 of AY] {3};
\node[] at (8, -2) (BY){$BY$}; \node[myGreen, left=-.1 of BY] {3};

\node[] at (.5, -4) (X){$X$}; \node[myGreen, left=-.1 of X] {2};
\node[] at (.5+7/3, -4) (A){$A$}; \node[myGreen, left=-.1 of A] {2};
\node[] at (.5+14/3, -4) (B){$B$}; \node[myGreen, left=-.1 of B] {2};
\node[] at (7.5, -4) (Y){$Y$}; \node[myGreen, left=-.1 of Y] {2};

\node[] at (4, -6) (0){$\emptyset$}; \node[myGreen, left=-.1 of 0] {0};

\node[] at (9.5, -3) (C){$C$}; \node[myGreen, left=-.1 of C] {2};

\path[] (0) edge (X);
\path[] (0) edge (A);
\path[] (0) edge (B);
\path[] (0) edge (Y);

\path[] (X) edge (AX);
\path[] (X) edge (BX);
\path[] (X) edge (XY);
\path[] (Y) edge (BY);
\path[] (Y) edge (AY);
\path[] (Y) edge (XY);
\path[] (A) edge (AX);
\path[] (A) edge (AY);
\path[] (B) edge (BX);
\path[] (B) edge (BY);

\path[] (AX) edge (ABXYC);
\path[] (BX) edge (ABXYC);
\path[] (XY) edge (ABXYC);
\path[] (AY) edge (ABXYC);
\path[] (BY) edge (ABXYC);

\draw[] (0) to [out=10, in=-100] (C);
\draw[] (C) to [out=100, in=-10] (ABXYC);
\end{tikzpicture}
\caption{A polymatroid $h$ (shown in {\color{myGreen}green})
on five variables $A, B, X, Y, C$.  For each
missing set of variables $\mv Z$, $h(\mv Z)\defeq h(\mv Z^+)$, where
$\mv Z^+$ is the smallest set shown in the figure that
contains $\mv Z$. This polymatroid proves Claim~\ref{clm:glvv:not:tight:2}.}
\Description{A polymatroid $h$ (shown in {\color{myGreen}green})
on five variables $A, B, X, Y, C$.  For each
missing set of variables $\mv Z$, $h(\mv Z)\defeq h(\mv Z^+)$, where
$\mv Z^+$ is the smallest set shown in the figure that
contains $\mv Z$. This polymatroid proves Claim~\ref{clm:glvv:not:tight:2}.}
\label{fig:yeung}
\end{figure}

\begin{proof}[Proof of Theorem~\ref{thm:glvv:not:tight}]
Consider the following which we shall call the {\em Zhang-Yeung} query:
\begin{align}
   Q(A,B,X,Y,C) & \cd K(A,B,X,Y,C), R(X,Y),S(A,X),T(A,Y),U(B,X),V(B,Y),W(C)
\end{align}
with given cardinality constraints
$|R|,|S|,|T|,|U|,|V| \leq N^3$, $|W|\leq N^2$ (none on $K$, or $|K|\leq N^c$
for large constant $c$), and the
following keys in $K$: $AB$, $ AXY$, $ BXY$, $AC$, $XC$, $YC$.
Hence $K$ satisfies the following functional dependencies:
\begin{align*}
   AB &\rightarrow XYC,& AXY &\rightarrow BC,\\
   BXY &\rightarrow AC,& AC &\rightarrow BXY,\\
   XC &\rightarrow ABY,& YC &\rightarrow ABX.
\end{align*}
We prove two claims.

\begin{claim}
The following inequality holds for all entropic functions $h \in \Gamma^*_5$:
\begin{align}
   11h(ABXYC) &\leq 3h(XY)+3h(AX)+3h(AY)+h(BX) +h(BY)+5h(C) \nonumber\\
   &+ (h(ABXYC|AB)+4h(ABXYC|AXY) +h(ABXYC|BXY)) \nonumber \\
   &+ (h(ABXYC|AC)+2h(ABXYC|XC)
   +2h(ABXYC|YC)).\label{eqn:claim1:entropic}
\end{align}
\label{clm:glvv:not:tight:1}
\end{claim}

\begin{claim}
There exists a polymatroid $h$ satisfying all cardinality
and FD constraints such that $h(ABXYC) = 4\log N$.
\label{clm:glvv:not:tight:2}
\end{claim}

The first claim above implies $11\log|Q| \leq 11 \log N^3 + 5 \log N^2 = 43 \log N$,
or $|Q| \leq N^{4 - \frac{1}{11}}$,
while the second claim shows that the polymatroid bound is $\geq N^4$.
In other words, the two claims show a ratio of $N^{\frac{1}{11}}$ between
the two bounds.
To construct a query with an amplified gap of $N^{s}$ between the two
bounds, consider a query which is a cross-product of $11s$ variable-disjoint
copies of the basic Zhang-Yeung query.
This proves Theorem~\ref{thm:glvv:not:tight}.

{\em Proof of Claim~\ref{clm:glvv:not:tight:1}.}
Let $I(\mv X; \mv Y | \mv Z) \defeq h(\mv X\mv Z)+h(\mv Y\mv Z)- h(\mv X\mv Y\mv Z)
- h(\mv Z)$ denote the conditional mutual information between random variables
$\mv X$ and $\mv Y$ conditioned on random variables $\mv Z$.
In a breakthrough paper in information theory,
Zhang and Yeung~\cite{DBLP:journals/tit/ZhangY98} proved that
$\overline \Gamma^*_4 \subsetneq \Gamma_4$ by proving that the
following inequality is a non-Shannon-type inequality
(see \cite[Th.15.7]{Yeung:2008:ITN:1457455}), i.e. an inequality satisfied by
all entropic functions and not satisfied by some polymatroid:
\begin{align*}
   2I(X; Y) & \leq I(A;B) + I(A;XY) + 3I(X;Y|A)+I(X;Y|B),
\end{align*}
or equivalently
\begin{align}
h(AB)+4h(AXY) +h(BXY)
   & \leq 3h(XY)+3h(AX) +3h(AY)+h(BX)+h(BY) \nonumber\\
   & - h(A)- 2h(X)-2h(Y). \label{eqn:zy:non:shannon:h}
\end{align}
Let $h$ be any entropic function on $5$ variables $A,B,X,Y,C$, then it still
satisfies inequality~\eqref{eqn:zy:non:shannon:h}, because the restriction of
$h$ on the $4$ variables $A,B,X,Y$ is the marginal entropy.
Now, add $h(ABXYC|AB)+4h(ABXYC|AXY)+h(ABXYC|BXY)$ to both sides:
\begin{align}
   6h(ABXYC) &\leq 3h(XY)+3h(AX)+3h(AY)+h(BX) +h(BY) - (h(A)+ 2h(X)+2h(Y)) \nonumber \\
   &+ (h(ABXYC|AB)+ 4h(ABXYC|AXY)+h(ABXYC|BXY))\label{eq:tmp1}
\end{align}

From these three Shannon-type inequalities, $h(A)+h(C)\geq h(AC)$,
$h(X)+h(C)\geq h(XC)$, and $h(Y)+h(C) \geq h(YC)$, we derive the
following:

\begin{align}
   5h(ABXYC) &\leq h(A)+2h(X)+2h(Y)+5h(C) \nonumber\\
   &+(h(ABXYC|AC)+2h(ABXYC|XC)+2h(ABXYC|YC)) \label{eq:tmp2}
\end{align}

By adding Eq.~(\ref{eq:tmp1}) and~(\ref{eq:tmp2}) we
obtain~\eqref{eqn:claim1:entropic}.

{\em Proof of Claim~\ref{clm:glvv:not:tight:2}.} The proof is shown in
Figure~\ref{fig:yeung}.  The figure shows a
polymatroid $h$ on the five variables $A, B, X, Y, C$.  For each
missing set of variables $\mv Z$, $h(\mv Z)\defeq h(\mv Z^+)$, where
$\mv Z^+$ is the smallest set shown in the figure that
contains $\mv Z$: in other words, the figure shows the closed sets of
a closure on $\set{A,B,X,Y,C}$.  For example,
$h(AB) \defeq h(AB^+) = h(ABXYC)=4$, etc.  One can check that
$\hat h \defeq \log N \cdot h$
is a polymatroid satisfying all cardinality constraints and functional
dependencies: for example $\hat h(XY)=3\log N$, $\hat h(AX)=3\log N$, and
$\hat h(AB) = \hat h(ABXYC)$ satisfies the FD
$AB\rightarrow XYC$.
\end{proof}

\section{Size bounds for disjunctive datalog rules}
\label{sec:size:bound:ddl}

\subsection{The entropic bound for disjunctive datalog}
\label{subsec:entropic:bound:ddl}

\blmm[Part $(i)$ of Theorem~\ref{thm:intro:size:bound:disjunctive}]
\label{lmm:size:bound:disjunctive}
For any disjunctive datalog rule $P$ of the
form~\eqref{eqn:disjunctive:datalog:query}, and any database instance
$\mv D$, the following holds:
\[ \log |P(\mv D)| \leq \outputbound_{\overline\Gamma^*_n \cap \hdc}(P)
              \leq \outputbound_{\Gamma_n \cap \hdc}(P).
\]
\elmm
\bp
Since $\overline \Gamma^*_n \subseteq \Gamma_n$, the second inequality
is trivial. To show the first inequality, let $T$ denote the set of all tuples
$\mv t$ satisfying the body of $P$;
construct a set $\overline T$ of tuples as follows.
We scan though tuples $\mv t \in T$ one at a time, and either add $\mv t$ to
$\overline T$ or ignore $\mv t$.
To decide whether to add $\mv t$ to $\overline T$, we also keep a collection
of tables
$\mv T = (\overline T_B)_{B\in\calB}$. These tables shall form a model of
the disjunctive datalog rule $P$.
Initially $\overline T$ and all the $\overline T_B$ are empty.
Consider the next tuple $\mv t$ taken from $T$. If $\Pi_B(\mv t) \in \overline T_B$ for
{\em any} target $B \in \calB$, then we ignore $\mv t$.
Otherwise, we add $\Pi_B(\mv t)$ to $\overline T_B$ for {\em every} $B\in \calB$, and
add $\mv t$ to $\overline T$. In the end, obviously the collection
$(\overline T_B)_{B\in\calB}$ is a model of the disjunctive datalog rule.
Furthermore, by construction the tuples $\mv t \in \overline T$ satisfy the
property that: for every two different tuples $\mv t, \mv t' \in \overline T$,
every $B\in\calB$,
we have $\Pi_B(\mv t) \neq \Pi_B(\mv t')$ and both
$\Pi_B(\mv t)$ and $\Pi_B(\mv t')$ are in $\overline T_B$.

Now, construct a joint probability distribution on $n$ variables by picking
uniformly a tuple from $\overline T$. Let $\overline h$ denote the entropy function
of this distribution. Then because $\Pi_B(\mv t) \neq \Pi_B(\mv t')$ for
every different $\mv t, \mv t' \in \overline T$ and $B\in\calB$:
\begin{equation}
   \log_2 |\overline T| = \overline h(\calV) = \overline h(B), \forall B \in \calB.
   \label{eqn:key:maximal:matching}
\end{equation}
Furthermore, $\overline h \in \overline\Gamma^*_n$ by definition,
and $\overline h \in \hdc$ because $\overline T \subseteq T$ and $T$ was the join
of input relations satisfying all degree constraints.  Consequently,
\begin{eqnarray*}
   \log |P(\mv D)| =   \min_{(T_B)_{B\in\calB} \models P}\max_{B\in\calB}\log_2 |T_B|
   &\leq&
   \max_{B\in\calB}\log_2 |\overline T_B|\\
   (\text{due to~\eqref{eqn:key:maximal:matching}})   &=&
   \max_{B\in\calB} \overline h(B)\\
   (\text{due to~\eqref{eqn:key:maximal:matching}})   &=&
   \min_{B\in\calB} \overline h(B)\\
   (\text{because } \bar h \in \overline\Gamma^*_n \cap \hdc)&\leq&
   \max_{h\in\overline\Gamma^*_n \cap \hdc}
   \min_{B\in\calB} h(B)\\
   &=& \outputbound_{\overline\Gamma^*_n \cap \hdc}(P).
\end{eqnarray*}
\ep

\subsection{The entropic bound is asymptotically tight under degree constraints}
\label{subsec:entropic:bound:is:tight}

Recall that $\hdc$, $\hcc$ and $\hfd$ denote
the collection of set functions $h$ satisfying the degree constraints, cardinality constraints and functional dependencies respectively.
Recently, Gogacz and Toru{\'n}czyk~\cite{szymon-2015} showed that\linebreak
$\outputbound_{\overline \Gamma^*_n \cap \hcc \cap \hfd}(Q)$ is an asymptotically
tight upper bound for full {\em conjunctive} queries $Q$
with {\em FDs}. This section proves a generalization of their result in two different directions: We show that
$\outputbound_{\overline \Gamma^*_n \cap \hdc}(P)$ is asymptotically tight for
{\em disjunctive datalog rules} $P$ with given {\em degree constraints} $\hdc$.

The proof is based on the same observation
made in~\cite{szymon-2015}: one can take the {\em group
characterizable entropy function} from Chan and Yeung~\cite{DBLP:journals/tit/ChanY02}
and turn it into a database instance.

Three technical issues we need to nail down to push the proof through
are the following.
First, we need to make sure that the database instance so constructed
satisfies the given degree constraints (which include FD and cardinality
constraints).
Second, it takes a bit more care to define what we meant by ``asymptotically tight''
in the case of queries with degree bounds.
When the input query has only input relation cardinality bounds and functional
dependencies, one can set all input relations to be of the same cardinality $N$ and
let $N$ go to infinity. The tightness of the bound is in the exponent
$\alpha$ (if the bound was $N^\alpha$). This was the result from Gogacz and
Toru{\'n}czyk. When the input query has degree bounds, it does not make sense to
set all degree bounds (and cardinality bounds) to be $N$: we want a finer level of
control over their relative magnitudes.
Third, and most importantly,
unlike in the full conjunctive query case where there is only one output;
in the disjunctive datalog rule case there are multiple models and we have to
show that {\em any} model to the rule must have size asymptotically
no smaller than the worst-case entropic bound.

We start with the construction from Chan and
Yeung~\cite{DBLP:journals/tit/ChanY02}.

\bdefn[Database instance from group system]
Let $G$ be a finite group and $G_1,\dots,G_n$ be $n$ subgroups of $G$.
A database instance associated with this group system $(G,G_1,\dots,G_n)$ is
constructed as follows. There are $n$ attributes $A_1,\dots,A_n$. The domain of
attribute $A_i$ is the left coset space $G/G_i$, i.e. the collection of all
left cosets $\{ gG_i \suchthat g \in
G\}$ of the subgroup $G_i$. For every $F\subseteq [n]$, define a relation $R_F$
by
\begin{equation}
R_F \defeq \{ (gG_i)_{i\in F} \suchthat g \in G\}.
\end{equation}
In other words, $R_F$ is the set of tuples $\mv a_F = (a_i)_{i\in F}$
on attributes $\mv A_F$, where $a_i = gG_i$ for some $g\in G$.
The group element $g$ is said to {\em define} the tuple $\mv a_F$.
\edefn

Given a group system $(G,G_1,\dots,G_n)$,
for any $\emptyset \neq F\subseteq [n]$, define
$G_F \defeq \bigcap_{i\in F} G_i$; and, define $G_\emptyset=G$.
Note that $G_Z$ is a subgroup of $G_Y$ for all $\emptyset \subseteq Y \subseteq Z \subseteq [n]$.
Recall the notion of degree defined in~\eqref{eqn:degree}. We prove the
following lemma, which is a slight generalization of Proposition 4 from
Chan~\cite{DBLP:journals/corr/abs-cs-0702064}.

\blmm\label{lmm:group:degree}
Let $\mv D$ be a database instance associated with the group system
$(G,G_1,\dots,G_n)$.
Let $\emptyset \subseteq Z \subset Y \subseteq [n]$ be two sets of attributes
($Z$ can be empty).
Then, for any tuple $\mv a_Z \in R_Z$:
\begin{equation} \deg_{R_Y}(Y | \mv a_Z) = \frac{|G_Z|}{|G_Y|}.  \end{equation}
In particular, $Z\to Y$ is an FD satisfied by the database
instance iff $G_Z$ is a subgroup of $\bigcap_{i\in Y-Z} G_i$.
\elmm
\bp
For any $F \subseteq [n]$,
two group elements $b, c\in G$ define the same tuple $\mv a_F$ in $R_F$,
i.e. $cG_i=bG_i, \forall i\in F$, if and only if $b^{-1}c \in \bigcap_{i\in F} G_i = G_F$,
which is equivalent to $b=cg$ for some $g \in G_F$.

Now suppose $\mv a_Z \in R_Z$ was defined by an element $b \in G$,
i.e. $\mv a_Z = (a_i)_{i \in Z} = (bG_i)_{i\in Z}$.
Consider any two group elements $c_1,c_2\in G$ which define the same $\mv a_Z$
tuple but different $\mv a_Y$ tuples.
First, $c_1,c_2 \in G$ define the same $\mv a_Z$ iff there are two elements
$g_1,g_2 \in G_Z$ such that $c_1=bg_1$ and $c_2=bg_2$.
Second, $c_1$ and $c_2$ define different $\mv a_Y$ tuples iff
$c_1^{-1}c_2 \notin G_Y$, which is equivalent to
$g_1^{-1}b^{-1}bg_2 = g_1^{-1}g_2 \notin G_Y$, which in turn is equivalent to
$g_1G_Y \neq g_2G_Y$.
Thus, the degree of $\mv a_Z$ in $R_Y$ is precisely the index of the subgroup
$G_Y$ in the group $G_Z$, which is equal to $|G_Z|/|G_Y|$.

If $Z \to Y$ is a functional dependency, then from the fact that $G_Y$ is a
subgroup of $G_Z$ and the inequality $|G_Z|/|G_Y| \leq 1$, we conclude that
$G_Z$ is equal to $G_Y$. However,
\begin{equation}
   G_Y = \bigcap_{i \in Y} G_i =
\bigcap_{i \in Z} G_i\cap \bigcap_{i \in Y-Z} G_i =
G_Z \cap \bigcap_{i \in Y-Z} G_i.
\label{eq:GY:GZ:GY-Z}
\end{equation}
Therefore, $G_Z$ has to be a subgroup of $\bigcap_{i \in Y-Z} G_i$.
For the opposite direction, suppose
$G_Z$ is a subgroup of $\bigcap_{i \in Y-Z} G_i$.
From~\eqref{eq:GY:GZ:GY-Z}, we have $G_Z = G_Y$.
\ep

\bdefn[Scaled-up degree constraints, $\hdc \times k$]
For a given set of degree constraints $\hdc$, let $\hdc \times k$ denote the
same set of constraints but with all the degree bounds multiplied by $k$.
The constraints $\hdc \times k$ are called the ``scaled-up'' degree constraints.
\label{defn:hdcXk}
\edefn

\blmm[Part $(ii)$ of Theorem~\ref{thm:intro:size:bound:disjunctive}]
Let $P$ be the {\em disjunctive datalog rule}~\eqref{eqn:disjunctive:datalog:query}
with degree constraints $\hdc$, where
$\outputbound_{\overline \Gamma^*_n \cap \hdc}(P) \geq 1.$
For any $0<\epsilon<1$, there exists a scale factor $k$
and a database instance $\mv D$ satisfying the scaled-up degree
constraints $\hdc \times k$ for which
the following holds:
\[ 
\log |P(\mv D)|
\geq \left( 1-\epsilon \right) \cdot
\outputbound_{\overline \Gamma^*_n \cap \hdc \times k}(P).
\]
\label{lm:entropic:tight:dc}
\elmm
\bp
Let $H$ be defined as follows:
\begin{equation}
   H \defeq \argmax_{ h \in \overline \Gamma^*_n \cap \hdc}\min_{B\in\calB} h(B).
   \label{eq:entropic:tight:H}
\end{equation}
Note that, $\overline\Gamma^*_n$ is a convex cone, and thus
$kH$ is an optimal solution to the problem
\[\max_{h \in \overline \Gamma^*_n \cap \hdc \times k}\min_{B\in\calB} h(B).\]
Following Chan and Yeung~\cite{DBLP:journals/tit/ChanY02} without loss
of generality we assume that there is a distribution on $n$ random variables
$A_1, \dots, A_n$ such that all their domains $\Dom(A_i)$ are discrete and
finite, and that $H$
is the joint entropy all of whose marginal entropies are rational numbers.
(In general, just as in Chan and Yeung~\cite{DBLP:journals/tit/ChanY02}, if
$H$ is not rational, or is not an entropy over discrete random variables of
finite domains, we construct a sequence of distributions over discrete and
finite domains whose entropies tend to $H$.)

For any tuple $\mv a_F \in \prod_{i\in F} \Dom(A_i)$, we use $p(\mv a_F)$
to denote $\pr[\mv A_F=\mv a_F]$, i.e. $p$ is the probability mass function.
Let $d$ denote the minimum common denominator of all probabilities $p(\mv
a_{[n]})$, over all tuples $\mv a_{[n]}$ for which $p(\mv a_{[n]})>0$.
Set $\delta$ such that $\epsilon = \frac{2\delta}{1+\delta}$,
and let $r$ be a multiple of $d$ sufficiently large so that
the following hold:
\begin{eqnarray}
   d\log_2 e + \frac d 2 \log_2 d &\leq& \delta r,\label{eq:delta-r}\\
   \log_2 r &\geq& d\log_2 (e^2).\label{eq:log2r}
\end{eqnarray}
Construct an $(n\times r)$-matrix $\mv M_r$ whose columns are in $\prod_{i=1}^n \Dom(A_i)$
where
column $\mv a_{[n]}$ appears precisely $r \cdot p(\mv a_{[n]})$ times.

The following construction of a group system is from
Chan and Yeung~\cite{DBLP:journals/tit/ChanY02}.
Let $G^r$ denote the group of permutations of columns of $\mv M_r$.
This group acts on the rows of $\mv M_r$.
For $i\in [n]$, let $G^r_i$ be the subgroup of $G^r$ that fixes the $i$th row of
$\mv M_r$, i.e $G^r$ is the stabilizer subgroup of $G^r$ with respect to the
$i$th row of $\mv M_r$.
Note that $|G^r_F| = \prod_{\mv a_F} (r\cdot p(\mv a_F))!$, where
the product is over all vectors $\mv a_F$ which have positive probability mass.
Consider the database instance $\mv D^r$ associated with the group system
$(G^r,G^r_1,\dots,G^r_n)$. We claim that this database instance satisfies all
degree constraints $\hdc\times k$ for scale factor $k=r(1+\delta)$.

First, we verify that $\mv D^r$ satisfies all the functional dependency
constraints (i.e. the degree upperbounds of $1$).
Let $Z \to Y$ be an FD. Then, $H[\mv A_Y | \mv A_Z]=0$, which means for every
tuple $\mv a_Z$ with positive probability mass, the tuple $\mv a_{Y-X}$ is
completely determined. In particular, there cannot be two vectors
$\mv a_Y \neq \mv a'_Y$ with positive probability mass for which
$\Pi_X(\mv a_Y) = \Pi_X(\mv a'_Y)$.
Thus, if a permutation fixes all the rows in $Z$, then it also fixes all the
rows in $Y-Z$ of the matrix $\mv M_r$.
This means $G_Z$ is a subgroup of $\bigcap_{i\in Y-Z}G_i$.
From Lemma~\ref{lmm:group:degree}, the associated database instance
satisfies the FD.

Second, we verify the higher-order degree constraints (with degree bounds
$n_{Y|Z} = \log_2 N_{Y|Z} \geq 1$).
From definition of probability mass and entropy, we know that
\begin{eqnarray*}
   \sum_{\mv a_F} p(a_F)&=&1,\\
   -\sum_{\mv a_F} p(\mv a_F)\log_2 p(\mv a_F)&=&H[\mv A_F].
\end{eqnarray*}
In the above and henceforth in this section, $\sum_{\mv a_F}$ denotes the sum
over all tuples $\mv a_F$ with positive probabilities $p(\mv a_F) > 0$.
Since $d$ is a common denominator of all positive probabilities $p(\mv
a_{[n]}) > 0$ and those probabilities sum up to $1$,
$d$ is an upper bound on the number of tuples $\mv a_{[n]}$ having $p(\mv
a_{[n]}) > 0$.
Hence for any $F \subseteq[n]$, $d$ is also an upper bound on the number of tuples $\mv a_F$
with $p(\mv a_F) > 0$:
\begin{eqnarray}
   \sum_{\mv a_F} 1 &\leq& d,\label{eq:count-a_F-d}\\
   \sum_{\mv a_F} \log_2 \left( \frac{1}{p(\mv a_F)} \right) &\leq&
   \sum_{\mv a_F} \log_2 d \leq
   d \cdot \log_2 d.\label{eq:d-log-d}
\end{eqnarray}
The first inequality in~\eqref{eq:d-log-d} relies on the monotonicity of the $\log$ function
(since $p(\mv a_F) \geq 1/d$ for every positive $p(\mv a_F)$, which follows from $d$ being the
common denominator).
We will also use Stirling approximation:
\begin{eqnarray*}
\log_2 n! &\leq&
\log_2(e)+n\log_2 n+\frac 1 2 \log_2 n-n\log_2 e\\
\log_2 n! &\geq&
\log_2(\sqrt{2\pi})+n\log_2 n+\frac 1 2 \log_2 n-n\log_2 e
\end{eqnarray*}
From the above, we have
\begin{eqnarray}
\log_2 |G^r_F|
   &=& \log_2 \prod_{ \mathbf a_F} (r \cdot p(\mv a_F))!\nonumber\\
   &\leq & \sum_{\mv a_F} \left(
      \log_2(e)+
         r \cdot p(\mv a_F) \log_2 \left(r \cdot
      p(\mv a_F)\right)
      + \frac 1 2 \log_2( r \cdot p(\mv a_F))
      - r \cdot p(\mv a_F) \log_2 e
      \right)\nonumber\\
   &=& r\log_2(r/e)-rH[\mv A_F]+ \sum_{\mv a_F}\left( \log_2 e+\frac
   1 2 \log_2 (r\cdot p(\mv a_F)) \right).
   \label{eq:stirling:up}
\end{eqnarray}
Similarly,
\begin{eqnarray}
   \log_2 |G^r_F| &\geq&
   r\log_2(r/e)-rH[\mv A_F]+ \sum_{\mv a_F}\left( \log_2 \sqrt{2\pi}+\frac
   1 2 \log_2 (r\cdot p(\mv a_F)) \right).
   \label{eq:stirling:down}
\end{eqnarray}
We now use Lemma~\ref{lmm:group:degree} alongside the above bounds to bound
\begin{eqnarray*}
\log_2 \frac{|G^r_Z|}{|G^r_Y|}
&=& \log_2
      \frac{\prod_{ \mathbf a_Z} (r \cdot \pr\bigl[\mathbf A_Z = \mathbf a_Z\bigr])!}
      { \prod_{ \mathbf a_Y} (r \cdot \pr\bigl[\mathbf A_Y = \mathbf a_Y\bigr])! }\\
&\leq &
   \left[r\log_2(r/e)-rH[\mv A_Z]+ \sum_{\mv a_Z}\left( \log_2 e+\frac
   1 2 \log_2 (r\cdot p(\mv a_Z)) \right)\right]\\
   &&
- \left[
   r\log_2(r/e)-rH[\mv A_Y]+ \sum_{\mv a_Y}\left( \log_2 \sqrt{2\pi}+\frac
   1 2 \log_2 (r\cdot p(\mv a_Y)) \right) \right]\\
   &=&
   rH[\mv A_Y | \mv A_Z] + \sum_{\mv a_Z}\left( \log_2 e+\frac
      1 2 \log_2 (r\cdot p(\mv a_Z)) \right)-
      \sum_{\mv a_Y}\left( \log_2 \sqrt{2\pi}+\frac
         1 2 \log_2 (r\cdot p(\mv a_Y)) \right)\\
   &=&
   rH[\mv A_Y | \mv A_Z] +
   \underbrace{\sum_{\mv a_Z}\log_2 e}_{\text{$\leq d\log_2 e$ by \eqref{eq:count-a_F-d}}}
   +\underbrace{\frac{1}{2} \sum_{\mv a_Z} \log_2 r
   -\frac{1}{2}\sum_{\mv a_Y}\log_2 r}_{\text{$\leq 0$ since $Z \subseteq Y$}}
   +\frac{1}{2} \sum_{\mv a_Z} \underbrace{\log_2 p(\mv a_Z)}_{\leq 0}\\
   &&\underbrace{-\sum_{\mv a_Y}\log_2 \sqrt{2\pi}}_{\leq 0}
   -\frac{1}{2}\sum_{\mv a_Y}\log_2 p(\mv a_Y)\\
   &\leq&
   rH[\mv A_Y | \mv A_Z] + d\log_2 e + \frac 1 2 \sum_{\mv a_Y} \log_2 \frac{1}{p(\mv
   a_Y)}\\
   \text{(by~\eqref{eq:d-log-d})}&\leq& rH[\mv A_Y | \mv A_Z] + d\log_2 e + \frac d 2 \log_2 d\\
   \text{(by~\eqref{eq:delta-r})}&\leq& rH[\mv A_Y | \mv A_Z] + \delta r\\
   &\leq& rn_{Y|Z}+ \delta r\\
   &\leq& r(1+\delta)n_{Y|Z}\\
   &=& k \cdot n_{Y|Z}.
\end{eqnarray*}
The inequality $H[\mv A_Y | \mv A_Z] \leq n_{Y|Z}$ used above follows from the fact that
$H$ is in $\hdc$ according to~\eqref{eq:entropic:tight:H}.
Third, let $\mv T = (T_B)_{B\in\calB}$ be any model of the rule
$P$ on the database instance $\mv D$ satisfying the scaled degree constraints
$\hdc \times k$.
We are to show that
\begin{equation}
   \max_{B\in\calB} \log_2 |T_B| \geq (1-\epsilon) \cdot
\outputbound_{\overline \Gamma^*_n\cap \hdc \times k}(P)  =
(1-\epsilon) \cdot k \cdot \min_{B\in\calB} H[\mv A_{B}],
   \label{eqn:kH:toprove}
\end{equation}
where the equality above follows from~\eqref{eq:entropic:tight:H}.
Let $Q$ denote the set of tuples satisfying the body of the rule
(i.e. $Q$ is the join of all atoms in the body).
To bound $\max_{B\in\calB}|T_B|$, we reason as follows.
First, without loss of generality, we can assume $T_B \subseteq R_B$, for all
$B\in\calB$, because tuples in $T_B \setminus R_B$ can be removed while keeping
$\mv T$ a model of $P$.
Second, for any $B\in\calB$, we say that $\mv t_B \in R_B$ ``covers''
a tuple $\mv a \in Q$ whenever $\Pi_B(\mv a) = \mv t_B$.
From Lemma~\ref{lmm:group:degree},
for every $B\in\calB$, and every tuple $\mv t_B \in R_B$,
there are precisely $\frac{|G^r_B|}{|G^r_{[n]}|}$ tuples
$\mv a \in Q$ covered by $\mv t_B$.
The size of $R_B$ is exactly $\frac{|G^r_\emptyset|}{|G^r_B|}$
and the size of $Q$ (i.e. $R_{[n]}$) is exactly
$\frac{|G^r_\emptyset|}{|G^r_{[n]}|}$.
Hence, no two tuples in $R_B$ cover the same tuple in $Q$.
Third, let $\bar B = \argmax_{B\in\calB} \frac{|G^r_B|}{|G^r_{[n]}|}$. Then,
\begin{eqnarray*}
   \frac{|G^r_\emptyset|}{|G^r_{[n]}|} &=&|Q|\\
   (\text{every tuple in $Q$ has to be covered})   &\leq& \sum_{B\in\calB} |T_B| \cdot \frac{|G^r_B|}{|G^r_{[n]}|}\\
   &\leq& \left(\max_{B\in\calB} |T_B|\right) \cdot \left( \sum_{B\in\calB} \frac{|G^r_B|}{|G^r_{[n]}|}\right)\\
   &\leq& \left( \max_{B\in\calB}|T_B| \right) \cdot \left( |\calB| \cdot \max_{B\in\calB} \frac{|G^r_B|}{|G^r_{[n]}|}\right)\\
   &=& |\calB| \left( \max_{B\in\calB}|T_B| \right) \cdot \left(
   \frac{|G^r_{\bar B}|}{|G^r_{[n]}|}\right).
\end{eqnarray*}
Hence,
\[ \max_{B\in\calB}|T_B| \geq \frac{1}{|\calB|}
\frac{|G^r_\emptyset|}{|G^r_{\bar B}|} = \frac{1}{|\calB|} |R_{\bar B}|.
\]
To complete the proof of~\eqref{eqn:kH:toprove}, it suffices to show that
\[ \log_2 \frac{|R_{\bar B}|}{|\calB|} \geq (1-\epsilon) \cdot
k \cdot \min_{B\in\calB} H[\mv A_{B}],
\]
which would follow immediately from showing that
\[ \log_2|R_{\bar B}| -\log_2|\calB| \geq (1-\epsilon) \cdot
k \cdot H[\mv A_{\bar B}].
\]

Recall that we assumed $\outputbound_{\overline \Gamma^*_n\cap \hdc}(P) \geq 1$.
From~\eqref{eq:entropic:tight:H}
and the fact that $\bar B \in \calB$, we have  $H[\mv A_{\bar B}] \geq 1$.
For a sufficiently large $r$, we have
\begin{equation}
   \frac 1 2 d \log_2 r + \log_2|\calB| \leq \delta r H[\mv A_{\bar B}]
   \label{eqn:large:r}
\end{equation}
Now we have,
\begin{eqnarray*}
   \log_2 |R_{\bar B}|-\log_2|\calB|
   &=& \log_2 \frac{r!}{|G^r_{\bar B}|}-\log_2|\calB|\\
   (\text{by~\eqref{eq:stirling:up} and~\eqref{eq:stirling:down}}) &\geq &
   \left[ r\log_2(r/e)+ \frac 1 2 \log_2 r+\log_2 \sqrt{2\pi} \right]\\
   &&
- \left[
   r\log_2(r/e)-rH[\mv A_{\bar B}]+ \sum_{\mv a_{\bar B}} \left( \log_2 e+\frac
   1 2 \log_2 (r\cdot p(\mv a_{\bar B})) \right)
   \right]-\log_2|\calB|\\
   &=&
   rH[\mv A_{\bar B}] + \log_2 \sqrt{2\pi}+\frac 1 2 \log_2 r - \frac{1}{2}\sum_{\mv
   a_{\bar B}} \log_2
   (re^2 p(\mv a_{\bar B}))-\log_2|\calB|\\
   &=&
   rH[\mv A_{\bar B}] + \log_2 \sqrt{2\pi}+\frac 1 2 \log_2 r
   - \frac{1}{2}\underbrace{\sum_{\mv a_{\bar B}} \log_2 (re^2)}_{\leq d\log_2(re^2)\text{ by~\eqref{eq:count-a_F-d}}}
   - \frac{1}{2}\sum_{\mv a_{\bar B}} \underbrace{\log_2 p(\mv a_{\bar B})}_{\leq 0}
   -\log_2|\calB|\\
   &\geq&
   rH[\mv A_{\bar B}] + \frac 1 2 \log_2 r - \frac 1 2 d \log_2
   (re^2)-\log_2|\calB|\\
&=& r(1+\delta)H[\mv A_{\bar B}] -\delta r H[\mv A_{\bar B}] - \frac 1 2 d
   \log_2 (r) + \frac 1 2 (\log_2 r - d\log_2(e^2))-\log_2|\calB|\\
(\text{by~\eqref{eq:log2r}})&\geq& r(1+\delta)H[\mv A_{\bar B}] -\delta r H[\mv A_{\bar B}] - \frac 1 2 d
   \log_2 (r)
   -\log_2|\calB|\\
   (\text{due to~\eqref{eqn:large:r}})   &\geq& r(1+\delta)H[\mv A_{\bar B}] - 2\delta r H[\mv A_{\bar B}]\\
   &=& \left( 1 - \frac{2\delta}{1+\delta}\right) r(1+\delta)H[\mv
   A_{\bar B}].\\
   &=& (1-\epsilon)  \cdot k \cdot H[\mv A_{\bar B}].
\end{eqnarray*}
\ep

\subsection{The polymatroid bound for disjunctive datalog is not tight}
\label{subsec:polymatroid:bound:ddl}

\blmm[Part $(iii)$ of Theorem~\ref{thm:intro:size:bound:disjunctive}]
The polymatroid bound for a disjunctive datalog rule is not tight, even if the
input constraints are only cardinality constraints, and all cardinality
upperbounds are identical.
\elmm
\bp
We first give a short proof for the case when the upperbounds vary, to
illustrate the main idea.
The non-Shannon inequality~\eqref{eqn:zy:non:shannon:h}
along with the following three submodularity inequalities
\begin{align*}
   h(AC) &\leq h(A)+h(C)
   && h(XC) \leq h(X)+h(C)
   && h(YC) \leq h(Y)+h(C)
\end{align*}
imply that the following hold for all entropic functions $h$ on five variables
$X,Y,A,B,C$,
\begin{align*}
   & h(AB)+4h(AXY) +h(BXY) \\
   & \leq 3h(XY)+3h(AX) +3h(AY)+h(BX)+h(BY) - h(A)- 2h(X)-2h(Y)\\
   & \leq 3h(XY)+3h(AX) +3h(AY)+h(BX)+h(BY) - h(AC)- 2h(XC)-2h(YC)+5h(C).
\end{align*}
Moving all negative terms to the left hand side, we have equivalently
\begin{multline}
h(AB)+4h(AXY) +h(BXY) + h(AC)+ 2h(XC)+2h(YC)\\
\leq 3h(XY)+3h(AX) +3h(AY)+h(BX)+h(BY)+5h(C).
   \label{eqn:simple:gap:bound}
\end{multline}
Now, consider the disjunctive datalog rule
\begin{align*}
   P:\ \ \ & T_1(AB)\vee T_2(AXY)\vee T_3(BXY)\vee T_4(AC) \vee T_5(XC)\vee T_6(YC)\\
   &  \qquad \qquad \cd
   R_1(XY)\wedge R_2(AX)\wedge R_3(AY) \wedge R_4(BX)\wedge R_5(BY)\wedge R_6(C),
\end{align*}
with the following cardinality bounds: $|R_1|,|R_2|,|R_3|,|R_4|, |R_5|\leq N^3$
and $|R_6|\leq N^2$. From~\eqref{eqn:simple:gap:bound}, the entropic bound for
the disjunctive datalog rule $P$ above is upper-bounded by
\begin{align}
   \outputbound_{\overline\Gamma^*_5 \cap \hcc}(P)
   &= \max_{h \in \overline\Gamma^*_5 \cap \hcc} \min
   \{ h(AB),h(AXY),h(BXY),h(AC),h(XC),h(YC) \} \nonumber\\
   &\leq \max_{h \in \overline\Gamma^*_5 \cap \hcc}\frac{1}{11}\bigl(h(AB)+4h(AXY) +h(BXY) + h(AC)+ 2h(XC)+2h(YC)\bigr)
   \nonumber\\
   &\leq \max_{h \in \overline\Gamma^*_5 \cap \hcc}\frac{1}{11}\bigl(3h(XY)+3h(AX) +3h(AY)+h(BX)+h(BY)+5h(C)\bigr)
   \nonumber\\
   &\leq \frac{1}{11}\bigl(11 \log N^3+5\log N^2\bigr) \nonumber\\
   &= \frac{43}{11} \log N.\label{eqn:43:11}
\end{align}
On the other hand, consider the polymatroid $\hat h = \log N \cdot h$, where $h$
is the polymatroid shown in Figure~\ref{fig:yeung}. (This is the same $\hat h$ as
in the proof of Theorem~\ref{thm:glvv:not:tight}.) It is easy to check that
$\hat h$ satisfies all the cardinality constraints, and
$\hat h(AB) =\hat h(AXY) =\hat h(BXY) =\hat h(AC) =\hat h(XC) =\hat h(YC) =
4\log N.$
Hence, $\outputbound_{\Gamma_5 \cap \hcc}(P) \geq 4\log N > \frac{43}{11}\log N$.

Next, we prove the stronger statement by using the fact that
inequality~\eqref{eqn:simple:gap:bound} holds for any entropic function $h$ on
$5$ variables $\{A,B,X,Y,C\}$.
Now, consider an entropic function $h$ on $8$ variables
$A',B',X',Y',A,B,X,Y$; then the restriction of $h$ on the $5$-variable sets
$\{A',B',X',Y',A\}$,
$\{A',B',X',Y',X\}$,
and
$\{A',B',X',Y',Y\}$
all satisfy~\eqref{eqn:simple:gap:bound}. In particular, we have the following
inequalities:
\begin{eqnarray}
   & & h(A'B')+4h(A'X'Y') +h(B'X'Y') + h(A'A)+ 2h(X'A)+2h(Y'A) \nonumber \\
   & \leq & 3h(X'Y')+3h(A'X') +3h(A'Y')+h(B'X')+h(B'Y')+5h(A), \label{eqn:49}\\
   & & h(A'B')+4h(A'X'Y') +h(B'X'Y') + h(A'X)+ 2h(X'X)+2h(Y'X) \nonumber \\
   & \leq & 3h(X'Y')+3h(A'X') +3h(A'Y')+h(B'X')+h(B'Y')+5h(X), \label{eqn:50}\\
   & & h(A'B')+4h(A'X'Y') +h(B'X'Y') + h(A'Y)+ 2h(X'Y)+2h(Y'Y) \nonumber \\
   & \leq & 3h(X'Y')+3h(A'X') +3h(A'Y')+h(B'X')+h(B'Y')+5h(Y).\label{eqn:51}
\end{eqnarray}
Add $5 \times$ inequality~\eqref{eqn:zy:non:shannon:h} with $1 \times$
inequality~\eqref{eqn:49},
$2 \times$ inequality~\eqref{eqn:50}, and $2 \times$ inequality~\eqref{eqn:51},
we obtain the following non-Shannon inequality:
\begin{eqnarray*}
   & & 5\bigl[ h(AB)+4h(AXY) +h(BXY) \bigr] \\
   &+& \bigl[ h(A'B')+4h(A'X'Y') +h(B'X'Y') + h(A'A)+ 2h(X'A)+2h(Y'A) \bigr] \\
   &+& 2\bigl[ h(A'B')+4h(A'X'Y') +h(B'X'Y') + h(A'X)+ 2h(X'X)+2h(Y'X) \bigr]\\
   &+& 2 \bigl[h(A'B')+4h(A'X'Y') +h(B'X'Y') + h(A'Y)+ 2h(X'Y)+2h(Y'Y) \bigr]\\
   & \leq & 5 \bigl[ 3h(XY)+3h(AX) +3h(AY)+h(BX)+h(BY) \bigr] \\
   &+& \bigl[ 3h(X'Y')+3h(A'X') +3h(A'Y')+h(B'X')+h(B'Y')) \bigr]\\
   &+& 2 \bigl[3h(X'Y')+3h(A'X') +3h(A'Y')+h(B'X')+h(B'Y') \bigr]\\
   &+& 2 \bigl[3h(X'Y')+3h(A'X') +3h(A'Y')+h(B'X')+h(B'Y') \bigr].
\end{eqnarray*}
The inequality simplifies to
\begin{multline}
   5 \bigl[ h(AB)+4h(AXY) +h(BXY) + h(A'B')+4h(A'X'Y') +h(B'X'Y') \bigr]\\
   +  h(A'A)+ 2h(X'A)+2h(Y'A) \\
   + 2h(A'X)+ 4h(X'X)+4h(Y'X)\\
   + 2h(A'Y)+ 4h(X'Y)+4h(Y'Y) \\
   \leq 5 \bigl[ 3h(XY)+3h(AX) +3h(AY)+h(BX)+h(BY) + \\
    3h(X'Y')+3h(A'X') +3h(A'Y')+h(B'X')+h(B'Y') \bigr]
    \label{eqn:non:shannon:big:ineq}
\end{multline}
Now, consider the disjunctive datalog rule $P$:
\begin{multline}
   T_1(AB) \vee T_2(AXY) \vee T_3(BXY) \vee T_4(A'B')
   \vee T_5(A'X'Y') \vee T_6(B'X'Y') \\
   \vee T_7(A'A) \vee T_8(X'A) \vee T_9(Y'A) \\
   \vee T_{10}(A'X) \vee T_{11}(X'X) \vee T_{12}(Y'X)\\
   \vee T_{13}(A'Y) \vee T_{14}(X'Y) \vee T_{15}(Y'Y) \\
   \cd R_1(XY) \wedge R_2(AX)  \wedge R_3(AY) \wedge R_4(BX) \wedge
   R_5(BY) \\
   \wedge
    R_6(X'Y') \wedge R_7(A'X')  \wedge R_8(A'Y') \wedge R_9(B'X') \wedge R_{10}(B'Y'),
\end{multline}
with a uniform cardinality bound $|R_i| \leq N^3$ for all $i \in [10]$.
Using the same averaging trick we used in~\eqref{eqn:43:11}, from
inequality~\eqref{eqn:non:shannon:big:ineq} we get:
\[ \outputbound_{\overline\Gamma^*_8 \cap \hcc}(P)
\leq \frac{110}{85} \log N^3 = \frac{330}{85}{\log N} < 4 \log N.
\]
\colorlet{myGreen}{green!60!black}
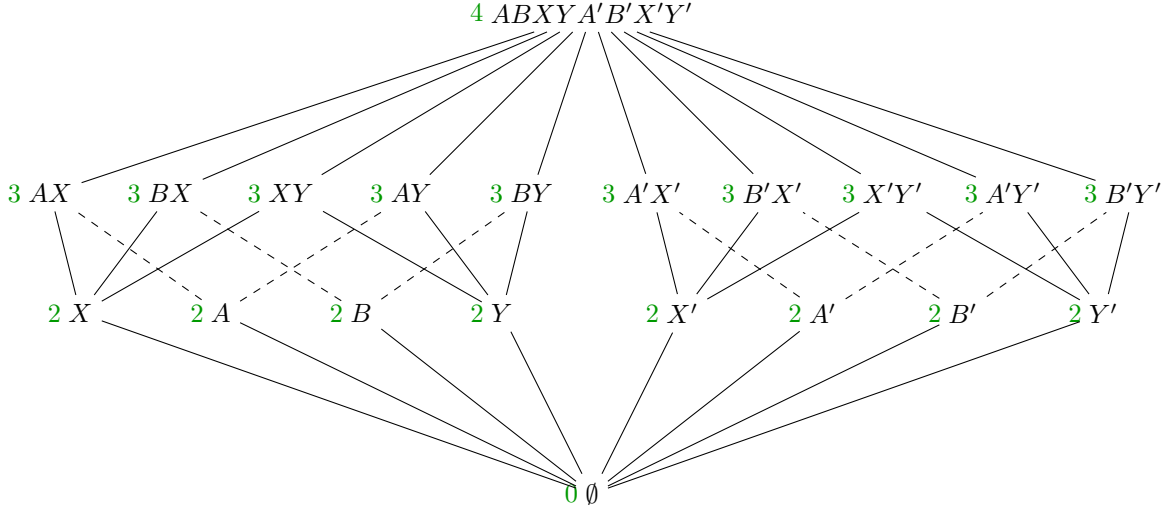
\begin{figure}[th!]
\centering
   \begin{tikzpicture}[scale=0.6]
\node[] at (9, 1) (top){$ABXYA'B'X'Y'$}; \node[myGreen, left=-.15 of top] {4};

\node[] at (0, -2) (AX){$AX$}; \node[myGreen, left=-.15 of AX] {3};
\node[] at (2, -2) (BX){$BX$}; \node[myGreen, left=-.15 of BX] {3};
\node[] at (4, -2) (XY){$XY$}; \node[myGreen, left=-.15 of XY] {3};
\node[] at (6, -2) (AY){$AY$}; \node[myGreen, left=-.15 of AY] {3};
\node[] at (8, -2) (BY){$BY$}; \node[myGreen, left=-.15 of BY] {3};

\node[] at (.5, -4) (X){$X$}; \node[myGreen, left=-.15 of X] {2};
\node[] at (.5+7/3, -4) (A){$A$}; \node[myGreen, left=-.15 of A] {2};
\node[] at (.5+14/3, -4) (B){$B$}; \node[myGreen, left=-.15 of B] {2};
\node[] at (7.5, -4) (Y){$Y$}; \node[myGreen, left=-.15 of Y] {2};

\node[] at (9, -7) (0){$\emptyset$}; \node[myGreen, left=-.15 of 0] {0};

\path[] (0) edge (X);
\path[] (0) edge (A);
\path[] (0) edge (B);
\path[] (0) edge (Y);

\path[] (X) edge (AX);
\path[] (X) edge (BX);
\path[] (X) edge (XY);
\path[] (Y) edge (BY);
\path[] (Y) edge (AY);
\path[] (Y) edge (XY);
\path[] (A) edge (AX);
\path[] (A) edge (AY);
\path[] (B) edge (BX);
\path[] (B) edge (BY);

\path[] (AX) edge (top);
\path[] (BX) edge (top);
\path[] (XY) edge (top);
\path[] (AY) edge (top);
\path[] (BY) edge (top);

\node[] at (10, -2) (AX2){$A'X'$}; \node[myGreen, left=-.15 of AX2] {3};
\node[] at (12, -2) (BX2){$B'X'$}; \node[myGreen, left=-.15 of BX2] {3};
\node[] at (14, -2) (XY2){$X'Y'$}; \node[myGreen, left=-.15 of XY2] {3};
\node[] at (16, -2) (AY2){$A'Y'$}; \node[myGreen, left=-.15 of AY2] {3};
\node[] at (18, -2) (BY2){$B'Y'$}; \node[myGreen, left=-.15 of BY2] {3};

\node[] at (10.5, -4)      (X2){$X'$}; \node[myGreen, left=-.15 of X2] {2};
\node[] at (10.5+7/3, -4)  (A2){$A'$}; \node[myGreen, left=-.15 of A2] {2};
\node[] at (10.5+14/3, -4) (B2){$B'$}; \node[myGreen, left=-.15 of B2] {2};
\node[] at (17.5, -4)     (Y2){$Y'$}; \node[myGreen, left=-.15 of Y2] {2};

\path[] (0) edge (X2);
\path[] (0) edge (A2);
\path[] (0) edge (B2);
\path[] (0) edge (Y2);

\path[] (X2) edge (AX2);
\path[] (X2) edge (BX2);
\path[] (X2) edge (XY2);
\path[] (Y2) edge (BY2);
\path[] (Y2) edge (AY2);
\path[] (Y2) edge (XY2);
\path[] (A2) edge (AX2);
\path[] (A2) edge (AY2);
\path[] (B2) edge (BX2);
\path[] (B2) edge (BY2);

\path[] (AX2) edge (top);
\path[] (BX2) edge (top);
\path[] (XY2) edge (top);
\path[] (AY2) edge (top);
\path[] (BY2) edge (top);
\end{tikzpicture}
\caption{A polymatroid $h$ (shown in {\color{myGreen}green})
over
$8$ variables $\{A,B,X,Y,A',B',X',Y'\}$, where for any
set whose value is not shown, $h$ takes the same value as the
smallest set in the figure containing it.}
\Description{A polymatroid $h$ (shown in {\color{myGreen}green})
over
$8$ variables $\{A,B,X,Y,A',B',X',Y'\}$, where for any
set whose value is not shown, $h$ takes the same value as the
smallest set in the figure containing it.}
\label{fig:yeung2}
\end{figure}
On the other hand, consider the function $h$ shown in Figure~\ref{fig:yeung2}.
This is a set function on $8$ variables $\{A,B,X,Y,A',B',X',Y'\}$, where for any
set whose (green) value is not shown, $h$ takes on the same value as the
smallest set in the figure containing it.
It is easy to check that this is a polymatroid, and the polymatroid $\hat h =
\log N \cdot h$ satisfies all input cardinality constraints; furthermore,
$\hat h(AB) = \hat h(AXY) = \cdots = \hat h(Y'Y) = 4 \log N$. Hence,
$\outputbound_{\Gamma_8 \cap \hcc}(P) \geq 4\log N$.

\ep

\section{Shannon flow inequalities}
\label{sec:ps}

The $\panda$ algorithm is built on the notion of a ``proof sequence'' for a
class of Shannon-type inequalities called the {\em Shannon flow inequalities}.

\bdefn
Let $\calB \subseteq 2^{\calV}$ denote a collection of nonempty subsets of $\calV$.
Let $\calC \subseteq 2^{\calV} \times 2^{\calV}$ denote a collection of
pairs $(X,Y)$ such that $\emptyset \neq X \subset Y \subseteq \calV$.
Let $\vec\lambda_\calB = (\lambda_F)_{F\in\calB} \in \Q_+^\calB$
and
$\vec\delta_\calC = (\delta_{Y|X})_{(X,Y) \in\calC} \in \Q_+^\calC$
denote two vectors of non-negative rationals.
For any polymatroid $h$, let $h(Y|X)$ denote $h(Y)-h(X)$.\footnote{The
quantity $h(Y|X)$ is the polymatroid-analog of the conditional entropy
$H[Y|X] = H[Y] - H[X]$.}
If the following inequality
\begin{equation}
\sum_{B\in \calB} \lambda_B \cdot h(B)
\leq
\sum_{(X,Y) \in \calC} \delta_{Y|X} \cdot h(Y|X)
\defeq
\sum_{(X,Y) \in \calC} \delta_{Y|X} \cdot (h(Y)-h(X))
\label{eqn:sfi}
\end{equation}
holds for all $h \in \Gamma_n$ (i.e. for all polymatroids), then it is called
a {\em Shannon flow inequality}.
The set $\calB$ is called the set of {\em targets} of the flow inequality.
\edefn
For a simple illustration, the following shows a Shannon-flow
inequality that we proved in Example~\ref{ex:intro:2b}:
\begin{align*}
  h(A_1A_2A_3) + h(A_2A_3A_4) \leq h(A_1A_2) + h(A_2A_3) + h(A_3A_4)
\end{align*}
Section~\ref{subsec:sfi:motivations} motivates the study of these
inequalities. Section~\ref{subsec:sfi:ps} explains why they are called ``flow''
inequalities.

\subsection{Motivations}
\label{subsec:sfi:motivations}

Fix a disjunctive datalog rule $P$ of the
form~\eqref{eqn:disjunctive:datalog:query}
with degree constraints $\dc$.
Abusing notations, we write $(X,Y) \in \dc$ whenever
$(X,Y,N_{Y|X}) \in \dc$.
In particular the set $\dc$ can play the role of the generic set $\calC$
in the definition of Shannon flow inequality.
To explain where the Shannon flow inequalities come from, we study
the ($\log$) polymatroid bound~\eqref{eqn:intro:polymatroid:disjunctive:2}
for $P$, which was defined by~\eqref{eqn:intro:output:bound} with $\calF$
chosen to be $\Gamma_n \cap \hdc$.
The bound $\outputbound_{\Gamma_n \cap \hdc}(P)$ is the optimal objective value of
the following optimization problem:
\begin{align}
   \max & &\min_{B\in\calB} h(B) \label{eqn:ddl:target:first:version}\\
   \text{such that} & &h(Y)-h(X)                     &\leq n_{Y|X}, & (X,Y,N_{Y|X}) \in \dc \nonumber\\
                    & &h(I\cup J | J) - h(I|I\cap J) &\leq 0,       & I \incomp J \nonumber\\
                    & &h(X) - h(Y)                   &\leq 0,       &\emptyset \neq X \subset Y \subseteq \calV \nonumber\\
                    & &h(Z)                          & \geq 0,      & \emptyset \neq Z \subseteq \calV.  \nonumber
\end{align}
(Recall that implicitly we have $h(\emptyset)=0$, and that $n_{Y|X}\defeq \log_2 N_{Y|X}$.)
Here, $I \incomp J$ means $I \not\subseteq J$ and $J \not\subseteq I$.
The optimization problem above is not easy to handle.
Lemma~\ref{lmm:lambda:1:reformulation} below shows that we can
reformulate the above maximin optimization problem into a linear program:

\blmm\label{lmm:lambda:1:reformulation}
There exists a non-negative vector
$\vec\lambda = (\lambda_B)_{B\in\calB}$, with
$\norm{\vec\lambda}_1=1$, such that
\begin{equation}
\outputbound_{\Gamma_n \cap \hdc}(P) =
   \max_{h \in \Gamma_n \cap \hdc} \sum_{B\in\calB}\lambda_B \cdot h(B).
\label{eqn:ddl:target}
\end{equation}
%
\elmm

Instead of proving~\ref{lmm:lambda:1:reformulation} directly, we prove a
slightly more general lemma:
\blmm[A generalization of Lemma~\ref{lmm:lambda:1:reformulation}]
\label{lmm:rewrite}
Let $\mv A \in \Q^{\ell\times m}, \mv b \in \R^\ell$, and $\mv C \in \Q_+^{m\times
p}$ be a matrix with columns $\mv c_1,\dots,\mv c_p$. Consider the following
maximin optimization problem:
\begin{equation}
   \max \{ \min_{k\in [p]} \mv c^\top_k \mv x \suchthat
   \mv A\mv x \leq \mv b, \mv x \geq \mv 0\}
   \label{eqn:LP1}
\end{equation}
If problem~\eqref{eqn:LP1}'s objective value is positive and bounded,
then there exists a
vector $\vec\lambda \in \Q_+^p$ satisfying the following conditions:
\bi
 \item[(a)] $\norm{\vec\lambda}_1 = 1$.
 \item[(b)] The problem~\eqref{eqn:LP1} has the same optimal objective value
as the following linear program:
\begin{equation}
   \max \{ (\mv C \vec\lambda)^\top \mv x \suchthat
   \mv A\mv x \leq \mv b, \mv x \geq \mv 0\}  \label{eqn:LP2}
\end{equation}
\ei
\elmm
Note that $\outputbound_{\Gamma_n \cap \hdc}(P)$ formulated
in~\eqref{eqn:ddl:target:first:version}
has the same form as the optimization problem in~\eqref{eqn:LP1},
thus making Lemma~\ref{lmm:lambda:1:reformulation} a special case of
Lemma~\ref{lmm:rewrite}. Next we prove the latter.
\bp[Proof of Lemma~\ref{lmm:rewrite}]
Let $\mv 1_p \in \R^p$ be the all-$1$ vector.
We first reformulate problem~\eqref{eqn:LP1} with an equivalent LP:
\begin{equation}
   \max \{ w \suchthat
   \mv A\mv x \leq \mv b, \mv 1_pw - \mv C^\top \mv x \leq \mv 0,
       \quad \mv x \geq \mv 0, w \geq 0\}
   \label{eqn:LP3}
\end{equation}
The dual of~\eqref{eqn:LP3} is
\begin{equation}
   \min \{ \mv b^\top\mv y \suchthat
   \mv A^\top\mv y \geq \mv C \mv z,
   \mv 1_p^\top\mv z \geq 1,
   \quad \mv y \geq \mv 0, \mv z \geq 0\}
   \label{eqn:LP4}
\end{equation}
Let $(w^*, \mv x^*)$ and $(\mv z^*,\mv y^*)$ be a pair of primal-optimal and
dual-optimal solutions to~\eqref{eqn:LP3} and \eqref{eqn:LP4} respectively.
Due to complementary slackness of the \eqref{eqn:LP3} and \eqref{eqn:LP4} pair,
$w^*>0$ implies
$\mv 1_p^\top\mv z^*=1$
and $(\mv 1_p w^*-\mv C^\top \mv x^*)^\top \mv z^*=0$.
We know $w^*>0$ because problem~\eqref{eqn:LP1} has a positive optimal value.
It follows that
$\norm{\mv z^*}_1=1$, and
$w^* = (\mv C^\top\mv x^*)^\top \mv z^*$.

Next, we show $\vec\lambda=\mv z^*$ satisfies $(a)$
and $(b)$. Condition $(a)$ follows from $\norm{\mv z^*}_1=1$. To show $(b)$,
note that $\mv x^*$ is a feasible solution to~\eqref{eqn:LP2} with objective
value $(\mv C \vec\lambda)^\top \mv x^*
= (\mv C \mv z^*)^\top \mv x^* = (\mv C^\top\mv x^*)^\top \mv z^* = w^*$.
Furthermore, for any $\mv x$ feasible to~\eqref{eqn:LP2}, we have
$(\mv C \vec\lambda)^\top \mv x
= (\mv C \mv z^*)^\top \mv x
\leq (\mv A^\top \mv y^*)^\top \mv x
= (\mv A \mv x )^\top \mv y^*
\leq \mv b^\top \mv y^* = w^*$,
where the first inequality follows from $(\mv z^*,\mv y^*)$ being feasible to
\eqref{eqn:LP4} along with $\mv x \geq \mv 0$ from~\eqref{eqn:LP2},
and the second inequality follows from $\mv A\mv x \leq \mv b$ in~\eqref{eqn:LP2}.
\ep

Along with Farkas's lemma~\cite{MR88m:90090} (re-stated in the proof of
Proposition~\ref{prop:sfi:ddl:target} below),
the linear program (LP) on the right hand side
of~\eqref{eqn:ddl:target} gives rise to Shannon flow inequalities. We first
need the dual LP of~\eqref{eqn:ddl:target}.  Associate a dual variable
$\delta_{Y|X}$ to each degree constraint, a variable $\sigma_{I,J}$ to each
submodularity constraint, and a variable $\mu_{X,Y}$ to each monotonicity
constraint.
The dual of the RHS of~\eqref{eqn:ddl:target} is
\begin{align}
   \min        &&\sum_{(X,Y)\in \dc}n_{Y|X} \cdot \delta_{Y|X}
   \label{eqn:dual:ddl:target}\\
   \text{such that} &&\flow(B)          & \geq \lambda_B, & \forall B \in \calB \nonumber\\
               &&\flow(Z)              & \geq 0, & \emptyset \neq Z \subseteq \calV.\nonumber\\
               &&(\vec\delta, \vec\sigma, \vec\mu) &\geq\mv 0.\nonumber
\end{align}
(recall that $n_{Y|X}\defeq\log_2 N_{Y|X}$), where
for any $\emptyset \neq Z \in 2^\calV$, $\flow(Z)$ is defined by
\begin{multline}
   \flow(Z) \defeq
   \sum_{X: (X,Z)\in \dc}\delta_{Z|X}-
   \sum_{Y: (Z,Y)\in \dc}\delta_{Y|Z}+
   \sum_{\substack{I\incomp J\\I\cap J = Z}}\sigma_{I,J}\\
   +
   \sum_{\substack{I\incomp J\\I\cup J = Z}}\sigma_{I,J}
   - \sum_{J: J\incomp Z}\sigma_{Z,J}-
   \sum_{X: X\subset Z}\mu_{X,Z}+ \sum_{Y: Z\subset Y}\mu_{Z,Y}.
   \label{eqn:flow}
\end{multline}
\begin{figure}[!ht]
\centering \begin{tikzpicture}[domain=0:20, 
every node/.style={font=\large}, scale=0.8, every node/.style={scale=0.8}]
\node[circle,draw] at (5,5) (Z) {$Z$};
\node[] at (3,8) (Y) {$Y$};
\node[] at (3,2) (X) {$X$};
\node[] at (5,7) (I) {$I$};
\node[] at (7,6) (J) {$J$};
\node[] at (7,8) (IJ) {$I\cup J$};
\node[] at (5,3) (I') {$I'$};
\node[] at (7,4) (J') {$J'$};
\node[] at (7,2) (IJ') {$I'\cap J'$};
\node[left = .2 of Z] (eq) {$I\cap J = I' \cup J' = Z$};
\node[] at (10,5) (subeq) {$X \subset Z \subset Y$};

\node[color=black!55!green] at (6,6.5) (sij) {$+\sigma_{I,J}$};
\node[color=black!55!green] at (6,3.5) (sij) {$+\sigma_{I',J'}$};

\draw[thick] (Z) -- (Y) node [midway,left=5pt,color=black!55!green] {$+\mu_{Z,Y}$};
\draw[thick] (Z) -- (X) node [midway,left=5pt,color=black!55!green] {$+\delta_{Z|X}$};
\draw (Z) -- (I);
\draw (Z) -- (J);
\draw (Z) -- (J');
\draw (Z) -- (I');
\draw (I) -- (IJ);
\draw (J) -- (IJ);
\draw (I') -- (IJ');
\draw (J') -- (IJ');

\node[circle,draw] at (14,5) (Z1) {$Z$};
\node[] at (13,8) (Y1) {$Y$};
\node[] at (13,2) (X1) {$X$};
\node[] at (16,5) (J1) {$J$};
\node[] at (15,3) (IJ1) {$Z\cap J$};
\node[] at (15,7) (IJ2) {$Z\cup J$};

\node[color=black!10!red] at (15,5) (sij1) {$-\sigma_{Z,J}$};

\draw[thick] (Z1) -- (Y1) node [midway,left=5pt,color=black!10!red] {$-\delta_{Y|Z}$};
\draw[thick] (Z1) -- (X1) node [midway,left=5pt,color=black!10!red] {$-\mu_{X,Z}$};
\draw (Z1) -- (IJ1);
\draw (Z1) -- (IJ2);
\draw (J1) -- (IJ2);
\draw (J1) -- (IJ1);
\end{tikzpicture}
\caption{Contributions of coefficients to $\flow(Z)$.
The left and right parts show the positive and negative contributions respectively.}
\Description{Contributions of coefficients to $\flow(Z)$.
The left and right parts show the positive and negative contributions respectively.}
\label{fig:flow}
\end{figure}
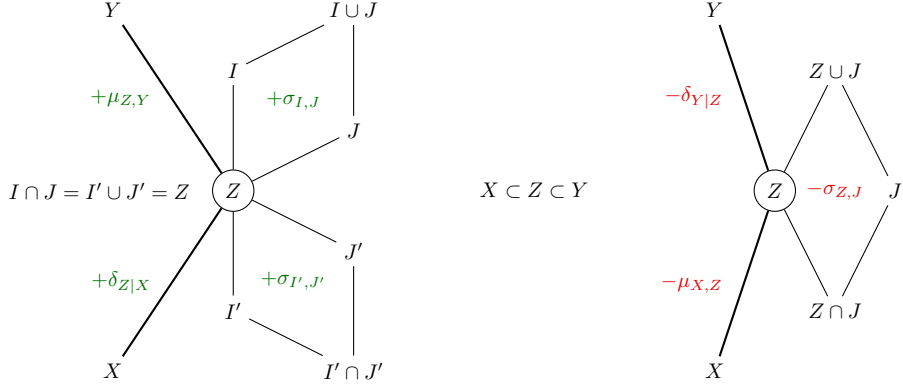
Again, recall that $I \incomp J$ means $I \not\subseteq J$ and $J \not\subseteq I$.
Note that the function $\flow : 2^\calV \to \Q_+$ is also a function of the
dual variables $(\vec\delta,\vec\sigma,\vec\mu)$. However, we do not explicitly
write down this dependency to avoid heavy-loading notations.
Figure \ref{fig:flow} illustrates the contributions of various coefficients to
$\flow(Z)$ for a given set $Z$. It is helpful in the rest of the paper to
keep this picture in mind when we reason about balancing the $\flow$
(in)equalities.

Let $\mv h^* = (h^*_Z)_{Z\subseteq\calV}$ denote an optimal solution
to~\eqref{eqn:ddl:target},
and let $(\vec\delta^*, \vec\sigma^*, \vec\mu^*)$ denote an optimal solution
to~\eqref{eqn:dual:ddl:target}, then
$\sum_{B\in\calB} \lambda_B \cdot h^*(B) = \sum_{(X,Y)\in\dc}\delta^*_{Y|X}
\cdot n_{Y|X}$, by strong duality.
In particular, instead of solving for the primal optimal solution $\mv h^*$, we can look for
the dual optimal solution $(\vec\delta^*, \vec\sigma^*,\vec\mu^*)$.
One way to characterize {\em any}
dual feasible solution $(\vec\delta, \vec\sigma, \vec\mu)$,
is to use Farkas' lemma~\cite{MR88m:90090}, which in our context takes the
following form.
\bprop Given non-negative vectors $\vec\lambda_{\calB}$ and $\vec\delta_{\dc}$, the inequality
\begin{equation}
   \sum_{B\in\calB} \lambda_B \cdot h(B) \leq \sum_{(X,Y)\in\dc}\delta_{Y|X} \cdot h(Y|X)
\label{eqn:sfi:ddl:target}
\end{equation}
is a Shannon flow inequality if and only if there exist $\vec\sigma$ and
$\vec\mu$ such that $(\vec\delta, \vec\sigma, \vec\mu)$ is feasible to the dual
LP~\eqref{eqn:dual:ddl:target}.
\label{prop:sfi:ddl:target}
\eprop
\bp
There are {\em many}
variants of Farkas' lemma~\cite{MR88m:90090}. We use a version whose proof
we also reproduce here because the proof is very short.

Let $\mv A \in \R^{m\times n}$ be a matrix and $\mv c\in\R^n$ be a vector.
Let $P = \{ \mv x \suchthat \mv A\mv x \leq \mv 0, \mv x \geq \mv 0\}$ be a polyhedron
and $D = \{ \mv y \suchthat \mv A^\top\mv y \geq \mv c, \mv y \geq \mv 0\}$ be
the dual polyhedron.
Then, a variant of Farkas' lemma states that
$D$ is non-empty if and only if there is {\em no} $\mv x \in P$ such that
$\mv c^\top\mv x > 0$.
To see this, note that the system
$\{ \mv c^\top\mv x > 0, \mv x \in P\}$ is
infeasible iff $\max \{ \mv c^\top\mv x \suchthat \mv x \in P \} = 0$,
which by strong duality is equivalent to
$\min \{ \mv 0^\top\mv y \suchthat \mv y \in D\}$
is feasible, which is the same as $D$ is non-empty.

Now, to see why the above variant of Farkas' lemma implies
Proposition~\ref{prop:sfi:ddl:target}, we note that~\eqref{eqn:sfi:ddl:target}
holds for all polymatroids iff
$\{ (\vec\lambda_\calB-\vec\delta_{\dc})^\top\mv h > 0 \suchthat \mv h \in
\Gamma_n\}$ is infeasible; now we are in the exact setting of the above variant of
Farkas' lemma and the rest follows trivially.
\ep

Note that inequality~\eqref{eqn:sfi:ddl:target} holds when $\vec\delta = \vec\delta^*$,
in which case the Shannon flow inequality implies the upper bound
\[ \sum_{B\in\calB} \lambda_B \cdot h(B) \leq \sum_{(X,Y)\in\dc}\delta^*_{Y|X}
\cdot n_{Y|X}, \]
where the right-hand side is the optimal objective value of both the primal
and the dual.

\subsection{Proof sequences}
\label{subsec:sfi:ps}

A key observation from our prior work~\cite{csma} was
that we can turn a proof of a special case of inequality~\eqref{eqn:sfi:ddl:target}
 into an
algorithm. The proof has to be performed in a sequential manner; and this
brings us to the concept of a proof sequence.
In this paper, we refine the proof sequence notion from~\cite{csma} in four
significant ways.
First, the definition of the proof sequence is different, allowing for a
simpler algorithm ($\panda$) than $\csma$ in~\cite{csma}.
Second, in~\cite{csma} we left open whether  proof sequences are
a complete proof system, even for special Shannon-flow inequalities; the $\csma$
algorithm used a specific workaround to achieve optimality even
without proving completeness of proof sequences. Our most important
contribution here is to prove completeness of our new proof sequence.
Third, we are able to bound the length of the proof sequence to be
polynomial
in the size of the linear program \eqref{eqn:dual:ddl:target}, as opposed to the doubly exponential length in~\cite{csma}.
Fourth, new technical ideas are introduced so that we can construct proof
sequences for the much more general Shannon flow
inequality~\eqref{eqn:sfi:ddl:target} (as opposed to the special case of ``output
inequality'' in~\cite{csma}, whose form is given in~\eqref{eqn:sfi:V:target}).

\bdefn[Conditional polymatroids]
Let $\calP \subseteq 2^\calV \times 2^\calV$ denote the set of all pairs $(X,Y)$
such that $\emptyset \subseteq X \subset Y \subseteq \calV$.
A vector $\mv f \in \R_+^\calP$ has coordinates indexed by pairs
$(X,Y) \in \calP$. We denote the corresponding coordinate value of $\mv f$
by $f(Y|X)$. The vector $\mv f$ is called a {\em conditional polymatroid}
iff there exists a polymatroid $h$ such that $f(Y|X) = h(Y)-h(X)$; and, we
say $h$ defines the conditional polymatroid $\mv f$.
Abusing notation somewhat, the conditional polymatroid defined by the polymatroid
$h$ is denoted by $\mv h$. In particular, $\mv h = (h(Y|X))_{(X,Y)\in\calP}$.
If $h$ is a polymatroid then $h(\emptyset)=0$, in which case we write $h(Y)$
instead of $h(Y|\emptyset)$.
\edefn

To formally define the notion of a proof sequence, we
rewrite the Shannon flow inequality~\eqref{eqn:sfi} as an inequality on conditional
polymatroids in the $\Q_+^\calP$ space.
We extend the vectors
$\vec\lambda_\calB \in \Q_+^\calB$
and
$\vec\delta_\calC \in \Q_+^\calC$
to become vectors
$\vec\lambda,\vec\delta$
in the $\Q_+^\calP$ space in the obvious way:
\begin{eqnarray*}
\lambda(Y|X)
&\defeq& \begin{cases}
\lambda_\calB(B) & \text{ when } Y=B, X = \emptyset\\
0 & \text{ otherwise.}
\end{cases}\\
\delta(Y|X) &\defeq&
\begin{cases}
\delta_\calC(Y|X) & \text{ when } (X,Y) \in \calC\\
0 & \text{ otherwise.}
\end{cases}
\end{eqnarray*}
Then, inequality~\eqref{eqn:sfi} can be written simply as
$\inner{\vec\lambda, \mv h} \leq \inner{\vec\delta,\mv h}.$
Note the crucial fact that, even though $\vec\lambda \in \Q_+^\calP$, for it to
be part of a Shannon flow inequality only the entries $\lambda_{B|\emptyset}$
can be positive.
We will often write $\lambda_B$ instead of $\lambda_{B|\emptyset}$.
These assumptions are implicit henceforth.
Proposition~\ref{prop:sfi:ddl:target} can now be written simply as:

\bprop
Given any $\vec\lambda, \vec\delta \in \Q_+^\calP$, where
$\lambda_{Y|X} > 0$ implies $X= \emptyset$,
the inequality $\inner{\vec\lambda,\mv h}
\leq \inner{\vec\delta,\mv h}$
is a Shannon flow inequality if and only if there exist $\vec\sigma$ and
$\vec\mu$ such that $(\vec\delta, \vec\sigma, \vec\mu)$ satisfy the
constraints
\[
   \flow(Z) \geq \lambda_Z, \forall\emptyset \neq Z \subseteq \calV \text{ and }
(\vec\delta, \vec\sigma, \vec\mu) \geq\mv 0
\]
\label{prop:sfi}
\eprop
Note that the following set is a polyhedron and it is independent of the values $n_{Y|X}$:
\begin{equation}
   \bigl\{ (\vec\delta,\vec\sigma,\vec\mu) \suchthat
\flow(Z) \geq \lambda_Z, \forall\emptyset \neq Z \subseteq \calV \text{ and }
(\vec\delta, \vec\sigma, \vec\mu) \geq\mv 0
   \bigr\}.
\label{eqn:dual:polyhedron}
\end{equation}

The conditional polymatroids satisfy four basic linear inequalities:
\begin{eqnarray*}
h(I\cup J|J)-h(I|I\cap J)&\leq&0, \myskip I \incomp J\text{ (submodularity)}\\
-h(Y|\emptyset)+h(X|\emptyset)&\leq&0, \myskip X\subset Y\text{ (monotonicity)}\\
h(Y|\emptyset)-h(Y|X)-h(X|\emptyset)&\leq&0, \myskip X\subset Y\text{ (composition)}\\
-h(Y|\emptyset)+h(Y|X)+h(X|\emptyset)&\leq&0, \myskip X\subset Y\text{ (decomposition)}
\end{eqnarray*}
For every $I\incomp J$, define a vector
$\mv s_{I,J} \in\Q_+^\calP$, and for every $X\subset Y$, define three vectors $\mv m_{X,Y}, \mv c_{X,Y},\mv d_{Y,X} \in
\Q_+^\calP$ such that the linear inequalities above can be written correspondingly in
dot-product form:
\begin{eqnarray}
   \inner{\mv s_{I,J},\mv h}&\leq&0, \myskip I \incomp J\text{ (submodularity)} \label{eqn:sij}\\
   \inner{\mv m_{X,Y},\mv h}&\leq&0, \myskip X\subset Y\text{ (monotonicity)}
   \label{eqn:mxy}\\
   \inner{\mv c_{X,Y},\mv h}&\leq&0, \myskip X\subset Y\text{ (composition)}
   \label{eqn:cxy}\\
   \inner{\mv d_{Y,X},\mv h}&\leq&0, \myskip X\subset Y\text{ (decomposition) }
   \label{eqn:dyx}
\end{eqnarray}

\bdefn[Proof sequence]
\label{defn:ps}
A {\em proof sequence} of a Shannon flow inequality
$\inner{\vec\lambda, \mv h} \leq \inner{\vec\delta,\mv h}$, is a sequence
$(w_1\mv f_1,\dots,w_\ell\mv f_\ell)$ satisfying the following:
\bi
 \item[(1)] $\mv f_i \in \{$
   $\mv s_{I,J}$,
   $\mv m_{X,Y}$,
   $\mv c_{X,Y}$,
$\mv d_{Y,X} \}$ for all $i \in [\ell]$.
The $\mv f_i$ are called {\em proof steps}.
 \item[(2)] $w_i \in \R_+$, $i\in [\ell]$ are the corresponding {\em weights} of
    the proof steps.
 \item[(3)] All the vectors $\vec\delta_0 \defeq \vec\delta, \vec\delta_1, \dots, \vec\delta_\ell$
    defined by $\vec\delta_{i} = \vec\delta_{i-1} + w_i \cdot \mv f_i$, $i \in
    [\ell]$ are component-wise non-negative.
 \item[(4)] Furthermore, $\vec\delta_\ell \geq \vec\lambda$ (component-wise comparisons).
\ei
\edefn

Due to the linear inequalities~\eqref{eqn:sij}--\eqref{eqn:dyx} above, if
$(w_1\mv f_1,\cdots,w_\ell \mv f_\ell)$ is a proof sequence for the Shannon flow
inequality $\inner{\vec\lambda,\mv h} \leq \inner{\vec\delta,\mv h}$, then
\[ \inner{\vec\delta,\mv h}=
\inner{\vec\delta_0,\mv h} \geq
\inner{\vec\delta_1,\mv h} \geq
\cdots
\inner{\vec\delta_\ell,\mv h} \geq
\inner{\vec\lambda,\mv h}.
\]
The proof step $\mv s_{I,J}$ is called a {\em submodularity step},
$\mv m_{X,Y}$ a {\em monotonicity step},
$\mv d_{Y,X}$ a {\em decomposition step}, and
$\mv c_{X,Y}$ a {\em composition step}.

\bdefn[Witness]
Let $\inner{\vec\lambda,\mv h} \leq \inner{\vec\delta,\mv h}$ be a Shannon flow
inequality. From Proposition~\ref{prop:sfi:ddl:target} there exists
$(\vec\sigma,\vec\mu)$ such that $(\vec\delta,\vec\sigma,\vec\mu)$ belongs to
the polyhedron~\eqref{eqn:dual:polyhedron}.
We call $(\vec\sigma,\vec\mu)$ a {\em witness} for the Shannon flow inequality.
\edefn

We next show one way to construct a proof sequence
for any given  Shannon flow inequality. (See Appendix~\ref{app:sec:ps}
for more advanced constructions of shorter proof sequences.)
\bthm[Constructing a proof sequence]
Let $\inner{\vec\lambda, \mv h} \leq \inner{\vec\delta,\mv h}$ be a Shannon
flow inequality with witness $(\vec\sigma,\vec\mu)$.
There exists a proof sequence for the inequality $\inner{\vec\lambda, \mv h} \leq \inner{\vec\delta,\mv h}$ with length at most
$D(3\norm{\vec\sigma}_1+\norm{\vec\delta}_1+\norm{\vec\mu}_1)$, where
$D$ is the minimum common denominator
of all entries in $(\vec\lambda,\vec\delta,\vec\sigma,\vec\mu)$.
Moreover, such a proof sequence can be constructed in time polynomial in
$D(\norm{\vec\lambda}_1+2\norm{\vec\sigma}_1+\norm{\vec\delta}_1+\norm{\vec\mu}_1)$.
\label{thm:ps:construction:1}
\ethm
\bp
We induct on the quantity
\begin{equation}
   \ell(\vec\lambda,\vec\delta,\vec\sigma,\vec\mu) \defeq D\left(\norm{\vec\lambda}_1 +
2\norm{\vec\sigma}_1+\norm{\vec\delta}_1+\norm{\vec\mu}_1\right),
\label{eq:ps:construction:ell}
\end{equation}
which is an
integer.
The base case is when $\norm{\vec\lambda}_1=0$, which is trivial
because the inequality has a proof sequence of length $0$.
In the inductive step, assume $\norm{\vec\lambda}_1>0$, meaning
there must be some nonempty $B\subseteq \calV$ for which $\lambda_B>0$.
We will produce a Shannon flow inequality
$\inner{\vec\lambda',\mv h} \leq \inner{\vec\delta',\mv h}$
witnessed by $(\vec\sigma',\vec\mu')$ such that
$\ell(
\vec\lambda',\vec\delta',\vec\sigma',\vec\mu'
) <
\ell(\vec\lambda,\vec\delta,\vec\sigma,\vec\mu)$, and $D$ is a common
denominator of the entries in $(\vec\lambda',\vec\delta',\vec\sigma',\vec\mu')$.
From the induction hypothesis we obtain a proof sequence $\proofseq'$
for $\inner{\vec\lambda',\mv h} \leq \inner{\vec\delta',\mv h}$.
Finally the proof sequence $\proofseq$ for
$\inner{\vec\lambda,\mv h} \leq \inner{\vec\delta,\mv h}$
is constructed from $\proofseq'$ by appending to the beginning one or two proof
steps as outlined below.

From Proposition~\ref{prop:sfi}, we know $\sum_{\emptyset
\neq W\subseteq \calV}\flow(W) \geq \lambda_B>0$.
Consequently, there must exist $Z\neq\emptyset$ for which
$\delta_{Z|\emptyset}>0$, because all the variables
$\delta_{Y|X}$ with $X \neq \emptyset$, $\sigma_{I,J}$, and $\mu_{X,Y}$
contribute a non-positive amount to the sum
$\sum_{\emptyset\neq W \subseteq\calV}\flow(W)$.
Let $w \defeq 1/D$, and
fix an arbitrary $Z\neq\emptyset$ where $\delta_{Z|\emptyset}>0$.
We initially set
$(\vec\lambda',\vec\delta',\vec\sigma',\vec\mu')=(\vec\lambda,\vec\delta,\vec\sigma,\vec\mu)$;
then we modify
$(\vec\lambda',\vec\delta',\vec\sigma',\vec\mu')$ slightly depending on the
cases below.

{\em Case (a): $\lambda_Z>0$.}
Reduce both $\lambda'_{Z}$ and $\delta'_{Z|\emptyset}$ by $w$,
which reduces $\ell(\vec\lambda',\vec\delta',\vec\sigma',\vec\mu')$ by $2$.
From Proposition~\ref{prop:sfi}, we can verify that
$\inner{\vec\lambda',\mv h}\leq \inner{\vec\delta',\mv h}$ is a Shannon flow inequality
witnessed by $(\vec\sigma',\vec\mu')=(\vec\sigma,\vec\mu)$.
By induction hypothesis, $\inner{\vec\lambda',\mv h}\leq \inner{\vec\delta',\mv h}$
has a proof sequence $\proofseq'$ of length at most
$D(3\norm{\vec\sigma}_1+\norm{\vec\delta'}_1+\norm{\vec\mu}_1)$.
Furthermore, the $\proofseq'$
{\em is} also a proof sequence for $\inner{\vec\lambda,\mv h}\leq
\inner{\vec\delta,\mv h}$.

{\em Case (b): $\lambda_Z=0$ and $\flow(Z)>0$.}
Reduce $\delta'_{Z|\emptyset}$ by $w$,
thus reducing $\ell(\vec\lambda',\vec\delta',\vec\sigma',\vec\mu')$ by 1.
Then, from Proposition~\ref{prop:sfi}, we can verify that
$\inner{\vec\lambda,\mv h}\leq \inner{\vec\delta',\mv h}$ is a Shannon flow inequality
witnessed by $(\vec\sigma,\vec\mu)$. The inductive step is now identical to that of
Case (a).

{\em Case (c): $\lambda_Z=0$ and $\flow(Z)=0$.} Since $\delta_{Z|\emptyset}>0$, there must be some
dual variable that is contributing a negative amount to $\flow(Z)$ (See Figure~\ref{fig:flow} [right]). In particular, one of the following three
cases must hold:

(1) There is some $X\subset Z$ such that $\mu_{X,Z} \geq w$.
Define $\vec\delta' = \vec\delta + w\cdot \mv m_{X,Z}$ and reduce $\mu'_{X,Z}$ by $w$.
By Definition of $\mv m_{X,Z}$, $\delta'_{Z|\emptyset} = \delta_{Z|\emptyset}-w$
and $\delta'_{X|\emptyset}=\delta_{X|\emptyset}+w$.
Therefore,
$\norm{\vec\delta'}_1=\norm{\vec\delta}_1$,
$\norm{\vec\mu'}_1=\norm{\vec\mu}_1-w$,
$\ell(
\vec\lambda',\vec\delta',\vec\sigma',\vec\mu'
) =
\ell(\vec\lambda,\vec\delta,\vec\sigma,\vec\mu) - 1$,
and by Proposition~\ref{prop:sfi}
$\inner{\vec\lambda,\mv h}\leq \inner{\vec\delta',\mv h}$ is a Shannon flow inequality.
witnessed by $(\vec\sigma,\vec\mu')$.
By induction hypothesis, $\inner{\vec\lambda,\mv h}\leq \inner{\vec\delta',\mv h}$
has a proof sequence $\proofseq'$ of length at most
$D(3\norm{\vec\sigma}_1+\norm{\vec\delta'}_1+\norm{\vec\mu'}_1)$.
It follows that $\proofseq \defeq (w\cdot \mv m_{X,Z}, \proofseq')$ is a proof
sequence for
$\inner{\vec\lambda,\mv h}\leq \inner{\vec\delta,\mv h}$ of length at most
$D(3\norm{\vec\sigma}_1+\norm{\vec\delta}_1+\norm{\vec\mu}_1)$.

(2) There is some $Y \supset Z$ such that $\delta_{Y|Z} \geq w$.
Define $\vec\delta' = \vec\delta + w\cdot \mv c_{Z,Y}$.
Note that $\norm{\vec\delta'}_1=\norm{\vec\delta}_1-w$,
$\ell(
\vec\lambda',\vec\delta',\vec\sigma',\vec\mu'
) =
\ell(\vec\lambda,\vec\delta,\vec\sigma,\vec\mu) - 1$,
and by Proposition~\ref{prop:sfi}
$\inner{\vec\lambda,\mv h}\leq \inner{\vec\delta',\mv h}$ is a Shannon flow
inequality witnessed by $(\vec\sigma,\vec\mu)$.
From the proof sequence $\proofseq'$ for
$\inner{\vec\lambda,\mv h}\leq \inner{\vec\delta',\mv h}$ we obtain the proof sequence
$\proofseq \defeq (w \cdot \mv c_{Z,Y},\proofseq')$ for
$\inner{\vec\lambda,\mv h}\leq \inner{\vec\delta,\mv h}$ of the desired length.

(3) There is some $J \incomp Z$ such that $\sigma_{Z,J} \geq w$.
Define $\vec\delta' = \vec\delta + w \cdot \mv d_{Z, Z\cap J}+w\cdot \mv
s_{Z,J}$, and reduce $\sigma'_{Z,J}$ by $w$.
In this case,
$\norm{\vec\delta'}_1=\norm{\vec\delta}_1+w$,
$\norm{\vec\sigma'}_1=\norm{\vec\sigma}_1-w$,
$\ell(
\vec\lambda',\vec\delta',\vec\sigma',\vec\mu'
) =
\ell(\vec\lambda,\vec\delta,\vec\sigma,\vec\mu) - 1$,
and by Proposition~\ref{prop:sfi}
$\inner{\vec\lambda,\mv h}\leq \inner{\vec\delta',\mv h}$ is a Shannon
flow inequality witnessed by $(\vec\sigma',\vec\mu)$.
By induction hypothesis, $\inner{\vec\lambda,\mv h}\leq \inner{\vec\delta',\mv h}$
has a proof sequence $\proofseq'$ of length at most
$D(3\norm{\vec\sigma'}_1+\norm{\vec\delta'}_1+\norm{\vec\mu}_1)$.
It follows that $\proofseq \defeq
(w \cdot \mv d_{Z, Z \cap J}, w\cdot \mv s_{Z,J}, \proofseq')$ is a proof
sequence for
$\inner{\vec\lambda,\mv h}\leq \inner{\vec\delta,\mv h}$ of length at most
$D(3\norm{\vec\sigma}_1+\norm{\vec\delta}_1+\norm{\vec\mu}_1)$.

Note that each step above can be implemented in time polynomial in
$\ell(\vec\lambda',\vec\delta',\vec\sigma',\vec\mu')$, which serves as an upper bound
on the number of non-zero values in $\vec\lambda', \vec\delta', \vec\sigma'$,
and $\vec\mu'$. (Recall that $D$ in~\eqref{eq:ps:construction:ell} is the common denominator.)
Since $\ell(
   \vec\lambda',\vec\delta',\vec\sigma',\vec\mu'
   ) <
   \ell(\vec\lambda,\vec\delta,\vec\sigma,\vec\mu)$,
the entire proof sequence can be constructed in time polynomial in
$\ell(\vec\lambda,\vec\delta,\vec\sigma,\vec\mu)$.
\ep

Appendix~\ref{app:subsec:bounding-delta-sigma-mu} presents bounds on
$\norm{\vec\sigma}_1$, $\norm{\vec\delta}_1$ and $\norm{\vec\mu}_1$.
By plugging those bounds into Theorem~\ref{thm:ps:construction:1}, we can
bound the length of the constructed proof sequence by
$D(3\norm{\vec\sigma}_1+\norm{\vec\delta}_1+\norm{\vec\mu}_1)\leq D\cdot(\frac{3}{2} n^3+n).$ (See Corollaries~\ref{cor:bound-mu} and~\ref{cor:2sigma<=n-cube-lambda}.)
Appendix~\ref{sec:poly-sized-proofseq} presents an alternative construction of a proof sequence with an even smaller length bound.

The $\panda$ algorithm needs another technical lemma, which we state below along with some prerequisite concepts.
\bdefn[Tight witness]
A witness is said to be {\em tight} if $\flow(Z) = \lambda_Z$ for all $Z$.
\label{defn:tight-witness}
\edefn
We remark that, if $\flow(Z)>\lambda_Z$, we can always increase
$\mu_{\emptyset,Z}$ by the amount $\flow(Z)-\lambda_Z$ so that $\flow(Z) =
\lambda_Z$. In particular, it is easy to turn any witness into a tight witness.
The proof of Lemma~\ref{lmm:truncate:sfi} below follows the same strategy as
in the proof of Theorem~\ref{thm:ps:construction:1}, but with some subtle differences.

\blmm[Truncating a Shannon flow inequality]
Let $\inner{\vec\lambda,\mv h}\leq \inner{\vec\delta,\mv h}$ be a Shannon flow
inequality with witness $(\vec\sigma,\vec\mu)$. Let $D$ be a common
denominator
of all entries in $(\vec\lambda,\vec\delta,\vec\sigma,\vec\mu)$, and
$w\defeq 1/D$.
Suppose $\norm{\vec\lambda}_1>0$
and $\delta_{Y|\emptyset} >0$. Then, there are vectors
$(\vec\lambda',\vec\delta',\vec\sigma',\vec\mu')$ satisfying the following conditions:
\bi
 \item[(a)] $\inner{\vec\lambda',\mv h} \leq \inner{\vec\delta',\mv h}$ is a
    Shannon flow inequality (with witness $(\vec\sigma',\vec\mu')$).
 \item[(b)] $\vec\lambda' \leq \vec\lambda$ and $\vec\delta' \leq \vec\delta$
    (component-wise comparisons).
 \item[(c)] $\norm{\vec\lambda'}_1\geq\norm{\vec\lambda}_1-w$ and
    $\delta'_{Y|\emptyset} \leq \delta_{Y|\emptyset}-w$.
 \item[(d)] $D$ is a common denominator of all entries in the vector
    $(\vec\lambda',\vec\delta',\vec\sigma',\vec\mu')$.
 \item[(e)] $D(3\norm{\vec\sigma'}_1+\norm{\vec\delta'}_1+\norm{\vec\mu'}_1)
    \leq D(3\norm{\vec\sigma}_1+\norm{\vec\delta}_1+\norm{\vec\mu}_1)-1$.
\ei
\label{lmm:truncate:sfi}
\elmm
\bp
W.L.O.G. we can assume that $(\vec\sigma,\vec\mu)$ is a tight witness.
Construct $(\vec\lambda',\vec\delta',\vec\sigma',\vec\mu')$
from $(\vec\lambda,\vec\delta,\vec\sigma,\vec\mu)$ as follows.
Initially set $(\vec\lambda',\vec\delta',\vec\sigma',\vec\mu') =
(\vec\lambda,\vec\delta,\vec\sigma,\vec\mu)$.
Let $\flow'(Z)$ denote the quantity $\flow(Z)$ measured on the vector
$(\vec\lambda',\vec\delta',\vec\sigma',\vec\mu')$.
Due to tightness of the witness, at this point $\flow'(Z)-\lambda'_Z=0$
for every $Z$.

Now, we start disturbing the flow balance equations starting from $Z=Y$
by setting $\delta'_{Z|\emptyset}=\delta'_{Z|\emptyset}-w$ which means
$\flow'(Z)$ was reduced by $w$. Note that $Z$ is the only point for which
$\flow'(Z)-\lambda'_Z \neq 0$ (it is negative).
If $\lambda'_Z>0$, then we simply reduce $\lambda'_Z$ by $w$ and terminate.
If $\lambda'_Z=0$, then either (1) there is some $X\subset Z$ such that
$\mu'_{X,Z} \geq w$,
(2) there is some $Y \supset Z$ such that $\delta'_{Y|Z} \geq w$,
or (3) there is some $J \incomp Z$ such that $\sigma'_{Z,J}
\geq w$.
Cases (1) and (2) are handled in a similar way to the proof of Theorem~\ref{thm:ps:construction:1} while case (3) is handled differently.
In particular, if (1) holds, then we reduce $\mu'_{X,Z}$ by $w$ and set $Z = X$.
If (2) holds, then we reduce $\delta'_{Y|Z}$ by $w$ and set $Z=Y$.
If (3) holds, then we reduce $\sigma'_{Z,J}$ by $w$, {\em increase} $\mu'_{Z\cap
J, J}$ by $w$, and set $Z=Z\cup J$.
In all three cases, (the new) $Z$ is the only point where $\flow'(Z)-\lambda'_Z$
has a deficit, and the process continues if the new $Z$ is not $\emptyset$.

The above process terminates because every time we move $Z$ to a new deficit
point, the quantity
$2\norm{\vec\sigma'}_1+\norm{\vec\delta'}_1+\norm{\vec\mu'}_1$ is reduced by $w$.
When the process terminates, all quantities $\flow'(Z)-\lambda'_Z=0$ and thus
$\inner{\vec\lambda',\mv h} \leq \inner{\vec\delta',\mv h}$ is a Shannon flow
inequality witnessed by $(\vec\sigma',\vec\mu')$ by Proposition~\ref{prop:sfi}.
Property (b) holds because we never increase the $\vec\lambda'$ and $\vec\delta'$
entries.
Property (c) holds because we started the process by reducing $\delta'_{Y|\emptyset}$
and the process terminates as soon as some $\lambda'_Z$ is reduced by $w$.
Property (d) hold trivially.
Property (e) holds because $\norm{\vec\delta'}_1 < \norm{\vec\delta}_1$
and the quantity $3\norm{\vec\sigma'}_1+\norm{\vec\mu'}_1$ never increases in the above process.
\ep

\section{The $\panda$ algorithm}
\label{sec:panda}

This section presents an algorithm called $\panda$ that computes a model
of a disjunctive datalog rule $P$ in time predicted by its polymatroid
bound~\eqref{eqn:intro:polymatroid:disjunctive:2}:
$$\tilde O(N + \poly(\log N)\cdot 2^{\outputbound_{\Gamma_n\cap
\hdc}(P)}).$$ (Recall $N$ was defined in~\eqref{eqn:N}.)
The principle in $\panda$ is the following: start by providing a proof
sequence for a Shannon flow inequality, then interpret each step of the
proof sequence as a relational operation on the query's input
relations.  The main result of this section is
Theorem~\ref{thm:main:panda:disjunctive}, whose proof is in
Section~\ref{subsec:panda:analysis}.

\subsection{The algorithm}
\label{subsec:panda:algorithm}

\begin{algorithm*}[ht!]
  \caption{$\panda$($\calR$, $\dc$, $(\vec\lambda,\vec\delta)$, $\proofseq$)}
  \label{algo:panda}
  \begin{algorithmic}[1]
    \Require{$\inner{\vec\lambda, \mv h} \leq \inner{\vec\delta,\mv h}$ is a
Shannon flow inequality with proof sequence \proofseq,
$0<\norm{\vec\lambda}_1\leq 1$}
    \Require{$\dc$ are input degree constraints guarded by input relations $\calR$}
    \Require{The degree-support invariant (invariant~(I1)) is satisfied}
    \Require{Invariant (I4) is satisfied}
\Statex
\If {an input relation $R \in\calR$ has attribute set $\mv A_B$, with $B \in\calB$}\label{alg:ln:basecase}
     \State \Return $T_B = R$\Comment{Only one table in the output}
\EndIf
\State Let $\proofseq = (w\cdot \mv f, \proofseq')$ \Comment{$\mv f$ is the first
proof step, with weight $w$}
     \State $\vec\delta' \gets \vec\delta+w \cdot \mv f$\Comment{Advance the proof step}
\If {$\mv f = \mv s_{I,J}$} \label{alg:ln:case1} \Comment{Case 1. $\delta_{I|I\cap J} \geq w > 0$ must hold}
   \State \Return $\panda$($\calR$, $\dc$, $(\vec\lambda,\vec\delta')$, $\proofseq'$)
\ElsIf {$\mv f = \mv m_{X,Y}$} \label{alg:ln:case2} \Comment{Case 2. $\delta_{Y| \emptyset} \geq w > 0$ must hold}
   \State Let $R$ be a guard for $(\emptyset,Y,N_{Y|\emptyset}) \in \dc$,
     which supports $\delta_{Y|\emptyset}$
     \State $\calR' = \calR \cup \{\Pi_X(R)\}$
     \State $\dc' = \dc \cup \{(\emptyset,X,N_{X|\emptyset}\defeq|\Pi_X(R)|)\}$
  \State \Return $\panda$($\calR'$, $\dc'$, $(\vec\lambda,\vec\delta')$, $\proofseq'$)
\ElsIf {$\mv f=\mv d_{Y,X}$} \label{alg:ln:case3} \Comment{Case 3. $\delta_{Y|\emptyset} \geq w > 0$
     must hold}
\State Let $R\in\calR$ be a guard for $(\emptyset, Y, N_{Y|\emptyset}) \in \dc$
\State Partition $R = R^{(1)} \cup \cdots \cup R^{(k)}$ as in~\eqref{eqn:partition}
\Comment{k=$O(\log_2 |R|)$ branches}
\For {$j \gets 1$ {\bf to} $k$}\label{alg:ln:subproblems}
\State $\calR^{(j)} \gets \calR \cup \{R^{(j)}\}$
\State $\dc^{(j)} = \dc \cup \bigl\{ (\emptyset, X, N^{(j)}_{X|\emptyset}),
(X,Y,N^{(j)}_{Y|X}) \bigr\}$
\State $(T^{(j)}_B)_{B\in\calB} \gets \panda(\calR^{(j)}, \dc^{(j)},
(\vec\lambda,\vec\delta'), \proofseq')$
\EndFor
\State \Return $\left( T_B \defeq \bigcup_{j=1}^k T^{(j)}_B \right)_{B\in\calB}$\Comment{Union of
results from branches}
\ElsIf {$\mv f = \mv c_{X,Y}$} \label{alg:ln:case4} \Comment{Case 4. $\delta_{X|\emptyset} \geq w > 0$ and $\delta_{Y|X} \geq w > 0$}
\State Let $R$ be a guard of $(\emptyset, X, N_{X|\emptyset})\in\dc$, which
supports $\delta_{X | \emptyset}$
\State Let $S$ be a guard of $(Z, W, N_{W|Z}) \in \dc$, which supports $\delta_{Y|X}$
   \If {$N_{X|\emptyset}\cdot N_{W|Z} \leq 2^\obj$}\label{alg:ln:case4a}\Comment{Case 4a. $\obj$ defined
in~\eqref{eqn:obj}}
     \State Compute $T(\mv A_Y) \gets \Pi_X(R) \Join \Pi_W(S)$ \Comment{Within
runtime $\tilde O(2^{\obj})$}
\State $\calR' = \calR \cup \{T\}$
\State $\dc' = \dc \cup \{(\emptyset,Y,N_{Y|\emptyset}\defeq|T|)\}$
  \State \Return $\panda$($\calR'$, $\dc'$, $(\vec\lambda,\vec\delta')$, $\proofseq'$)
     \Else\label{alg:ln:case4b}\Comment{Case 4b.}
     \State Let $\inner{\vec\lambda', \mv h} \leq \inner{\vec\delta',\mv h}$ be
the truncated Shannon flow inequality stated in Lemma~\ref{lmm:truncate:sfi}
     \State Recompute a fresh $\proofseq'$ for $\inner{\vec\lambda', \mv h} \leq \inner{\vec\delta',\mv h}$
     \State \Return $\panda$($\calR$, $\dc$, $(\vec\lambda',\vec\delta')$, $\proofseq'$)
   \EndIf
\EndIf
  \end{algorithmic}
\end{algorithm*}

A simple demonstration of the algorithm was given in Example~\ref{ex:intro:panda}.
The algorithm sketch is shown in the box Algorithm~\ref{algo:panda}.
In general, $\panda$ takes as input the collection of input relations $\calR$,
the degree constraints $\dc$, a Shannon flow inequality and its proof sequence.
The Shannon flow inequality is constructed by solving the optimization problem
\eqref{eqn:intro:output:bound} with $\calF=\Gamma_n \cap \hdc$:
\begin{equation}
   \outputbound_{\Gamma_n \cap \hdc}(P) \defeq \max_{h\in\Gamma_n \cap \hdc} \min_{B\in\calB} h(B).
\end{equation}
From Lemma~\ref{lmm:lambda:1:reformulation}, we can find a vector
$\vec\lambda_\calB$
with $\norm{\vec\lambda}_1=1$ such that the problem has the same optimal
objective value as the linear program
$\max_{h\in\Gamma_n\cap\hdc} \inner{\vec\lambda, \mv h}$.
(Recall from Section~\ref{subsec:sfi:ps} that,
when we extend $\vec\lambda_\calB$ to the (conditional polymatroid)
space $\vec\lambda \in \Q_+^{\calP}$, only the entries
$\lambda_{B|\emptyset}$ for $B\in\calB$ can be positive.)
Let $(\vec\delta,\vec\sigma,\vec\mu)$ denote an optimal dual solution
to this LP, then by strong duality
\begin{equation}
\sum_{(X,Y)\in \dc}n_{Y|X} \cdot \delta_{Y|X}=\outputbound_{\Gamma_n\cap
\hdc}(P)
\label{eqn:panda:dual-obj=bound}
\end{equation}
Moreover, from Proposition~\ref{prop:sfi} we
know $\inner{\vec\lambda, \mv h}\leq \inner{\vec\delta,\mv h}$ is a Shannon
flow inequality. From Theorem~\ref{thm:ps:construction:1}, we obtain a proof sequence
for the Shannon flow inequality.

For brevity, let the constant $\obj$ denote our budget of $\outputbound_{\Gamma_n \cap \hdc}(P)$:
\begin{equation}
\obj \defeq \outputbound_{\Gamma_n \cap \hdc}(P).
\label{eqn:obj}
\end{equation}

Throughout the algorithm, the collection of input relations $\calR$,
the degree constraints $\dc$, the Shannon flow inequality $\inner{\vec\lambda, \mv h}\leq \inner{\vec\delta,\mv h}$, and the associated proof sequence
will all be updated.
However for inductive purposes, the following invariants will always be maintained:
\be
\item[(I1)] \emph{Degree-support invariant:}
For every $\delta_{Y|X} > 0$, there exist $Z\subseteq X$, $W \subseteq
Y$ such that $W-Z = Y-X$ and $(Z,W,N_{W|Z}) \in \dc$.
(Note that, if $X=\emptyset$ then $W=Y$ and $Z=\emptyset$.)
\begin{figure}[!ht]
\centering \begin{tikzpicture}[domain=0:20, 
every node/.style={font=\large}, scale=0.8, every node/.style={scale=0.8}]
\node[] at (4,4) (Y) {$Y$};
\node[] at (2,2) (X) {$X$};
\node[] at (6,2) (W) {$W$};
\node[] at (4,0) (Z) {$Z$};

\node[color=black!55!green] at (2.2,3.3) (deltayx) {$\delta_{Y|X}>0$};
\node[color=black!55!green] at (6.5,0.6) (dc) {$(Z,W,N_{W|Z}) \in \dc$};

\draw (X) -- (Y);
\draw (X) -- (Z);
\draw (Z) -- (W);
\draw (W) -- (Y);

\node[] at (4,-1) {(a) Degree-support invariant};

\node[] at (14,4) (I) {$I$};
\node[] at (12,2) (I cap J) {$I\cap J$};
\node[] at (16,2) (W1) {$W$};
\node[] at (14,0) (Z1) {$Z$};
\node[] at (12,6) (I cup J) {$I \cup J$};
\node[] at (10,4) (J) {$J$};

\draw (I) -- (W1);
\draw (W1) -- (Z1);
\draw (Z1) -- (I cap J);
\draw (I cap J) -- (I);
\draw (I) -- (I cup J);
\draw (J) -- (I cup J);
\draw (J) -- (I cap J);

\node[color=black!55!green] at (10,5) () {$\delta'_{I \cup J | J}>0$};
\node[color=black!55!green] at (16.5,0.6) (dc) {$(Z,W,N_{W|Z}) \in \dc$};

\node[] at (12,-1) {(b) Case 1 of $\panda$};

\end{tikzpicture}
   \caption{Degree-support invariant and its usage}
   \Description{Degree-support invariant and its usage}
\label{fig:dsi}
\end{figure}
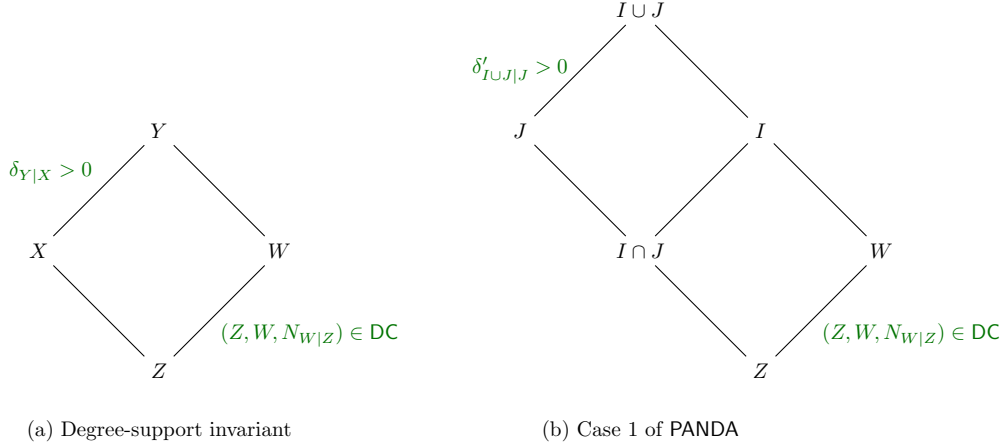
The degree constraint $(Z,W,N_{W|Z})$ is said to {\em support}
the positive $\delta_{Y|X}$. See Figure~\ref{fig:dsi} (a).
(If there are multiple constraints $(Z,W,N_{W|Z})$ supporting $\delta_{Y|X}$,
then the one with the minimum $N_{W|Z}$ is said to be \emph{the}
constraint that supports $\delta_{Y|X}$, where ties can be broken arbitrarily.)
\item[(I2)] $\vec\lambda$ satisfies
\begin{equation}
   0<\norm{\vec\lambda}_1\leq 1.
   \label{eqn:invar:norm-lambda}
\end{equation}
\item[(I3)] The Shannon flow inequality along with the supporting degree constraints satisfy the following:
\begin{equation}
\sum_{(X,Y)}n(\delta_{Y|X}) \leq \norm{\vec\lambda}_1\cdot \obj,
\label{eqn:invar:potential}
\end{equation}
where
\[ n(\delta_{Y|X}) \defeq \begin{cases}
\delta_{Y|X}\cdot n_{W|Z} & \text{ if } \delta_{Y|X}>0 \text{ and } \\\ & (Z,W,N_{W|Z}) \text{ supports it}\\
0 & \text{ if } \delta_{Y|X}=0.
\end{cases}
\]
(Recall that $n_{W|Z}\defeq\log_2 N_{W|Z}$.)
We call the quantity $\sum_{(X,Y)}n(\delta_{Y|X})$ \emph{the potential}.
\item[(I4)] For every $\delta_{Y|\emptyset}>0$, the supporting degree constraint $(\emptyset, Y, N_{Y|\emptyset})$ satisfies $n_{Y|\emptyset}\leq \obj$.
\ee

At the very beginning, all the above invariants are satisfied.
In particular, invariant (I1) is satisfied because of~\eqref{eqn:panda:dual-obj=bound},
which implies that for every $\delta_{Y|X} > 0$, $(X, Y) \in \dc$.
Invariant (I2) is satisfied because initially $\norm{\vec\lambda}_1=1$ thanks to
Lemma~\ref{lmm:lambda:1:reformulation}.
Invariant (I3) holds because~\eqref{eqn:panda:dual-obj=bound} and~\eqref{eqn:obj}
imply that $\sum_{(X,Y)}n(\delta_{Y|X}) = \obj$.
Invariant (I4) holds because of Proposition~\ref{prop:panda:i4}, which will be discussed later.

A very high-level description of the algorithm is as follows.
Recall from Definition~\ref{defn:ps} that a proof
sequence $\proofseq$ is a series of proof steps, which are used by
$\panda$ as ``symbolic instructions''.
For each instruction, $\panda$ does some computation, spawns a number
of subproblem(s) all of which are disjunctive datalog rules,
and creates new (intermediate) relations to become input of the subproblems if
necessary.
The output of the $i$th subproblem is a set of tables $T^{(i)}_B$ for
$B\in\calB$.
The overall output is the set of tables $T_B = \bigcup_i T^{(i)}_B$,
$B\in\calB$; namely for each $B \in \calB$ we take the union of the
corresponding tables from the subproblems's outputs.
The number of subproblems will be shown to be polylogarithmic in the
input size.

We now walk the reader step-by-step through the algorithm.
We will keep every step of the algorithm to run within $\tilde O(2^{\obj})$.
Specifically, we will keep every
intermediate relation the algorithm computes of size $\leq2^\obj$.
In the base case, the algorithm stops as soon as there is a
relation $R \in\calR$ with attribute set $\mv A_B$ where $B \in\calB$,
in which case $R$ is a target relation (line~\ref{alg:ln:basecase} in Algorithm~\ref{algo:panda}).
Otherwise, the algorithm takes steps which are modeled after the
proof steps. Let $\mv f$ be the first proof step (instruction) with
weight $w$, i.e. $\proofseq = (w \cdot \mv f, \proofseq')$
where $\proofseq'$ contains the rest of the instructions.

{\bf Case 1:} $\mv f=\mv s_{I,J}$ is a submodularity step
(line~\ref{alg:ln:case1} in Algorithm~\ref{algo:panda}).
By definition of proof sequence, $\vec\delta + w\cdot\mv s_{I,J} \geq \mv 0$,
and thus $\delta_{I|I\cap J} \geq w > 0$.
Let $(Z, W, N_{W|Z}) \in \dc$ be the degree constraint supporting
$\delta_{I|I\cap J}$; then
$Z \subseteq I\cap J$, $W \subseteq I$, and $W-Z=I-I\cap J$.
The algorithm proceeds by setting $\vec\delta'=\vec\delta+w \cdot \mv s_{I,J}$.
Note that $\delta'_{I\cup J | J}$ is now positive, and so it needs a supporting
degree constraint to maintain invariant (I1).
From the fact that $W-Z = I-I\cap J=I\cup J-J$, $(Z,W,N_{W|Z})$ can
support $\delta'_{I\cup J | J}$ (see Figure~\ref{fig:dsi} (b)).
Since $\mv f$ was the next step in the proof sequence,
$\inner{\vec\lambda,\mv h} \leq \inner{\vec\delta',\mv h}$ is a Shannon flow
inequality with proof sequence $\proofseq'$.
Moreover, the potential $\sum_{(X,Y)}n(\delta_{Y|X})$ remains unchanged because
$\delta_{I|I\cap J}$ was reduced by $w$,
$\delta_{I\cup J|J}$ was increased by $w$, and they have the same support.
Hence, invariant~\eqref{eqn:invar:potential} remains satisfied.
Moreover since we didn't change $\vec\lambda$ above,
invariant~\eqref{eqn:invar:norm-lambda} holds assuming it held before this step.
Finally, invariant (I4) holds because in the newly added
term $\delta'_{I\cup J | J}$, we have $J \neq \emptyset$ by definition of
submodularity step~\eqref{eqn:sij}.

{\bf Case 2:} $\mv f=\mv m_{X,Y}$ is a monotonicity step
(line~\ref{alg:ln:case2} in Algorithm~\ref{algo:panda}).
By definition of proof sequence, $\vec\delta + w\cdot\mv m_{X,Y} \geq \mv 0$,
and thus $\delta_{Y| \emptyset} \geq w > 0$.
Let $(\emptyset, Y, N_{Y|\emptyset}) \in \dc$ be the degree constraint supporting
$\delta_{Y|\emptyset}$, and $R\in\calR$ be a guard for this degree
constraint (which implies $|\Pi_Y(R)|\leq N_{Y|\emptyset}$; recall Definition~\ref{defn:guard}).
By invariant (I4) above, we have $N_{Y|\emptyset} \leq 2^\obj$.
We proceed by setting $\vec\delta'=\vec\delta+w \cdot \mv m_{X,Y}$.
Note that $\delta'_{X | \emptyset}$ is positive, and so it needs a supporting
degree constraint, which is the newly added degree constraint
$(\emptyset, X, N_{X|\emptyset})$, guarded by $R$, where
$N_{X|\emptyset} \defeq|\Pi_X(R)|\leq |\Pi_Y(R)|\leq N_{Y|\emptyset}$.
Invariant~\eqref{eqn:invar:potential} remains satisfied because we subtracted $w\cdot n_{Y|\emptyset}$ from the potential
and added $w\cdot n_{X|\emptyset}\leq w\cdot n_{Y|\emptyset}$ instead.
Moreover invariant (I4) remains satisfied because $N_{X|\emptyset} \leq N_{Y|\emptyset}\leq 2^\obj$.
Finally, invariant~\eqref{eqn:invar:norm-lambda} remains satisfied because we didn't change
$\vec \lambda$ above.

{\bf Case 3:} $\mv f=\mv d_{Y,X}$ is a decomposition step with weight $w$
(line~\ref{alg:ln:case3} in Algorithm~\ref{algo:panda}).
From $\vec\delta+w\cdot \mv d_{Y,X}\geq \mv 0$, it follows
that $\delta_{Y|\emptyset} \geq w > 0$.
From the guarantee that $\delta_{Y|\emptyset}$ has a supporting degree
constraint, it follows that there is a relation $R\in\calR$ guarding
$(\emptyset, Y, N_{Y|\emptyset})$, which means $|\Pi_Y(R)| \leq N_{Y|\emptyset}$.
By invariant (I4), $N_{Y|\emptyset}\leq 2^\obj$.
We will need the following lemma.

\blmm
Let $X\subset Y \subseteq \calV$.
Let $T(\mv A_Y)$ be a table with $|T| \leq N_{Y|\emptyset}$.
Then, $T$ can be partitioned into at most
$k = 2\log |T|$ sub-tables $T^{(1)},\dots,T^{(k)}$
such that $N^{(j)}_{X|\emptyset} \cdot N^{(j)}_{Y|X} \leq N_{Y|\emptyset}$,
for all $j \in [k]$, where
\begin{eqnarray*}
   N^{(j)}_{X|\emptyset} &\defeq& |\Pi_X(T^{(j)})|,\\
   N^{(j)}_{Y|X} &\defeq& \max_{\mv t_X\in\Pi_X(T^{(j)})}\deg_{T^{(j)}}(Y | \mv t_X),
\end{eqnarray*}
(where $\deg_{T^{(j)}}(Y | \mv t_X)$ was defined by~\eqref{eqn:degree} in Definition~\ref{defn:guard}.)
\label{lmm:partition}
\elmm
\bp
To obtain the sub-tables $T^{(j)}$, observe that the number of tuples $\mv t_X \in
\Pi_X(T)$ with $\log$-degree in the interval
$[j,j+1)$ is at most $|T|/2^j \leq 2^{n_{Y|\emptyset}-j}$.
Hence, if we partition $T$ based on which of the buckets $[j,j+1)$ the
$\log$-degree falls into, we would almost have  the required inequality:
$n^{(j)}_{X|\emptyset}+n^{(j)}_{Y|X} \leq (n_{Y|\emptyset}-j)+(j+1) = n_{Y|\emptyset}+1$.
To resolve the situation, we partition each $T^{(j)}$ into two
tables whose projections onto $X$ are equal-sized. Overall, we need $k = 2\log |T|$.
\ep

By applying the above lemma on $T(\mv A_Y)\defeq \Pi_Y(R)$, we show that
$R$ can be partitioned into at most
$(k = 2\log_2 |R|\leq 2\cdot \obj)$ sub-tables $R^{(1)},\dots,R^{(k)}$
such that $N^{(j)}_{X|\emptyset} \cdot N^{(j)}_{Y|X} \leq N_{Y|\emptyset}$, for
all $j \in [k]$, where
\begin{eqnarray}
   N^{(j)}_{X|\emptyset} &\defeq& |\Pi_X(R^{(j)})|,\label{eqn:partition}\\
   N^{(j)}_{Y|X} &\defeq& \max_{\mv t_X\in\Pi_X(R^{(j)})}\deg_{R^{(j)}}(Y | \mv t_X).\nonumber
\end{eqnarray}
For each of these sub-tables $R^{(j)}$ of $R$, we create a subproblem with
the same input tables but with $R$ replaced by $R^{(j)}$ (see line~\ref{alg:ln:subproblems} in Algorithm~\ref{algo:panda}).
The $j$th subproblem has degree constraints $\dc^{(j)}$ where
\begin{equation} \dc^{(j)} = \dc \cup \bigl\{ (\emptyset, X, N^{(j)}_{X|\emptyset}),
(X,Y,N^{(j)}_{Y|X}) \bigr\}.
\label{eq:panda:dec:dcs}
\end{equation}
The table $R^{(j)}$ guards both of the new degree constraints above.
Set $\vec\delta' = \vec\delta + w \cdot \mv d_{Y,X}$.
Since $\mv f = \mv d_{Y,X}$ was the next step in the proof sequence,
$\inner{\vec\lambda,\mv h} \leq \inner{\vec\delta',\mv h}$ is a Shannon flow
inequality with proof sequence $\proofseq'$.
For the $j$-th subproblem, we take
$\inner{\vec\lambda, \mv h}\leq \inner{\vec\delta', \mv h}$
as its Shannon flow inequality
and $\proofseq'$ as its proof sequence.
Invariant (I1) is satisfied in the $j$-th subproblem because the two new terms
$\delta'_{X|\emptyset}$ and $\delta'_{Y|X}$ are supported by corresponding degree
constraints in~\eqref{eq:panda:dec:dcs}.
Moreover, invariant~\eqref{eqn:invar:potential} still holds in the $j$th subproblem because we subtracted $w \cdot n_{Y|\emptyset}$ from the potential, and added
$w \cdot (n^{(j)}_{X|\emptyset}+n^{(j)}_{Y|X}) \leq w \cdot n_{Y|\emptyset}$ to
the potential. Invariant (I4) holds because $N^{(j)}_{X|\emptyset}\leq N_{Y|\emptyset}\leq 2^\obj$.
Invariant~\eqref{eqn:invar:norm-lambda} remains satisfied because $\vec \lambda$ remains the
same above.

{\bf Case 4:} $\mv f=\mv c_{X,Y}$ is a composition step with weight $w$
(line~\ref{alg:ln:case4} in Algorithm~\ref{algo:panda}).
By definition of proof sequence, $\vec\delta + w\cdot\mv c_{X,Y} \geq \mv 0$,
and thus $\delta_{X|\emptyset} \geq w > 0$ and $\delta_{Y|X} \geq w > 0$.
Because $\delta_{X|\emptyset}$ has a support,
there must be an input relation $R$ for which $|\Pi_X(R)| \leq N_{X|\emptyset}$; and because $\delta_{Y|X}$ has a support, there must be two sets $Z \subseteq X$ and $W \subseteq Y$ for which $W-Z=Y-X$ and
$(Z, W, N_{W|Z}) \in \dc$ which is guarded by an input relation $S$.
Note that $X \cup (W-Z) = X \cup (Y-X) = Y$.  We consider two cases:

{\em (Case 4a)} If $N_{X|\emptyset} \cdot N_{W|Z} \leq 2^\obj$, then we
can compute the table
$T(\mv A_Y) \defeq \Pi_X(R) \Join \Pi_W(S)$ by going over all tuples in
$\Pi_X(R)$ and expanding them using matching tuples in $\Pi_W(S)$
(see line~\ref{alg:ln:case4a} in the algorithm).
The runtime of the join is $\tilde O(N_{X|\emptyset}\cdot N_{W|Z}) = \tilde
O(2^{\obj})$, and the size of $T$ is $\leq N_{X|\emptyset}\cdot N_{W|Z}\leq 2^\obj$.
The Shannon flow inequality is modified by setting
$\vec\delta' = \vec\delta + w \cdot \mv c_{X,Y}$, with the
proof sequence $\proofseq'$,
and the set of degree constraints is extended by adding the constraint
$(\emptyset,Y,N_{Y|\emptyset}\defeq |T|)$, guarded by $T$.
Invariant (I1) remains satisfied because the new term $\delta'_{Y|\emptyset}$
is supported by the new degree constraint $(\emptyset,Y,N_{Y|\emptyset})$.
Moreover, invariant~\eqref{eqn:invar:potential} still holds because we subtracted $w \cdot (n_{X|\emptyset}+n_{W|Z})$ and added
at most the same amount to the potential.
Invariant (I4) holds because $|T|\leq 2^\obj$.
Invariant~\eqref{eqn:invar:norm-lambda} remains satisfied because $\vec \lambda$ remains
unchanged.

{\em (Case 4b)} If $N_{X|\emptyset}\cdot N_{W|Z} > 2^\obj$, then we
{\em will not} perform this join.
Instead, we {\em restart} the subproblem with a fresh inequality
(line~\ref{alg:ln:case4b}).
In particular, set $\vec\delta = \vec\delta + w \cdot \mv c_{X,Y}$.
Now we have $\delta_{Y|\emptyset}\geq w$.
We restart the problem with the inequality
$\inner{\vec\lambda', \mv h} \leq
   \inner{\vec\delta',\mv h}$
satisfying the conditions stated in Lemma~\ref{lmm:truncate:sfi}.
Since $\vec\delta' \leq \vec\delta$ and $\delta'_{Y|\emptyset}
\leq \delta_{Y|\emptyset}-w$, the potential is reduced by at least $w \cdot (n_{X|\emptyset}+n_{W|Z}) > w \cdot \obj$.
Since $\norm{\vec\lambda'}_1\geq\norm{\vec\lambda}_1-w$,
the right-hand side of \eqref{eqn:invar:potential} was reduced by at most $w\cdot \obj$, hence
invariant~\eqref{eqn:invar:potential} still holds.
Moreover, let $\sum_{(X,Y)}n(\delta'_{Y|X})$ be the new potential.
Now we have
\[0\leq \sum_{(X,Y)}n(\delta'_{Y|X}) < \sum_{(X,Y)}n(\delta_{Y|X})-w\cdot \obj\leq (\norm{\vec\lambda}_1-w)\cdot \obj \leq \norm{\vec\lambda'}_1\cdot \obj.\]
This proves $\norm{\vec\lambda'}_1 > 0$.
And because $\vec\lambda' \leq \vec\lambda$, we have $\norm{\vec\lambda'}_1 \leq \norm{\vec\lambda}_1\leq 1$. Hence, invariant~\eqref{eqn:invar:norm-lambda} holds.
Invariants (I1) and (I4) hold because $\vec\delta' \leq \vec\delta$ hence we can use the same
degree constraints supporting positive terms in $\vec\delta$ to support the corresponding
terms in $\vec\delta'$.

Finally, we get back to our earlier assumption that invariant (I4) was initially satisfied.
\bprop
Invariant (I4) is satisfied at the beginning of the $\panda$ algorithm.
\label{prop:panda:i4}
\eprop
\bp
If initially there was some positive $\delta_{Y|\emptyset}$ with $n_{Y|\emptyset} > \obj$,
then we could have replaced the original Shannon-flow inequality and witness with the inequality
$\inner{\vec\lambda', \mv h} \leq \inner{\vec\delta',\mv h}$ along with the witness $(\vec\sigma',\vec\mu')$ satisfying the conditions
of Lemma~\ref{lmm:truncate:sfi}.
Similar to Case 4b above, we conclude that $\norm{\vec\lambda'}_1 > 0$
and $\sum_{(X,Y)}n(\delta'_{Y|X}) < \norm{\vec\lambda'}_1 \cdot\obj$.
But this is a contradiction, because
\begin{align*}
   \obj
   &= \max_{h \in \Gamma_n \cap \hdc} \min_{B \in\calB} h(B)\\
   &\leq \max_{h \in \Gamma_n \cap \hdc} \sum_{B \in\calB}
   \frac{\lambda'_B}{\norm{\vec\lambda'}_1} h(B)\\
   &= \frac{1}{\norm{\vec\lambda'}_1}
   \max_{h \in \Gamma_n \cap \hdc} \inner{\vec\lambda', \mv h}\\
   &\leq \frac{1}{\norm{\vec\lambda'}_1}
   \max_{h \in \Gamma_n \cap \hdc} \inner{\vec\delta', \mv h}\\
   &= \frac{1}{\norm{\vec\lambda'}_1}
   \max_{h \in \Gamma_n \cap \hdc} \sum_{(X,Y)\in \dc} \delta'_{Y|X} h(Y|X)\\
   &\leq \frac{1}{\norm{\vec\lambda'}_1}
   \sum_{(X,Y)\in \dc} \delta'_{Y|X} n_{Y|X}\\
   &= \frac{1}{\norm{\vec\lambda'}_1}
   \sum_{(X,Y)\in \dc} n(\delta'_{Y|X})\\
   &< \obj.
\end{align*}
\ep

\subsection{Analysis}
\label{subsec:panda:analysis}

\bp[Proof of Theorem~\ref{thm:main:panda:disjunctive}]
Consider an input proof sequence of length $\ell
\leq D(3\norm{\vec\sigma}_1+\norm{\vec\delta}_1+\norm{\vec\mu}_1)$,
thanks to Theorem~\ref{thm:ps:construction:1}.
If we do not hit Case 4b, then there will be at most $\ell$ steps in the
algorithm, where each step either takes $\tilde O(2^\obj)$ time or spawns
$O(\obj)$ subproblems, for a total of $\tilde O(\poly(\obj)\cdot
2^\obj)$-time.
A subproblem will terminate at producing a relation $T(\mv A_B)$ for some $B\in
\calB$ because the proof sequence will, by definition, reach a point where
$\delta_{B | \emptyset} \geq w > 0$ for some $B \in \calB$.

The worst case is obtained when the algorithm branches as far as possible only
to have to restart at Case 4b with a slightly shorter proof sequence
of length at most
$D(3\norm{\vec\sigma}_1+\norm{\vec\delta}_1+\norm{\vec\mu}_1)-1$, thanks to Lemma~\ref{lmm:truncate:sfi}.
Thus, overall the exponent of $\obj$ (in $\poly(\obj)$) will be
$\leq \frac 1 2
D^2(3\norm{\vec\sigma}_1+\norm{\vec\delta}_1+\norm{\vec\mu}_1)^2$.
(Note that this constant is data-independent because the optimal dual solution
$(\vec\delta,\vec\sigma,\vec\mu)$ can be taken to be an extreme point of the
dual polyhedron \eqref{eqn:dual:polyhedron}, whose constraints are only dependent on the input query.)
Since $\obj$ equals the polymatroid bound $\outputbound_{\Gamma_n\cap\hdc}(P)$, which is bounded
by the vertex bound $\log (N^n)$, and $n$ is a constant in data complexity, we have
$\obj=O(\log N)$, and the runtime is
$\tilde O(\poly(\log N)\cdot
2^{\outputbound_{\Gamma_n\cap\hdc}(P)})$, as desired.
\ep

\section{Degree-Aware width parameters and algorithms}
\label{sec:it}

This section explains how $\panda$ can be used to evaluate queries within the
analogs of fractional hypertree width and submodular width under general degree
constraints.
In particular, we shall present a proof of
Theorem~\ref{thm:main:panda:submodular}. Towards this goal, we need to
generalize traditional width parameters to handle degree constraints.
The typical definitions of these width parameters
(Definition~\ref{defn:traditional:widths})  do not generalize in any obvious way
to handle functional dependencies, let alone degree constraints.
Hence, our first task in Section~\ref{subsec:minimax:maximin} is to reformulate
known width parameters under the information theoretic view. Our reformulation
leads naturally to generalizations, described in Section~\ref{subsec:da:widths}
with degree constraints taken into account.
Section~\ref{subsec:bound:summary} demonstrates the utility of our formulation
by summarizing all major known bounds and widths under the same umbrella.
Finally, Section~\ref{subsec:specializations} shows how $\panda$ is used to
achieve a runtime predicted by these new degree-aware widths.

\subsection{Minimax and maximin widths}
\label{subsec:minimax:maximin}

We slightly reformulate existing width parameters under a common framework.
Recall from Section~\ref{sec:background}  that there are two classes of
width parameters: the first class captures algorithms
seeking the best tree decomposition with the
worst bag runtime, while the second class captures algorithms
adapting the tree decomposition to the instance at hand.
In the definitions below, the maximin width notion is from
Marx~\cite{MR3144912,DBLP:journals/mst/Marx11}.

\bdefn
Let $\calF$ denote a topologically closed class of non-negative set functions
on $\calV$. The {\em $\calF$-minimax width}
and {\em $\calF$-maximin width} of a query $Q$ are defined by
{\small
\begin{eqnarray}
   \minimaxw_\calF(Q) &\defeq& \min_{(T,\chi)} \max_{t\in V(T)} \max_{h\in\calF}
h(\chi(t)), \label{eqn:minimax}\\
   \maximinw_\calF(Q) &\defeq& \max_{h\in\calF}\min_{(T,\chi)} \max_{t\in V(T)}
h(\chi(t)).
\label{eqn:maximin}
\end{eqnarray}
}
\label{defn:maximin:minimax}
\edefn

The following observation is straightforward:
\blmm
If $\calG \subseteq \calF$ are two topologically closed classes of functions, then
\begin{eqnarray*}
   \minimaxw_\calG(\calH) &\leq& \minimaxw_\calF(\calH)\\
   \maximinw_\calG(\calH) &\leq& \maximinw_\calF(\calH).
\end{eqnarray*}
For a fixed class $\calF$ of functions, we have
$\maximinw_\calF(\calH) \leq \minimaxw_\calF(\calH)$.
\label{lmm:minimax:maximin}
\elmm
\bp
We prove the only non-trivial inequality that
$\maximinw_\calF(\calH)\leq \minimaxw_\calF(\calH)$.
From definition:
\begin{eqnarray*}
   \maximinw_\calF(\calH) &=&
   \max_{h\in\calF}\min_{(T,\chi)} \max_{t\in V(T)} h(\chi(t))\\
   &\leq&
   \min_{(T,\chi)}
   \max_{h\in\calF}
   \max_{t\in V(T)} h(\chi(t))\\
   &=&
   \min_{(T,\chi)}
   \max_{t\in V(T)}
   \max_{h\in\calF}
   h(\chi(t))\\
   &=&\minimaxw_\calF(\calH).
\end{eqnarray*}
\ep

These width notions are used by specializing $\calF$ to capture two aspects of
the input. The first aspect uses either the class of entropic functions or its
relaxation, coming from the chain of inclusion
$\Mod_n \subset \overline\Gamma^*_n \subset \Gamma_n \subset \sa_n$
(See Figure~\ref{fig:set:functions}).
The second aspect models the granularity level of statistics we know from the
input database instance, with the following inclusion chain
\begin{equation}
\hdc \subset
\hcc \subset
\ed \cdot \log N \subset
\vd \cdot \log N.
\label{eqn:data-model-chain}
\end{equation}
(Recall notation from~\eqref{eqn:hdc} and Section~\ref{sec:background}.)
Note that the bounds in the constraints $\hdc$ are not normalized as in the sets
$\ed$ or $\vd$ in the traditional width parameters.
This is because normalizing degree constraints makes them less general than they can be, and it does
not make practical sense to assume that all degree bounds to be the same! (For
example, the FD-based degree bounds are always $0$, while the
relation-size-based degree bounds are $\log_2 N_F$.)
Consequently, we used the $\log_2N$ scaled up versions of the traditional
width parameters to compare with our new width parameters.

From these specializations, the minimax and maximin
widths capture all width parameters we discussed in
Section~\ref{sec:background}, summarized in the following proposition.
\bprop\label{prop:unifying}
Let $Q$ be a conjunctive query with only cardinality constraints (no FDs nor proper degree constraints) whose hypergraph
is $\calH=(\calV,\calE)$.  Then the followings hold (Recall notation from Section~\ref{sec:background}):
\begin{align}
   1+\tw(\calH) \nonumber
   &= \minimaxw_{\calF \cap \vd}(Q) \\
   &=  \maximinw_{\calF \cap \vd}(Q)
   &\forall \calF \in \{\Mod_n, \overline \Gamma^*_n, \Gamma_n, \sa_n\}
   \label{eqn:u5} \\
   \ghtw(\calH)
   &= \minimaxw_{\sa_n \cap \ed}(Q)\nonumber\\
   &= \maximinw_{\sa_n \cap \ed}(Q) \label{eqn:u6}\\
   \fhtw(\calH)
   &= \minimaxw_{\calF \cap \ed}(Q)
   & \forall \calF \in \{\Mod_n, \overline \Gamma^*_n, \Gamma_n\}
   \label{eqn:u7}\\
   \subw(\calH)
   &= \maximinw_{\Gamma_n \cap \ed}(Q)\label{eqn:u8} \\
   \adw(\calH)
   &= \maximinw_{\Mod_n \cap \ed}(Q)\label{eqn:u9}.
\end{align}
\eprop
\bp
To prove~\eqref{eqn:u5}, note that for $h \in \sa_n$ and
any set $F \subseteq \calV$, from subadditivity we have $h(F) \leq
\sum_{v\in F}h(v) \leq |F|$.
Hence, recalling the definition of tree-width from Section~\ref{subsec:no-FD-HDC},
we have
\begin{eqnarray*}
   \minimaxw_{\sa_n \cap \vd}(Q)
   &=& \min_{(T,\chi)} \max_{t\in V(T)} \max_{h\in \sa_n \cap \vd} h(\chi(t))\\
   \text{(because $h \in \sa_n$)} &\leq& \min_{(T,\chi)} \max_{t\in V(T)} \max_{h\in \sa_n \cap \vd} \sum_{v \in \chi(t)}h(v)\\
   \text{(because $h \in \vd$~\eqref{eqn:vd})} &\leq& \min_{(T,\chi)} \max_{t\in V(T)} |\chi(t)|\\
   \text{(by Definition~\ref{defn:traditional:widths})}&=& \tw(\calH)+1.
\end{eqnarray*}
To show the reverse, define the function
$\bar h(F)=|F|$ for all $F \subseteq \calV$.
This function is modular and vertex-dominated, and thus
\begin{eqnarray*}
   \tw(\calH)+1 &\geq&
   \minimaxw_{\sa_n \cap \vd}(Q)\\
  &\geq&
   \minimaxw_{\Mod_n \cap \vd}(Q)\\
   &\geq&
   \maximinw_{\Mod_n \cap \vd}(Q)\\
   &=& \max_{h\in \Mod_n \cap \vd}\min_{(T,\chi)} \max_{t\in V(T)}
   h(\chi(t))\\
   &\geq& \min_{(T,\chi)} \max_{t\in V(T)} \bar h(\chi(t))\\
   &=& \tw(\calH)+1.
\end{eqnarray*}
Thus the bound hierarchy collapses for the $\vd$ constraints.

The first equality in~\eqref{eqn:u6} is proved similarly to that of
identity~\eqref{eqn:u2}.
To prove the second equality in~\eqref{eqn:u6}, define the function $\bar h(B)=\rho(B)$
for all $B\subseteq \calV$. This function is in $\sa_n$ and is also edge-dominated.
Then
\begin{eqnarray*}
   \maximinw_{\sa_n \cap \ed}(Q)
   &=& \max_{h\in \sa_n \cap \ed} \min_{(T,\chi)}\max_{t\in
\chi(t)} h(\chi(t))\\
   &\geq & \min_{(T,\chi)}\max_{t\in \chi(t)} \bar h(\chi(t))\\
   &= &  \ghtw(\calH)\\
&= &  \minimaxw_{\sa_n \cap \ed}(Q) .
\end{eqnarray*}

Identity~\eqref{eqn:u7} follows immediately from Lemma~\ref{lmm:modularization},
picking $B$ to be the bag with the worst-case
bound for any tree decomposition.
Identities~\eqref{eqn:u8} and~\eqref{eqn:u9} are just their definitions.
\ep

Note from~\eqref{eqn:u6}
the interesting fact that the set $\sa_n$ is so large that switching to the
$\maximinw$ does not help reduce the objective.
While $\sa_n$ yields too large of an upperbound, and
$\Mod_n$ only yields a lowerbound, depending on the constraints
we want to impose, some parts of the hierarchy collapse.

To further illustrate the strength of the minimax/maximin
characterization, Example~\ref{ex:gap1} below uses the above
characterization to bound the $\subw$ of a cycle query,
and to show that the gap between $\fhtw$ and $\subw$ is
unbounded.
(Marx~\cite{DBLP:journals/mst/Marx11,MR3144912}
already constructed a hypergraph where $\subw$ is bounded and
$\fhtw$ is unbounded, which is a stronger result; on the other hand, his
example is much more involved.)

\begin{ex}[Unbounded gap between $\fhtw$ and $\subw$]\label{ex:gap1}
   Consider a query whose graph $\calH = (V,\calE)$ and is defined as follows.
   The vertex set $V = I_1 \cup I_2 \cup \cdots \cup I_{2k}$ is a disjoint
   union of $2k$ sets of vertices. Each set $I_j$ has $m$ vertices in it.
   There is no edge between any two vertices within the set $I_j$ for every
   $j \in [2k]$, namely $I_j$ is an independent set.
   The edge set $\calE$ of the hypergraph is the union of $2k$ complete
   bipartite graphs $K_{m,m}$:
   \[ \calE \defeq I_1 \times I_2 \cup I_2 \times I_3 \cup \cdots \cup
   I_{2k-1} \times I_{2k} \cup I_{2k} \times I_1.
   \]
   To each edge in $\calE$ there corresponds an input relation.

   We first bound the $\fhtw$ of this graph.
   Consider any non-redundant tree decomposition $(T,\chi)$ of $\calH$ with
   fractional hypertree width equal to $\fhtw(\calH)$.
   Let $t \in V(T)$ be a leaf node of the tree $T$ and $t'$ be the (only)
   neighbor of $t$. Due to non-redundancy, there must be a vertex $v \in \chi(t)$
   such that $v \notin \chi(t')$.
   It follows that $\chi(t)$ is the only bag in the tree decomposition containing
   $(v,u)$, for any $u$ for which $(v,u) \in \calE$.
   This means $\chi(t)$ contains the entire neighborhood of $v$ in the graph
   $\calH$.
   The neighborhood of every vertex
   contains $2m+1$ vertices: the vertex and its neighboring two independent sets
   for a total of $2m$ independent vertices.
   Since no edge can cover two vertices of an independent set,
   the best fractional cover bound for $\chi(t)$ is at least $2m$,
   namely $\fhtw(\calH) \geq 2m$.

   Next, we bound $\calH$'s submodular width. Let $h$ be any submodular function.
   \bi
\item Case 1: $h(I_i) \leq \theta$ for some $i \in [2k]$. WLOG assume $h(I_1) \leq
   \theta$. Consider the tree decomposition

      
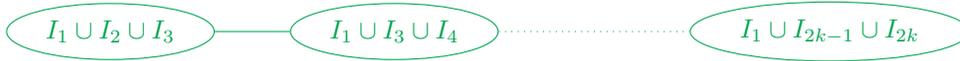
\begin{tikzpicture}[color=Green]
         \node[draw,ellipse] (123) {$I_1 \cup I_2 \cup I_3$};
         \node[draw,ellipse, right = of 123] (134) {$I_1 \cup I_3 \cup I_4$};
         \node[draw,ellipse, right = 1in of 134] (1k) {$I_1 \cup I_{2k-1} \cup I_{2k}$};
         \draw (123) -- (134);
         \draw[dotted] (134) -- (1k);
      \end{tikzpicture}

      For bag $B=I_1 \cup I_i \cup I_{i+1}$,
      \[ h(B) \leq h(I_1)+h(I_i \cup I_{i+1}) \leq \theta+m. \]
\item Case 2: $h(I_i) > \theta$ for all $i \in [2k]$. Consider the tree
   decomposition

      
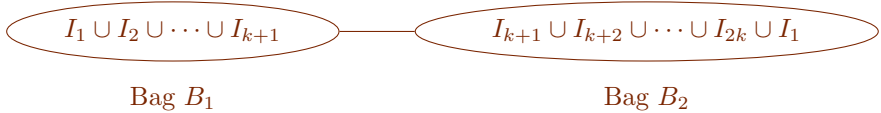
\begin{tikzpicture}[color=Brown]
         \node[draw,ellipse] (B1) {$I_1 \cup I_2 \cup \cdots \cup I_{k+1}$};
         \node[draw,ellipse, right = of B1] (B2) {$I_{k+1} \cup I_{k+2} \cup \cdots \cup I_{2k} \cup I_1$};
         \draw (B1) -- (B2);
         \node[below = 0.21 of B1] {Bag $B_1$};
         \node[below = 0.21 of B2] {Bag $B_2$};
      \end{tikzpicture}

   From submodularity, it is easy to see that
   \begin{eqnarray*}
      h(B_1) \leq & h(I_1\cup I_2) + \displaystyle{\sum_{i=3}^{k+1}h(I_i\cup I_{i-1}|I_{i-1})} & \leq km - (k-1) \theta \\ 
      h(B_2) \leq & h(I_{2k}\cup I_{1}) + \displaystyle{\sum_{i=k+1}^{2k-1}h(I_i\cup I_{i+1}|I_{i+1})} & \leq km - (k-1) \theta
   \end{eqnarray*}
   Thus, by setting $\theta = (1-1/k)m$ we have just proved that
   $\subw(\calH) \leq m(2-1/k).$ By increasing $m$, the gap between
   $m(2-1/k)$ and $2m$ is infinite for a fixed $k \geq 2$.  \ei
\end{ex}

\bcor
Let $Q$ be a conjunctive query (with no FDs nor degree constraints) whose
hypergraph is $\calH$, then
$
 1+\tw(\calH) \geq \ghtw(\calH) \geq \fhtw(\calH) \geq \subw(\calH) \geq
 \adw(\calH).
$
Moreover, the gap between any two consecutive entries in the above series
is unbounded.
\ecor

\subsection{New width parameters}
\label{subsec:da:widths}

Using the maximin and minimax formalism, we easily extend the traditional width
parameters to handle general degree constraints.
As shown by Theorem~\ref{thm:glvv:not:tight} we know
that there is a gap between the polymatroid bound and the entropic bound; and
hence it is natural to use $\overline \Gamma^*_n$ itself instead of some
approximation of it.

\bdefn
We define the following width parameters for queries $Q$ with degree constraints
$\dc$. The first two parameters are generalizations of $\fhtw$ and $\subw$
under degree constraints, and the last two are their entropic versions:
\begin{eqnarray}
   \dafhtw(Q)  &\defeq& \minimaxw_{\Gamma_n \cap \hdc}(Q) \label{eqn:dafhtw}\\
   \dasubw(Q)  &\defeq& \maximinw_{\Gamma_n \cap \hdc}(Q) \label{eqn:dasubw}\\
   \edafhtw(Q) &\defeq& \minimaxw_{\overline\Gamma^*_n \cap \hdc}(Q)\label{eqn:edafhtw}\\
   \edasubw(Q) &\defeq& \maximinw_{\overline\Gamma^*_n \cap \hdc}(Q)\label{eqn:edasubw}.
\end{eqnarray}
\label{defn:dafhtw:dasubw:edafhtw:edasubw}
\edefn
({\sf da} stands for ``degree-aware'', and {\sf eda} for ``entropic
degree-aware''.)
The following relationships hold between these four quantities.
\bprop For any query $Q$ with degree constraints
\[
\begin{matrix}
   \edasubw(Q) & \leq & \edafhtw(Q)\\
   \rotatebox[origin=c]{-90}{$\leq$} &  & \rotatebox[origin=c]{-90}{$\leq$}\\
   \dasubw(Q) & \leq & \dafhtw(Q).
\end{matrix}
\]
The quantities $\edafhtw(Q)$ and $\dasubw(Q)$ are not comparable.
The gap between the two sides of any of the above four inequalities can be made
arbitrarily large by some input.
\label{prop:4:quantities}
\eprop
\bp
The inequalities above follow from~\eqref{eqn:dafhtw}-\eqref{eqn:edasubw},
Lemma~\ref{lmm:minimax:maximin} and the fact that
$\overline\Gamma^*_n \subseteq \Gamma_n$.
Next we prove that $\edafhtw(Q)$ and $\dasubw(Q)$ are incomparable.
Let $\zy$ denote the Zhang-Yeung query defined in the proof of
Theorem~\ref{thm:glvv:not:tight}.
Add one additional relation $R_\calV$ to the query, with a very large
relation size bound $|R_\calV| \geq N^4$, where $N$ is the one used in the proof
of Theorem~\ref{thm:glvv:not:tight}.
Call the resulting query $\zy^+$.
Because $\zy^+$ has one relation $R_\calV$ involving all the variables $\calV$,
it has only one non-redundant tree decomposition $(T, \chi)$ and $(T, \chi)$ has only one bag
containing all the variables, i.e. $V(T) = \{t\}$ and $\chi(t) = \calV$.
Therefore,
\begin{multline*}
   \edasubw(\zy^+) = \edafhtw(\zy^+) =
    \outputbound_{\overline\Gamma^*_n \cap \hdc}(\zy^+) <\\
   \outputbound_{\Gamma_n \cap \hdc}(\zy^+) = \dasubw(\zy^+)=\dafhtw(\zy^+).
\end{multline*}
The equalities above follow from the fact that $\zy^+$ has only one non-redundant tree decomposition with
only one bag.
The inequality follows from the proof of Theorem~\ref{thm:glvv:not:tight}.
In fact, with the gap amplification trick we showed in the proof of
Theorem~\ref{thm:glvv:not:tight}, we can show that there are queries for which
the gap is as large as one wants.

On the other hand, for the $4$-cycle query $C_4$, we show
in Example~\ref{ex:4:cycle:edasubw} that
$2\log N = \edafhtw(C_4) > 3/2\log N = \dasubw(C_4)$.
\ep

\begin{ex}[Computing $\dafhtw,\dasubw,\edafhtw,$ and $\edasubw$ for a 4-cycle]
   Consider the $4$-cycle query $C_4$ shown in equation~\eqref{eq:c4},
   which has no proper degree bounds (only input size bounds).
   Let $N$ be the common upperbound on all input relations' sizes, we will show
   that $\dasubw(C_4)=\edasubw(C_4)=3/2\log N$ and
   $\dafhtw=\edafhtw(C_4)=2\log N$.

   Since there is no degree bound,
   $\dafhtw(C_4)=\fhtw(C_4) \cdot \log N = 2\log N$,
   and
   $\dasubw(C_4)=\subw(C_4) \cdot \log N \leq 3/2 \log N$
   (see Example~\ref{ex:gap1}).
   To show $\dasubw(C_4) \geq 3/2 \log N$,
   consider the function $\bar h(F) = \frac{|F|}{2}\log N$, for all $F\subseteq
   [4]$. This function is in $\Mod_4$ hence it is also in $\Gamma^*_4$
   (Recall~\eqref{eq:chain:M-SA-} and Figure~\ref{fig:set:functions}). Moreover, for any bag $B$ of the two tree
   decompositions shown in Figure~\ref{fig:tree:decomposition}, $h(B)=3/2\log N$.
   Hence, $3/2\log N \leq \edasubw(C_4) \leq \dasubw(C_4) \leq 3/2\log N$.
   The fact that $\edafhtw(C_4)=2\log N$ can be shown similarly.
   In particular, if we choose either tree decomposition in Figure~\ref{fig:tree:decomposition},
   there exists a function $h\in\Mod_4 \subsetneq \Gamma^*_4$ where $h(B) = 2 \log N$ for either bag $B$
   of the chosen tree decomposition.
   For tree decomposition 1, there is $h\in\Mod_4$ defined by $h(1)=h(3) = \log N$,
   $h(2)=h(4) = 0$, and $h(F) = \sum_{i \in F} h(i)$ otherwise.
   For tree decomposition 2, there is $h\in\Mod_4$ defined by $h(2)=h(4) = \log N$,
   $h(1)=h(3) = 0$, and $h(F) = \sum_{i \in F} h(i)$ otherwise.
   Hence, $2\log N \leq \edafhtw(C_4) \leq \dafhtw(C_4) \leq 2\log N$.
   \label{ex:4:cycle:edasubw}
\end{ex}

Another somewhat interesting observation which follows from the above is
the following.
Due to the fact that every non-negative modular set function
is entropic, we have
\bcor
When $Q$ has only edge domination constraints $\ed$
(i.e. no FD nor proper degree bounds), we have $\adw(Q) \leq \edasubw(Q)$.
\label{cor:adw:edasubw}
\ecor
Following Marx~\cite{DBLP:journals/mst/Marx11,MR3144912}, for these
queries with only $\ed$ constraints, we have
$\dasubw(Q) = \subw(Q) = O(\adw^4(Q)) = O(\edasubw^4(Q))$.
Thus, when there is only $\ed$ constraints, if a class of queries
has bounded $\edasubw$, then it has bounded $\dasubw$.

\subsection{Summary of known bounds and width parameters}
\label{subsec:bound:summary}

We have mentioned quite a few known and proved new bounds in this paper. The
bounds can be summarized systematically as follows. Each bound is identified by
coordinates $(X, Y, Z)$.
The $X$-axis represents the entropy approximation that is being used: one
starts from the desired target $\overline\Gamma^*_n$, then relaxes it to
$\Gamma_n$ and $\sa_n$. The inclusion chain is $\overline\Gamma^*_n \subset
\Gamma_n \subset \sa_n$.
The $Y$-axis represents the constraints we can extract from the input
database instance, where we can go from bounding domain sizes, relation sizes,
to incorporating more refined degree bounds and functional dependencies.
One chain of inclusion was given by~\eqref{eqn:data-model-chain}.
The $Z$-axis represents the level of sophistication of the query plan
that is being considered in this bound. The simplest query plan just joins
everything together without computing any tree decomposition -- or,
equivalently, this is the plan that uses the trivial tree decomposition with
one bag containing all attributes. (Recall the bounds $\daentropic(Q)$ and $\dapolymatroid(Q)$ from~\eqref{eqn:daeb:dbpb}.) Then, one can get more sophisticated by
computing a tree decomposition then computing its bags. And, lastly the query
plan can also be adaptive to the input instance, yielding the submodular-width
style of complexity.
The bounds are summarized in Figure~\ref{fig:diagram}.
If bound $A$ has coordinates $(X_A, Y_A, Z_A)$ that are less
than the corresponding coordinates $(X_B, Y_B,Z_B)$ of bound $B$
(i.e. if $X_A\leq X_B\wedge Y_A\leq Y_B\wedge Z_A\leq Z_B$), then bound
$A$ $\leq$ bound $B$.

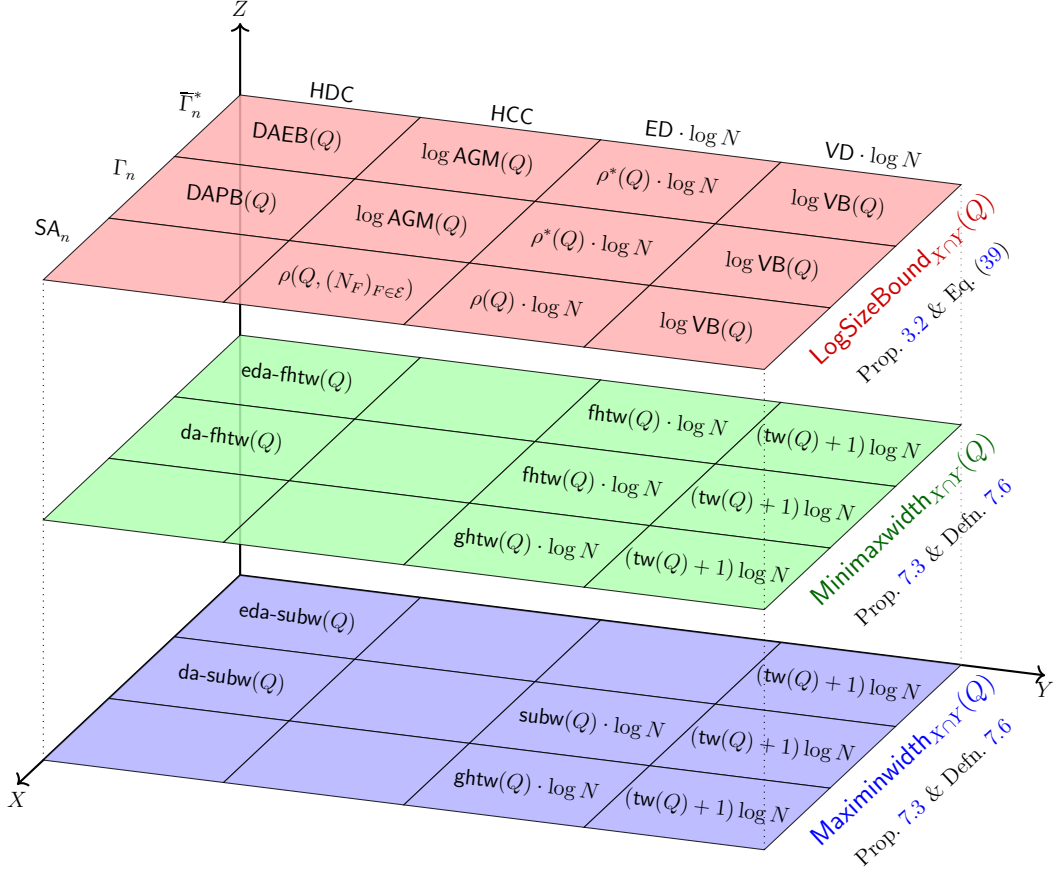
\begin{figure}[!tbp]
\centering{
\tdplotsetmaincoords{70}{110}
   \begin{tikzpicture}[tdplot_main_coords,,scale=0.85, every node/.style={scale=0.75}]
\pgfmathsetmacro{\rot}{-6.8}
\pgfmathsetmacro{\rott}{43}

\draw[thick,->,black] (0,0,0) -- (10.2,0,0) node[anchor=north]
{\color{black}\large $X$};
\draw[thick,->,black] (0,0,0) -- (0,13.4,0) node[anchor=north]
{\color{black}\large $Y$};
\draw[thick,->,black] (0,0,0) -- (0,0,9.2) node[anchor=south]
{\color{black}\large $Z$};

\foreach \x in {0,3,...,6}
	\foreach \y in {0,3,...,9}
		{\draw[opacity=.85,fill=blue!30!white] (\x,\y,0)--(\x+3,\y,0)--(\x+3,\y+3,0)--(\x,\y+3,0)--cycle;}
\foreach \x in {0,3,...,6}
	\foreach \y in {0,3,...,9}
		{\draw[opacity=.85,fill=green!30!white] (\x,\y,4)--(\x+3,\y,4)--(\x+3,\y+3,4)--(\x,\y+3,4)--cycle;}
\foreach \x in {0,3,...,6}
	\foreach \y in {0,3,...,9}
		{\draw[opacity=.85,fill=red!30!white] (\x,\y,8)--(\x+3,\y,8)--(\x+3,\y+3,8)--(\x,\y+3,8)--cycle;}

\draw[dotted] (9,0,0)--(9,0,8);
\draw[dotted,gray] (0,12,0)--(0,12,8);
\draw[dotted] (9,12,0)--(9,12,8);

\node[anchor=center, above left, rotate=\rot]at (1.5,0,8) {\large$\overline\Gamma^*_n$};
\node[anchor=center, above left, rotate=\rot]at (4.5,0,8) {\large$\Gamma_n$};
\node[anchor=center, above left, rotate=\rot]at (7.5,0,8) {\large$\sa_n$};

\node[anchor=center, above, rotate=\rot]at (0,1.5,8) {\large$\hdc$};
\node[anchor=center, above, rotate=\rot]at (0,4.5,8) {\large$\hcc$};
\node[anchor=center, above, rotate=\rot]at (0,7.5,8) {\large$\ed \cdot \log N$};
\node[anchor=center, above, rotate=\rot]at (0,10.5,8) {\large$\vd \cdot \log N$};

\node[anchor=center, align=center, rotate=\rott]at (5,13.1,8)
{\Large\color{red!80!black}$\outputbound_{X \cap Y}(Q)$\\\\
\large Prop.~\ref{prop:unifying:bound} \& Eq.~\eqref{eqn:daeb:dbpb}};

\node[anchor=center, align=center, rotate=\rot]at (7.5,10.5,8)
{\large$\log \vertexbound(Q)$};
\node[anchor=center, align=center, rotate=\rot]at (4.5,10.5,8)
{\large$\log \vertexbound(Q)$};
\node[anchor=center, align=center, rotate=\rot]at (1.5,10.5,8)
{\large$\log \vertexbound(Q)$};

\node[anchor=center, align=center, rotate=\rot]at (7.5,7.5,8)
{\large$\rho(Q) \cdot \log N$};
\node[anchor=center, align=center, rotate=\rot]at (4.5,7.5,8)
{\large$\rho^*(Q) \cdot \log N$};
\node[anchor=center, align=center, rotate=\rot]at (1.5,7.5,8)
{\large$\rho^*(Q) \cdot \log N$};

\node[anchor=center, align=center, rotate=\rot]at (7.5,4.5,8)
{\large$\rho(Q, (N_F)_{F\in\calE})$};
\node[anchor=center, align=center, rotate=\rot]at (4.5,4.5,8)
{\large$\log \agm(Q) $};
\node[anchor=center, align=center, rotate=\rot]at (1.5,4.5,8)
{\large$ \log \agm(Q) $};

\node[anchor=center, align=center, rotate=\rot]at (4.5,1.5,8)
{\large$\dapolymatroid(Q)$};
\node[anchor=center, align=center, rotate=\rot]at (1.5,1.5,8)
{\large$\daentropic(Q)$};

\node[anchor=center, align=center, rotate=\rott]at (5,13.1,4)
{\Large\color{green!40!black}$\minimaxw_{X\cap Y}(Q)$\\\\
\large Prop.~\ref{prop:unifying} \& Defn.~\ref{defn:dafhtw:dasubw:edafhtw:edasubw}};

\node[anchor=center, align=center, rotate=\rot]at (7.5,10.5,4)
{\large$(\tw(Q)+1)\log N$};
\node[anchor=center, align=center, rotate=\rot]at (4.5,10.5,4)
{\large$(\tw(Q)+1)\log N$};
\node[anchor=center, align=center, rotate=\rot]at (1.5,10.5,4)
{\large$(\tw(Q)+1)\log N$};

\node[anchor=center, align=center, rotate=\rot]at (7.5,7.5,4)
{\large$\ghtw(Q)\cdot \log N$};
\node[anchor=center, align=center, rotate=\rot]at (4.5,7.5,4)
{\large$\fhtw(Q)\cdot \log N$};
\node[anchor=center, align=center, rotate=\rot]at (1.5,7.5,4)
{\large$\fhtw(Q)\cdot \log N$};

\node[anchor=center, align=center, rotate=\rot]at (4.5,1.5,4)
{\large$\dafhtw(Q)$};
\node[anchor=center, align=center, rotate=\rot]at (1.5,1.5,4)
{\large$\edafhtw(Q)$};

\node[anchor=center, align=center, rotate=\rott]at (5,13.1,0)
{\Large\color{blue}$\maximinw_{X\cap Y}(Q)$\\\\
\large Prop.~\ref{prop:unifying} \& Defn.~\ref{defn:dafhtw:dasubw:edafhtw:edasubw}};

\node[anchor=center, align=center, rotate=\rot]at (7.5,10.5,0)
{\large$(\tw(Q)+1)\log N$};
\node[anchor=center, align=center, rotate=\rot]at (4.5,10.5,0)
{\large$(\tw(Q)+1)\log N$};
\node[anchor=center, align=center, rotate=\rot]at (1.5,10.5,0)
{\large$(\tw(Q)+1)\log N$};

\node[anchor=center, align=center, rotate=\rot]at (7.5,7.5,0)
{\large$\ghtw(Q)\cdot \log N$};
\node[anchor=center, align=center, rotate=\rot]at (4.5,7.5,0)
{\large$\subw(Q)\cdot \log N$};

\node[anchor=center, align=center, rotate=\rot]at (4.5,1.5,0)
{\large$\dasubw(Q)$};
\node[anchor=center, align=center, rotate=\rot]at (1.5,1.5,0)
{\large$\edasubw(Q)$};

\end{tikzpicture}
}
\caption{A hierarchy of bounds: Each bound corresponds to an entry in a three-dimensional space.
On the $Z$-axis, we have three levels:
The top level (in red) depicts $\outputbound_{X\cap Y}(Q)$ (where $X$ and $Y$ are the $X$- and $Y$- coordinates),
the middle level (in green) depicts $\minimaxw_{X\cap Y}(Q)$,
and the lowest level (in blue) depicts $\maximinw_{X\cap Y}(Q)$.
On the $X$-axis, we have three different coordinates:
$\sa_n$, $\Gamma_n$, and $\overline\Gamma^*_n$, and they are ordered:
$\sa_n\supset \Gamma_n\supset \overline\Gamma^*_n$.
On the $Y$-axis, we have four different ordered coordinates:
$\vd\cdot\log N\supset\ed\cdot\log N\supset\hcc\supset\hdc$.
Notations are defined in Section~\ref{sec:background}.
For example, on the top level of the $Z$-axis (the red level),
the bound whose $X$-coordinate is $\Gamma_n$ and 
whose $Y$-coordinate is $\hcc$ should be $\outputbound_{\Gamma_n\cap\hcc}(Q)$,
which by Proposition~\ref{prop:unifying:bound} corresponds to $\log\agm(Q)$.
Similarly on the lowest level of the $Z$-axis (the blue level),
the bound whose $X$-coordinate is $\Gamma_n$ and whose $Y$-coordinate is $\ed\cdot\log N$ should be $\maximinw_{\Gamma_n\cap (\ed\log N)}(Q)$,
which by Proposition~\ref{prop:unifying} corresponds to $\subw(Q)\cdot\log N$.
This diagram satisfies the following property:
{\it Given a bound $A$ whose coordinates are $(X_A,Y_A,Z_A)$ and a bound $B$ whose coordinates are $(X_B,Y_B,Z_B)$, if $X_A\leq X_B$, $Y_A\leq Y_B$ and $Z_A\leq Z_B$, then bound $A$ $\leq$ bound $B$.}
For example, if bound $A$ has coordinates $(X_A,Y_A,Z_A)=(\Gamma_n, \hdc, \maximinw)$
and bound $B$ has coordinates $(X_B,Y_B,Z_B)=(\Gamma_n, \ed\cdot \log N, \minimaxw)$,
then we can infer that bound $A$ $\leq$ bound $B$.
This is true because bound $A$ corresponds to $\dasubw(Q)$ and bound $B$ corresponds to $\fhtw(Q)\cdot\log N$.
From the same property, we can also infer that the bound $\edasubw(Q)$ is $\leq$ any other bound in the diagram since it has the smallest coordinates $(\overline\Gamma^*_n,\hdc,\maximinw)$ on all three axes.
For more details and insights about this diagram, see Section~\ref{subsec:bound:summary}.}
\Description{A hierarchy of bounds: Each bound corresponds to an entry in a three-dimensional space.}
\label{fig:diagram}
\end{figure}

\subsection{Achieving degree-aware width parameters}
\label{subsec:specializations}

With increasing levels of complexity, the corollaries below
explain how $\panda$ can be used to evaluate a (full or Boolean)
conjunctive query achieving the degree-aware polymatroid size
bound defined in~\eqref{eqn:daeb:dbpb},
the degree-aware version of the fractional hypertree width defined
in~\eqref{eqn:dafhtw}, and
the degree-aware version of the submodular width defined
in~\eqref{eqn:dasubw}.
(In Section~\ref{sec:conclusion}, we briefly show how to use
$\panda$ to solve other conjunctive queries as well as aggregate queries.)

A basic fact that we employ in this section is the following: given an
$\alpha$-acyclic query (with no FD nor degree constraints), the query can be
computed in time
$\tilde O(|\text{\sf input}| + \outputsize)$
\cite{DBLP:conf/vldb/Yannakakis81,DBLP:journals/jacm/FlumFG02}.
In particular, let $(T,\chi)$ be some tree decomposition, and suppose we
have a query for which every bag $B$ of the tree decomposition
corresponds uniquely to an input relation $R_B$; then, the query is
$\alpha$-acyclic (the tree decomposition is the join tree of the query),
and it can be computed in linear
time in the input plus output sizes (modulo a $\log$ factor).

\subsubsection{Conjunctive query with degree constraints}

Consider a full conjunctive query $Q$ with degree constraints $\hdc$.
Since a full conjunctive query is just the disjunctive datalog
rule~\eqref{eqn:disjunctive:datalog:query} with $\calB=\{\calV\}$,
we have an immediate corollary.
However, we spell out more details than necessary here for the reader to get a
sense of how the Shannon flow inequality looks in this simple case.
In this case, $\outputbound_{\Gamma_n\cap\hdc}(Q)$ specializes to
$\dapolymatroid(Q)$ whose LP~\eqref{eqn:daeb:dbpb}
can be written more concretely as
\begin{align}
   \max && h(\calV) \label{eqn:V:target}\\
   \text{s.t.}&& h(Y)-h(X) &\leq n_{Y|X} && (X,Y,N_{Y|X}) \in \dc &
   \text{(degree constraints)} \nonumber\\
   && h(I\cup J | J) - h(I|I\cap J) &\leq 0, && I \incomp J &\text{(submodularity)}\nonumber\\
   && h(X) - h(Y) &\leq 0, &&\emptyset \neq X \subset Y \subseteq \calV
   &\text{(monotonicity)}\nonumber\\
   && h(Z) &\geq 0, && \emptyset \neq Z \subseteq \calV.
   &\text{(non-negativity)}\nonumber
\end{align}
The dual of~\eqref{eqn:V:target} is
\begin{eqnarray}
   \min        &\sum_{(X,Y)\in \dc}n_{Y|X} \cdot \delta_{Y|X} \label{eqn:dual:V:target}\\
   \text{s.t.} &\flow(\calV)          & \geq 1 \nonumber\\
               &\flow(Z)              & \geq 0, \emptyset \neq Z \subset \calV.\nonumber\\
               &(\vec\delta, \vec\sigma, \vec\mu) &\geq\mv 0.\nonumber
\end{eqnarray}
Let $h^*$ denote an optimal solution to~\eqref{eqn:V:target},
and $(\vec\delta^*, \vec\sigma^*, \vec\mu^*)$ a dual-optimal solution, then
$h^*(\calV) = \sum_{(X,Y)\in\dc}\delta^*_{Y|X} \cdot n_{Y|X}$.
The $\dapolymatroid(Q)$ bound can also be written as
\[ \dapolymatroid(Q) = 2^{h^*(\calV)}
= \prod_{(X,Y)\in\dc}2^{n_{Y|X}\delta^*_{Y|X}}
= \prod_{(X,Y)\in\dc}N_{Y|X}^{\delta^*_{Y|X}}.
\]
This was in the same form as the more familiar $\agm(Q)$ bound for queries with no
degree constraints.
(See also Proposition~\ref{prop:unifying:bound}).
In this case, Proposition~\ref{prop:sfi:ddl:target} states the following.
Given $\vec\delta \in \Q_+^\dc$, the inequality
\begin{equation}
   h(\calV) \leq \sum_{(X,Y)\in\dc}\delta_{Y|X} \cdot h(Y|X)
\label{eqn:sfi:V:target}
\end{equation}
is a Shannon flow inequality if and only if there exist $\vec\sigma$ and
$\vec\mu$ such that
$(\vec\delta, \vec\sigma, \vec\mu)$ is feasible to the dual LP~\eqref{eqn:dual:V:target}.
(In particular, the inequality holds when $\vec\delta = \vec\delta^*$.)
Note that inequality~\eqref{eqn:sfi:V:target} implies
the upperbound on $h(\calV)$ we wanted because $h(Y|X) \leq n_{Y|X}$.
Note also that inequality \eqref{eqn:sfi:V:target} has only one target $\calV$.

Let $B\subseteq \calV$ be any fixed set.
Let $\mv e^B = (e(Y|X))_{(X,Y)\in\calP} \in \Q_+^\calP$ denote the unit
vector where $e(B|\emptyset)=1$
and $e(Y|X)=0$ otherwise. Then,
inequality~\eqref{eqn:sfi:V:target} is
$\inner{\mv e^\calV, \mv h} \leq \inner{\vec\delta,\mv h}$.

\bcor
A full or Boolean conjunctive query $Q$ with degree constraints can be solved
using $\panda$ in time $\tilde O(N + \poly(\log N)\cdot 2^{\dapolymatroid(Q)})$.
\label{cor:panda:worst:case}
\ecor
\bp
Let $(\vec\delta^*,\vec\sigma^*,\vec\mu^*)$
denote an optimal solution to the dual~\eqref{eqn:dual:V:target}.
Then, from Proposition~\ref{prop:sfi:ddl:target}
$\inner{\mv e^\calV, \mv h} \leq \inner{\vec\delta^*, \mv h}$
is a Shannon flow inequality.
When feeding this Shannon flow inequality to $\panda$, the quantity $\obj$
defined in~\eqref{eqn:obj} is exactly $\dapolymatroid(Q)$ due to strong duality.
The output of $\panda$ is a single table $T_\calV$ which is a \emph{superset}
of $Q$. To compute $Q$ exactly, we semijoin-reduce $T_\calV$ with every input relation $R_F$,
$F\in\calE$; namely we set $T_\calV = T_\calV \ltimes R(F)$ for all $F\in\calE$.
\ep

\subsubsection{Achieving the degree-aware fractional hypertree width}
\label{app:subsec:panda:dafhtw}

Next let us consider the degree-aware version of the fractional hypertree width
defined in~\eqref{eqn:dafhtw}, whose full definition is
\[ \dafhtw(Q) \defeq \min_{(T,\chi)}\max_{t\in V(T)}\max_{h \in \Gamma_n \cap
   \hdc}h(\chi(t)).
\]
Suppose we want to compute $\dafhtw$, a bruteforce algorithm is to go over
all tree decompositions $(T,\chi)$ of $\calH$;\footnote{There are at most $n!$ non-redundant tree
decompositions, each of which has at most $n$ bags (Proposition~\ref{prop:td:properties}).}
for each $(T,\chi)$ we take a bag $B = \chi(t)$ and solve the
inner most optimization problem:
\footnote{The total number of distinct optimization problems that we have to solve for all $n!$ tree decompositions is $\leq 2^n$, since we have one problem for each distinct choice of $B\subseteq[n]$.}
\begin{equation}
\max \{ h(B) \suchthat h \in \Gamma_n \cap \hdc\} \label{eqn:B:target}
\end{equation}
Thanks to Lemma~\ref{lmm:modularization}, we know
$\max \{ h(B) \suchthat h \in \Gamma_n \cap \hdc\}
= \max \{ h(B) \suchthat h \in \Mod_n \cap \hdc\}$
when all constraints are cardinality constraints, in which case this LP can be
reduced to the LP~\eqref{eqn:V:target} by restricting all functions down to the
universe $B$. However, this does not hold when we add FDs to the query, and it is
certainly not true when we add arbitrary degree constraints.
The dual of~\eqref{eqn:B:target} is the following LP:
\begin{eqnarray}
   \min        &\sum_{(X,Y)\in \dc}n_{Y|X} \cdot \delta_{Y|X} \label{eqn:dual:B:target}\\
   \text{s.t.} &\flow(B)          & \geq 1 \nonumber\\
               &\flow(Z)              & \geq 0, \emptyset \neq Z \subseteq \calV.\nonumber\\
               &(\vec\delta, \vec\sigma, \vec\mu) &\geq\mv 0.\nonumber
\end{eqnarray}
Proposition~\ref{prop:sfi:ddl:target} states the following:
given $\vec\delta  \in \Q_+^\dc$, the inequality
\begin{equation}
   h(B) \leq \sum_{(X,Y)\in\dc}\delta_{Y|X} \cdot h(Y|X)
\label{eqn:sfi:B:target}
\end{equation}
is a Shannon flow inequality if and only if there exist $\vec\sigma$ and
$\vec\mu$ such that the vector $(\vec\delta, \vec\sigma, \vec\mu)$ is feasible to the dual
LP~\eqref{eqn:dual:B:target}.
More compactly,
inequality~\eqref{eqn:sfi:B:target} can be written as
$\inner{\mv e^B, \mv h} \leq \inner{\vec\delta,\mv h}$.

\bcor
A full or Boolean conjunctive query $Q$ with degree constraints can be
solved using
$\panda$ in time
\[\tilde O(N + \poly(\log N)\cdot 2^{\dafhtw(Q)}+\outputsize).\]
\label{cor:panda:dafhtw}
\ecor
\bp
As discussed, a tree decomposition $(T,
\chi)$ of $Q$ for which $\max_{t\in V(T)} \max_{h\in \Gamma_n \cap \hdc}
h(\chi(t)) = \dafhtw(Q)$ can be computed bruteforcely by
going through at most $n!$ tree decompositions and
solving a total of $\leq 2^n$ distinct linear programs,
each of which is data independent (except for a
$\log$-factor in the data to compute the degree bounds in advance).

Now suppose we have already fixed an optimal tree decomposition $(T,\chi)$.
For every bag $B$ of this tree decomposition,
we use $\panda$ to compute a relation $T_B$
for which $\Pi_B(Q) \subseteq T_B$.
After that we set $T_B = T_B \ltimes R_F$ for every $F\in\calE$.
(By definition of a tree decomposition,
for every input relation $R_F$, there must be a bag $B\supseteq F$.)
Finally, $Q$ is the join of all these tables $T_B$, which is now an
$\alpha$-acyclic query solvable in linear time (in input plus output size)
by Yannakakis's algorithm~\cite{DBLP:conf/vldb/Yannakakis81,DBLP:journals/jacm/FlumFG02}.

For a fixed bag $B$ of the optimal tree decomposition,
let $\opt$ be the
optimal objective value of the LP~\eqref{eqn:B:target}.
Let $(\vec\delta^*,\vec\sigma^*,\vec\mu^*)$
denote a dual-optimal solution. Then, from Proposition~\ref{prop:sfi:ddl:target}
we know $\inner{\mv e^B, \mv h} \leq \inner{\vec\delta^*, \mv h}$
is a Shannon flow inequality.
When feeding this Shannon flow inequality to $\panda$, the quantity $\obj$
defined in~\eqref{eqn:obj} is exactly $\opt$ due to strong duality.
\ep

\subsubsection{Achieving the degree-aware submodular width}
\label{app:subsec:panda:dasubw}

The third corollary is on achieving the
the degree-aware submodular width. Unlike the first two corollaries,
proving this requires a couple of new ideas.
In order to compute
$\dasubw$, even in a bruteforce manner, we need an auxiliary lemma,
which is somewhat related to Neumann's minimax
theorem \cite{du1995minimax}.

\blmm
\label{lmm:minimax}
Let $A$ and $B$ be two finite sets, and $f : A \times B \to \R$ be any function.
Let $B^A$ denote the set of all maps from $A$ to $B$.
Then, the following holds:
\[ \min_{a\in A}\max_{b\in B} f(a,b)
   =
   \max_{\beta \in B^A} \min_{a\in A} f(a,\beta(a)).
\]
\elmm
\bp
For any $a\in A$, define $\beta^*(a)\defeq \argmax_{b\in B} f(a, b)$.
\[ \min_{a\in A}\max_{b\in B} f(a, b) =
   \min_{a\in A}f(a, \beta^*(a))\leq \max_{ \beta \in  B^A}
   \min_{a\in A} f(a,\beta(a)).
\]
Conversely,
\[
   \max_{\beta \in  B^A} \min_{a\in A} f(a, \beta(a)) \leq
   \max_{\beta \in B^A} \min_{a\in A} f(a, \beta^*(a))
   = \min_{a\in A} f(a,\beta^*(a))
   = \min_{a\in A}\max_{b\in  B} f(a, b).
\]
\ep

Intuitively, on the LHS we select for each $a \in A$ a neighbor $b$ for which
$f(a,b)$ is maximized; call such neighbor $a$'s ``representative''. Then, we
select the $a$ with the least-weight representative.
On the RHS, we have a ``representative selector'' $\beta$; we pick the $a$-value
with the least-weight selected representative, and then maximize over all selectors.

\bcor[Restatement of Theorem~\ref{thm:main:panda:submodular}]
A full or Boolean conjunctive query $Q$ with degree constraints $\dc$ can be solved using $\panda$ in time
\[\tilde O(N + \poly(\log N)\cdot 2^{\dasubw(Q)}+\outputsize).\]
\label{cor:panda:dasubw}
\ecor
\vspace{-.3in}
\bp
We first apply Lemma~\ref{lmm:minimax} to reformulate~\eqref{eqn:dasubw}.
To this end, we need a few notations.
(Recall from Section~\ref{subsec:no-FD-HDC} that $\td$ denotes the set of all non-redundant tree decompositions of $Q$.)
Let $\calM$ be the set of all maps $\beta : \td \to 2^\calV$,
such that $\beta(T,\chi) = \chi(t)$ for some $t \in V(T)$. In English,
$\beta$ is a ``bag selector'' map that picks out a bag from each tree
decomposition $(T,\chi)$.
Let $\bB$ be the collection of images of all $\beta \in \calM$, i.e.
    \begin{equation}
       \bB = \{ \calB \suchthat \calB = \image(\beta)
                  \text{ for some } \beta \in \calM \}.
       \label{eqn:defn:bB}
    \end{equation}
Using Lemma~\ref{lmm:minimax}, we can rewrite \eqref{eqn:dasubw} as follows.

\begin{eqnarray}
   \dasubw(\calH) &=& \max_{h\in\Gamma_n \cap \hdc} \min_{(T,\chi) \in \td} \max_{t\in V(T)}
   h(\chi(t))\nonumber\\
   (\text{Lemma~\ref{lmm:minimax}})
   &= & \max_{h\in \Gamma_n\cap \hdc} \max_{\beta \in \calM}
   \min_{(T,\chi)} h(\beta(T, \chi))\nonumber\\
   &= &
   \max_{h\in\Gamma_n\cap\hdc}
   \max_{\beta \in \calM}
   \min_{B\in \image(\beta)} h(B) \nonumber\\
   &= &
   \max_{\beta \in \calM}
   \max_{h\in\Gamma_n\cap\hdc}
   \min_{B\in \image(\beta)}
   h(B) \nonumber\\
   &= &
   \max_{\calB \in \bB}
   \max_{h\in\Gamma_n\cap\hdc}
   \min_{B\in \calB}
   h(B) \nonumber\\
   (\text{Lemma~\ref{lmm:lambda:1:reformulation}})
   &= &
   \max_{\calB \in \bB}
   \underbrace{
   \max_{h \in\Gamma_n\cap\hdc}
   \left\{
      \sum_{B \in \calB} \lambda_B h(B)
   \right\}
   }_{\text{Linear program~\eqref{eqn:ddl:target}}}
   \label{eqn:dasubw-maxLPs}
\end{eqnarray}
In \eqref{eqn:dasubw-maxLPs}, for a fixed $\calB \in \bB$ the inner
$\max$ is exactly the LP on the right-hand side of~\eqref{eqn:ddl:target} whose dual
is~\eqref{eqn:dual:ddl:target}.
In particular, to compute the $\dasubw(Q)$, we can
solve a collection of linear programs and take the maximum solution among them.
Since there is a different linear program for each valid choice of $\calB$ (and $\calB$ is a set of subsets of $[n]$), the
total number of linear programs is $\leq 2^{2^{n}}$.

In order to compute $Q$ in the desired time, we mimic this strategy in the
algorithm.
For each $\calB \in \bB$, we solve the LP~\eqref{eqn:ddl:target}. Let
$(\vec\delta^*, \vec\sigma^*,\vec\mu^*)$ denote a dual optimal solution.
From Proposition~\ref{prop:sfi:ddl:target},
$\inner{\vec\lambda, \mv h} \leq \inner{\vec\delta^*,\mv h}$ is a Shannon flow
inequality. On this input $\panda$ computes a tuple
$\mv T_\calB = (T_B)_{B\in\calB}$
of tables such that, for every $\mv a \in Q$ there exists a $B\in\calB$ for
which $\Pi_B(\mv a) \in T_B$.

Let $M = |\bB|$ and suppose $\bB = \{\calB_1,\dots,\calB_M\}$.
We prove the following claims:

\begin{claim}
For every $(B_1,\dots,B_M) \in \prod_{i=1}^M\calB_i$, there is a
tree decomposition $(T,\chi) \in \td(Q)$ such that, for every tree node
$t \in V(T)$, $\chi(t) = B_j$ for some $j \in [M]$.
Breaking ties arbitrarily, we call this tree decomposition {\em the} tree
decomposition (of $Q$) associated with the tuple $(B_1,\dots,B_M)$.
\label{clm:panda:dasubw1}
\end{claim}

\begin{claim}
For any tuple $(B_1,\dots,B_M) \in \prod_{i=1}^M\calB_i$ with
associated tree decomposition $(T, \chi)$, define
$$J(B_1,\dots,B_M) \defeq \Join_{t\in V(T)} T_{\chi(t)}.$$ Then,
\begin{equation} Q \subseteq
   \left(\bigcup_{(B_1,\dots,B_M) \in \prod_{i=1}^M\calB_i} \Join_{j=1}^M
   T_{B_j}\right)
   \subseteq
   \left(\bigcup_{(B_1,\dots,B_M) \in \prod_{i=1}^M\calB_i}
   J(B_1,\dots,B_M)\right).
\label{eqn:inclusion}
\end{equation}
\label{clm:panda:dasubw2}
\end{claim}
Assuming the claims, the query can be computed by taking each tuple $(B_1,\dots,B_M) \in \prod_{i=1}^M\calB_i$ and running Yannakakis algorithm
to compute the output of the associated join query $J(B_1,\dots,B_M)$ within a runtime of
$\tilde O(2^{\dasubw(Q)}+|J(B_1,\dots,B_M)\cap Q|)$.
If we apply Yannakakis algorithm straight up on the join $J(B_1,\dots,B_M)$, then we
can attain the runtime $\tilde O(2^{\dasubw(Q)}+|J(B_1,\dots,B_M)|)$,
because every table $T_{B_j}$ has size bounded by $2^{\dasubw(Q)}$.
To reduce the runtime down to
$\tilde O(2^{\dasubw(Q)}+|J(B_1,\dots,B_M)\cap Q|)$, we semijoin-reduce every
table $T_{\chi(t)}$ in the join $J(B_1,\dots,B_M)$ with every input relation before we join the tables $T_{\chi(t)}$ together.
The above process has to be repeated for every tuple $(B_1,\dots,B_M) \in \prod_{i=1}^M\calB_i$.
There are $\prod_{i=1}^M|\calB_i|$ such tuples, which is a query-complexity
quantity.

We next prove Claim~\ref{clm:panda:dasubw1}. Fix a tuple $(B_1,\dots,B_M) \in \prod_{i=1}^M\calB_i$.
Suppose to the contrary that for {\em every} tree decomposition $(T,\chi)$
there is a tree node $t \in
V(T)$ such that $\chi(t)\notin \{B_1,\dots,B_M\}$.
Call the bag $\chi(t)$ a {\em missed} bag of the tree decomposition $(T,\chi)$.
Consider a bag selector $\bar\beta : \td(Q) \to 2^\calV$ where $\bar\beta(T,\chi)$ is
exactly the missed bag of the tree decomposition $(T,\chi)$. Note that
by definition of $\bB$ we have $\image(\bar\beta) = \calB_k$ for some $k\in [M]$.
This is a contradiction because $B_k\in\calB_k$ must then be the missed bag of some tree
decomposition, but it is not missed anymore (since it belongs to $\{B_1,\dots,B_M\}$).

Finally, we prove Claim~\ref{clm:panda:dasubw2}. Consider an output tuple $\mv a \in Q$.
For each $j \in [M]$, let $B_j$ denote the target in $\calB_j$ for which $\Pi_{B_j}(\mv a) \in
T_{B_j}$. Then, obviously $\mv a \in \Join_{j=1}^M T_{B_j}$.
This proves the first inclusion in~\eqref{eqn:inclusion}.
The second inclusion is obvious because the join $J(B_1,\dots,B_M)$ drops
some tables from the join $\Join_{j=1}^M T_{B_j}$.
\ep

For a simple example on the algorithm from the proof of Corollary~\ref{cor:panda:dasubw}, see Example~\ref{ex:intro:subw}.
\section{Discussions}
\label{sec:conclusion}

Our negative answer to Question 1 from Section~\ref{sec:the:problems} leads to a natural question: can we design an
algorithm whose runtime matches the entropic bound under the presence of FDs or
degree constraints? Worst-case optimal join algorithms~\cite{NPRR, skew, LFTJ, tetris}
were able to achieve this when there are no FDs (nor degree constraints). And, as shown in~\cite{csma}
there are classes of queries with FDs for which the answer is positive (using
the chain algorithm). A natural direction is to extend the class of queries
with FDs where the entropic bound can be met, beyond
what was shown in~\cite{csma}.

Along the same line, the next natural open question is to design algorithms
to evaluate disjunctive datalog rules matching the {\em entropic} bound
$\outputbound_{\overline\Gamma^*_n\cap \hdc}(P)$. From there, the possibility of
achieving $\edasubw$ and/or $\edafhtw$ is within reach. We already have an
example where $\panda$ was able to achieve $\edasubw$ and $\edafhtw$:
the $4$-cycle query.
In general, the inner-most column of Figure~\ref{fig:diagram} (i.e. the column with $X$-coordinate of $\overline\Gamma^*_n$ and $Y$-coordinate of $\hdc$) contains open
algorithmic questions: we do not know of algorithms meeting bounds involving both $\overline\Gamma^*_n$ and $\hdc$.
Another big open question is to remove the polylog factor from the runtime of
$\panda$.

As was mentioned right after Corollary~\ref{cor:adw:edasubw}, for queries with
only cardinality constraints, $\dasubw(Q) = O(\edasubw(Q)^4)$ and hence
bounded entropic submodular width implies bounded submodular width and vice
versa. It is open whether or not the same relationship holds when $Q$ has FDs
and/or degree bounds. 

The algorithmic results we formally stated in the paper apply only to full and to Boolean
conjunctive queries (Theorem~\ref{thm:main:panda:submodular}). This begs a 
natural question: ``what happens to proper
conjunctive queries and to aggregate queries (such as $\faq$-queries over a single
semiring, called the {\sf SumProd} or $\faqcs$ queries~\cite{faq,AM00})?''
Our technique and results easily extend
to the case of general conjunctive queries $Q$ (i.e. queries whose set of free 
variables isn't necessarily empty nor equal to $\calV$).  
To deal with these queries, the first minor change is to 
generalize the notions of $\maximinw$ and $\minimaxw$ defined in
Definition~\ref{defn:maximin:minimax}: In particular, the $\min_{(T,\chi)}$ should now  range only over 
``free-connex`` tree decompositions $(T,\chi)$ instead of ranging over all tree
decompositions. A ``free-connex'' tree decomposition is a tree decomposition
constructed from a GYO variable elimination ordering in which non-free variables are eliminated
before the free variables (see~\cite{faq} for how one obtains a tree
decomposition from a variable ordering).
Alternatively, a free-connex tree decomposition can be defined as a tree decomposition whose bags form a free-connex acyclic query~\cite{Segoufin13, Bagan2007}.
Achieving the $\minimaxw$ can be done in the exact same way as before. Achieving the $\maximinw$ requires
a second minor change:
we need a collection of auxiliary disjunctive datalog rules. These rules
are obtained using the distributivity law in exactly the same way it was applied
in~\eqref{eqn:4:cycle:distributivity}, except that the head conjunction is now only
over free-connex tree decompositions. In other words, the ``bag selector map''
$\beta$ in the proof of Corollary~\ref{cor:panda:dasubw} selects bags only from
``free-connex'' tree decompositions.
In the case of $\faqcs$ queries, we can easily achieve the $\dafhtw$-runtime
from Corollary~\ref{cor:panda:dafhtw}, with $\dafhtw$ replaced by the
width over free-connex tree decompositions as above.
However, achieving the $\dasubw$ for general $\faqcs$ queries remains an open problem.
To keep the paper accessible, we have decided against presenting the more
general treatment outlined in this paragraph, concentrating only on Boolean and
full conjunctive queries.


\bibliographystyle{acm}
\bibliography{main}

\newpage

\appendix

\section{Missing details from the introduction}
\label{app:sec:the:problems}

First, we prove the correctness of three bounds for the query from Example~\ref{ex:intro:1}.  

\begin{itemize}
\item Bound (a) follows from:
\begin{eqnarray*}
\log|Q| = h(A_1A_2A_3A_4) &\leq& h(A_1A_2)+h(A_3A_4)\leq 2\log N,
\end{eqnarray*}
which implies $|Q| \leq N^2$.

\item Bound (b) follows from:
\begin{eqnarray*}
3 \log N + 2 \log D &\geq& 
h(A_2A_3) + h(A_3A_4) + h(A_4A_1) + h(A_2|A_1) + h(A_1|A_2)\\
&\geq& h(A_3)+h(A_2A_3A_4)+ h(A_4A_1) + h(A_2|A_1) + h(A_1|A_2)\\
&\geq& h(A_3A_4A_1)+h(A_2A_3A_4)+ h(A_2|A_1) + h(A_1|A_2)\\
&\geq& h(A_3A_4A_1)+h(A_2A_3A_4)+ h(A_2|A_3A_4A_1) + h(A_1|A_2A_3A_4)\\
&=& 2h(A_1A_2A_3A_4)\\
&=& 2\log |Q|,
\end{eqnarray*}
which implies $|Q|\leq N^{3/2}\cdot D$.

\item Bound (c) follows from (b) by setting $D=1$.
\end{itemize}
Now, we prove the tightness of the three bounds:
\begin{itemize}
\item Bound (a) is tight on the following database instance:
$R_{12}=R_{34}=[N]\times[1]$, $R_{23}=R_{41}=[1]\times[N]$.
The output is $Q(A_1,A_2,A_3,A_4)=[N]\times[1]\times[N]\times[1]$.
\item Let $K\defeq\left\lfloor\sqrt{N}\right\rfloor$. Bound (c) is \emph{asymptotically} tight
 on the instance
$R_{12} = \setof{(i,i)}{i \in [K]}$,
$R_{23}=R_{34}=R_{41}=[K] \times [K]$. The output is
$Q(A_1, A_2, A_3, A_4)=\left\{(i, i, j, k)\suchthat i, j, k\in[K]\right\}.$
\item Bound (b) is tight on the following instance 
(which is a generalization of the previous one to $1\leq D\leq K$.)
\[R_{23}=R_{34}=R_{41}=[K] \times [K],\]
\[R_{12}=\left\{(i, j)\suchthat i, j\in[K], (j-i) \mod K < D\right\}.\]
\end{itemize}

\section{More on Shannon-flow inequalities and proof sequences}
\label{app:sec:ps}
This section presents extra results on Shannon flow inequalities and their proof sequences that go beyond the results of Section~\ref{sec:ps}.

\subsection{Bounding $\norm{\vec\delta}_1, \norm{\vec\mu}_1$, and $\norm{\vec\sigma}_1$ (w.r.t $\norm{\vec\lambda}_1$)}
\label{app:subsec:bounding-delta-sigma-mu}

Given a Shannon flow inequality $\inner{\vec\lambda, \mv h}\leq\inner{\vec\delta, \mv h}$
with a witness $(\vec\sigma,\vec\mu)$, there are various ways to construct a proof sequence for
$\inner{\vec\lambda, \mv h}\leq\inner{\vec\delta, \mv h}$.
(Theorem~\ref{thm:ps:construction:1} showed one possible construction,
and we will see more efficient constructions in Section~\ref{sec:poly-sized-proofseq}.)
What those various constructions have in common is that the length
of the resulting proof sequence depends on $\norm{\vec\sigma}_1$,
$\norm{\vec\delta}_1$ and/or $\norm{\vec\mu}_1$.
In turn, the runtime of the $\panda$ algorithm depends on the length of the proof sequence,
which provides a motivation for minimizing $\norm{\vec\sigma}_1$,
$\norm{\vec\delta}_1$ and $\norm{\vec\mu}_1$ as much as possible, which is our target in this section.
In particular, we want to replace the original inequality (and its witness)
with a new inequality that is ``just as good'' but has lower
$\norm{\vec\sigma}_1$,  $\norm{\vec\delta}_1$ and $\norm{\vec\mu}_1$.

This section is outlined as follows. We start Section~\ref{subsec:bounding-conditional-mu}
with some definitions that formalize what we meant by an inequality being ``just as good'' as another,
then we proceed to bounding the total of monotonicity terms of the form
$\mu_{X, Y}$ for
$X\neq\emptyset$.
In Section~\ref{subsec:bounding-whole-delta}, we bound the total of terms of the form
$\delta_{Y|\emptyset}$.
Section~\ref{subsec:bounding-mu} shows that this is equivalent to bounding monotonicity terms
$\mu_{\emptyset, Y}$,
hence we now have a bound on all monotonicity terms $\norm{\vec\mu}_1$.
Finally, Section~\ref{subsec:bounding-delta-and-sigma} shows that bounding
$\norm{\vec\mu}_1$ imposes a bound on both $\norm{\vec\delta}_1$ and $\norm{\vec\sigma}_1$.

\subsubsection{Bounding the total of $\mu_{X,Y}$ for $X\neq\emptyset$}
\label{subsec:bounding-conditional-mu}

\bdefn[$\vec\delta'$ dominated by $\vec\delta$]
\label{defn:delta-domination}
Given two vectors $\vec\delta, \vec\delta'\in\Q_+^\calP$,
we say that $\vec\delta'$ is \emph{dominated} by $\vec\delta$ if there exists a function $f:\calP\times\calP\rightarrow \Q_+$ that satisfies the following conditions.
\bi
\item If $f((X,Y),(X',Y'))\neq 0$, then $Y'\setminus X'\subseteq Y\setminus X$ and $X'\supseteq X$.
\item $\forall (X,Y)\in\calP, \quad \sum_{(X',Y')\in\calP}f((X,Y),(X',Y'))\leq \delta_{Y|X}$.
\item $\forall (X',Y')\in\calP, \quad \sum_{(X,Y)\in\calP}f((X,Y),(X',Y'))\geq \delta'_{Y'|X'}$.
\ei
The function $f$ is called the \emph{domination function for $(\vec\delta,\vec\delta')$}.
\edefn

\bdefn[Rational compatibility]
Given two vectors of rational numbers $\mv q\in\Q_+^l, \mv q'\in\Q_+^{l'}$(for some integers $l, l'$),
we say that $\mv q'$ is \emph{rationally-compatible} with $\mv q$
if the minimum common denominator of all entries in $\mv q$
is a common denominator (not-necessarily minimum) for all entries in $\mv q'$.
\edefn

\blmm
For any Shannon flow inequality $\inner{\vec\lambda, \mv h} \leq \inner{\vec\delta,\mv h}$
with a witness $(\vec\sigma,\vec\mu)$,
there exists a Shannon flow inequality $\inner{\vec\lambda', \mv h} \leq \inner{\vec\delta',\mv h}$
with a tight\footnote{See Definition~\ref{defn:tight-witness}.} witness $(\vec\sigma',\vec\mu')$ such that $\vec\delta'$ is dominated by $\vec\delta$ and $\vec\lambda'$ dominates $\vec\lambda$ and $\mu'_{X,Y}=0$ for all $\emptyset\neq X\subset Y\subseteq \calV$.
Moreover, $(\vec\lambda',\vec\delta',\vec\sigma',\vec\mu')$ is rationally-compatible with
$(\vec\lambda,\vec\delta,\vec\sigma,\vec\mu)$.
\label{lmm:bound-conditional-mu}
\elmm
\bp
W.L.O.G. we can assume the witness $(\vec\sigma,\vec\mu)$ to be tight.
We start by choosing $\vec\delta'=\vec\delta$, $\vec\lambda'=\vec\lambda$,
$(\vec\sigma',\vec\mu')=(\vec\sigma,\vec\mu)$.
Moreover, we initially choose the domination function $f$ for $(\vec\delta,\vec\delta')$ to be
\[
f((X,Y),(X',Y'))=
\begin{cases}
\delta_{Y|X} & \text{if $X'=X$ and $Y'=Y$}\\
0 & \text{otherwise.}
\end{cases}
\]
Similarly, the domination function $g$ for $(\vec\lambda',\vec\lambda)$ is chosen to be
\[
g((X',Y'),(X,Y))=
\begin{cases}
\lambda_{Y|X} & \text{if $X'=X$ and $Y'=Y$}\\
0 & \text{otherwise.}
\end{cases}
\]
Let $D$ be the common dominator of $(\vec\lambda',\vec\delta',\vec\sigma',\vec\mu')$, and $w$ be $1/D$.
While there are $\emptyset\neq X\subset Y\subseteq \calV$ where $\mu'_{X,Y}>0$ we apply the following.
If $\lambda'_{X}>0$, then we reduce both $\lambda'_{X}$ and $\mu'_{X,Y}$ by $w$ and increase $\lambda'_{Y}$ by $w$.
Moreover, we choose an arbitrary $Z$ (if any) where $g((\emptyset,X),(\emptyset,Z))>0$,
and we reduce $g((\emptyset,X),(\emptyset,Z))$ by $w$ and increase $g((\emptyset,Y),(\emptyset,Z))$ by $w$.
Otherwise (if $\lambda'_X=0$), since the witness $(\vec\sigma',\vec\mu')$ is tight and $\mu'_{X,Y}$ is increasing $\flow(X)$, there must be some dual variable that is reducing $\flow(X)$ back to $0$. In particular, we recognize the following three cases, which are depicted in Figure~\ref{fig:bound-conditional-mu}:
\be
\item If there is $W\subset X$ such that $\mu'_{W,X}>0$,
then we reduce both $\mu'_{W,X}$ and $\mu'_{X,Y}$ by $w$ and increase $\mu'_{W,Y}$ by $w$.
\item If there is some $Y'\supset X$ such that $\delta'_{Y'|X}>0$,
then we reduce both $\mu'_{X, Y}$ and $\delta'_{Y'|X}$ by $w$
and increase both $\delta'_{Y\cup Y'|Y}$ and $\mu'_{Y',Y\cup Y'}$ by $w$.
Moreover, we choose an arbitrary $W\subset Z\subseteq \calV$ (if any) where $f((W,Z),(X,Y'))>0$,
and we reduce $f((W,Z),(X,Y'))$ by $w$ and increase $f((W,Z),(Y,Y\cup Y'))$ by $w$.
\item If there is some $X'\incomp X$ such that $\sigma'_{X,X'}>0$,
then we reduce both $\mu'_{X,Y}$ and $\sigma'_{X,X'}$ by $w$ and increase each one of $\sigma'_{Y,X'}$, $\mu'_{X\cup X',Y\cup X'}$, $\mu'_{X\cap X',Y\cap X'}$ by $w$.
\ee
In all cases, we are maintaining $(\vec\sigma',\vec\mu')$ as a valid and tight witness.
We are also maintaining that $\vec\delta'$ is dominated by $\vec\delta$ and $\vec\lambda'$ dominates $\vec\lambda$.
To prove that the above process terminates, consider the following non-negative, bounded, integral function.
\[\phi(\vec\lambda',\vec\delta',\vec\sigma')\defeq
D\cdot\left[\sum_{B}\lambda'_B(n-|B|)
+\sum_{X\subset Y} \delta'_{Y|X}(|Y|-|X|)(n-|X|)
+\sum_{I\incomp J}\sigma'_{I,J}(2n-|I|-|J|)\right].\]
In all cases except (1), $\phi(\vec\lambda',\vec\delta',\vec\sigma')$ decreases by $1$.
Case (1) reduces the quantity $D\cdot \norm{\vec\mu'}_1$ by $1$.
Initially, we have $\norm{\vec\mu'}_1= \norm{\vec\mu}_1$. However, $ \norm{\vec\mu'}_1$ increases in case (3).
Hence, initially case (1) can maximally be repeated $D\cdot \norm{\vec\mu}_1$ consecutive times.
Later on, every time we increase some $\mu'_{X,Y}$ in cases (2) and (3), we can immediately check whether there is some $W\subset X$ where $\mu'_{W,X}>0$ and if so apply case (1). By doing so, we can amortize the cost of case (1) over cases (2) and (3).
\ep

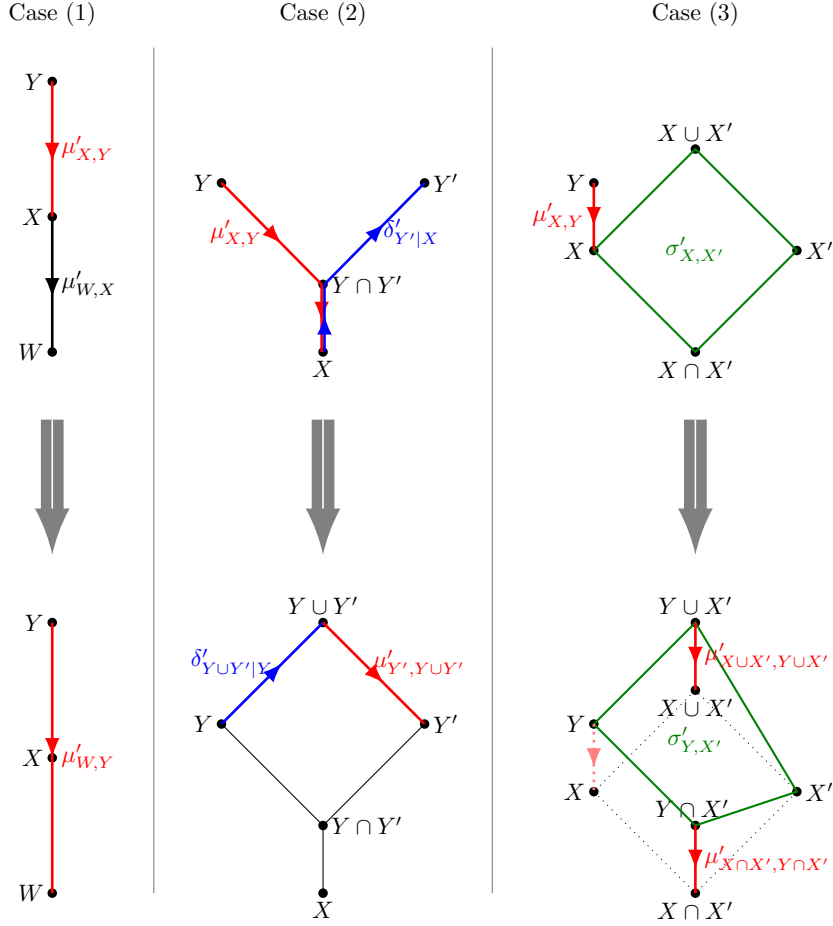
\begin{figure}[ht!]
\newcommand{\drawdot}[5]
{
   \coordinate (#3) at(#1, #2);
   \fill (#1,#2) circle[radius=2pt] node[#4] {#5};
}
\tikzset
{
   ->-/.style=
   {
      decoration=
      {
         markings,
         mark=at position #1 with {\arrow{Latex}[scale=2]}
      },
      postaction={decorate},
      line width=1pt
   },
   ->-/.default=0.6
}
\centering
\begin{tikzpicture}[scale=1, every node/.style={transform shape}]
   \begin{scope}[shift={(0,0)}]
      \drawdot{0}{0}{Y}{left}{$Y$};
      \drawdot{0}{-2}{X}{left}{$X$};
      \drawdot{0}{-4}{W}{left}{$W$};
      \draw[->-,red] (Y)--(X) node[midway, right]{$\mu'_{X,Y}$};
      \draw[->-] (X)--(W) node[midway, right]{$\mu'_{W,X}$};
   \end{scope}
   \begin{scope}[shift={(0,-8)}]
      \drawdot{0}{0}{Y}{left}{$Y$};
      \drawdot{0}{-2}{X}{left}{$X$};
      \drawdot{0}{-4}{W}{left}{$W$};
      \draw[->-=.5,red] (Y)--(W) node[midway, right]{$\mu'_{W,Y}$};
   \end{scope}
   
   \begin{scope}[shift={(2,0)}]
      \drawdot{.5}{-1.5}{Y}{left}{$Y$};
      \drawdot{3.5}{-1.5}{Y'}{right}{$Y'$};
      \drawdot{2}{-3}{YnY'}{right}{$Y\cap Y'$};
      \drawdot{2}{-4}{X}{below}{$X$};
      \draw[->-,red] (Y)--([xshift=-.6pt]YnY') node[midway, left]{$\mu'_{X,Y}$};
      \draw[->-=.55,red] ([xshift=-.6pt]YnY')--([xshift=-.6pt]X);
      \draw[->-=.55,blue] ([xshift=+.6pt]X)--([xshift=+.6pt]YnY');
      \draw[->-,blue] ([xshift=+.6pt]YnY')--(Y') node[midway, right]{$\delta'_{Y'|X}$};
   \end{scope}
   \begin{scope}[shift={(2,-8)}]
      \drawdot{.5}{-1.5}{Y}{left}{$Y$};
      \drawdot{3.5}{-1.5}{Y'}{right}{$Y'$};
      \drawdot{2}{-3}{YnY'}{right}{$Y\cap Y'$};
      \drawdot{2}{0}{YuY'}{above}{$Y\cup Y'$};
      \drawdot{2}{-4}{X}{below}{$X$};
      \draw[->-,blue] (Y)--(YuY') node[pos=.6, left]{$\delta'_{Y\cup Y'|Y}$};
      \draw[->-,red] (YuY')--(Y') node[pos=.4, right]{$\mu'_{Y',Y\cup Y'}$};
      \draw[] (Y)--(YnY');
      \draw[] (YnY')--(X);
      \draw[] (YnY')--(Y');
   \end{scope}
   
   \begin{scope}[shift={(8,0)}]
      \drawdot{0}{-1.5}{Y}{left}{$Y$};
      \drawdot{0}{-2.5}{X}{left}{$X$};
      \drawdot{3}{-2.5}{X'}{right}{$X'$};
      \drawdot{1.5}{-4}{XnX'}{below}{$X\cap X'$};
      \drawdot{1.5}{-1}{XuX'}{above}{$X\cup X'$};
      \draw[DarkGreen, thick] (X)--(XnX')--(X')--(XuX')--cycle;
      \node[DarkGreen] at($(X)!.5!(X')$) {$\sigma'_{X,X'}$};
      \draw[->-=.65,red] (Y)--(X) node[midway,left] {$\mu'_{X,Y}$};
   \end{scope}
   \begin{scope}[shift={(8,-8)}]
      \drawdot{0}{-1.5}{Y}{left}{$Y$};
      \drawdot{0}{-2.5}{X}{left}{$X$};
      \drawdot{3}{-2.5}{X'}{right}{$X'$};
      \drawdot{1.5}{-4}{XnX'}{below}{$X\cap X'$};
      \drawdot{1.5}{-3}{YnX'}{above,xshift=-.05cm}{$Y\cap X'$};
      \drawdot{1.5}{-1}{XuX'}{below}{$X\cup X'$};
      \drawdot{1.5}{0}{YuX'}{above}{$Y\cup X'$};
      \draw[dotted] (X)--(XnX')--(X')--(XuX')--cycle;
      \draw[DarkGreen, thick] (Y)--(YnX')--(X')--(YuX')--cycle;
      \node[DarkGreen] at($(Y)!.5!(X')+(0,.25cm)$) {$\sigma'_{Y,X'}$};
      \draw[->-=.65,red!50,dotted] (Y)--(X);
      \draw[->-=.65,red] (YuX')--(XuX') node[midway, right] {$\mu'_{X\cup X',Y\cup X'}$};
      \draw[->-=.65,red] (YnX')--(XnX') node[midway, right] {$\mu'_{X\cap X',Y\cap X'}$};
   \end{scope}
   
   \foreach \x in {0, 4, 9.5}
      \draw[-latex,double, line width=4, gray] (\x, -5)--(\x, -7);
   \node at(0  , 1) {Case (1)};
   \node at(4  , 1) {Case (2)};
   \node at(9.5, 1) {Case (3)};
   \foreach \x in {1.5, 6.5}
      \draw[gray] (\x,.5)--(\x,-12);
\end{tikzpicture}
\caption{Illustration of the proof of Lemma~\ref{lmm:bound-conditional-mu}.}
\Description{Illustration of the proof of Lemma~\ref{lmm:bound-conditional-mu}.}
\label{fig:bound-conditional-mu}
\end{figure}

\bcor[The total of $\mu_{X,Y}$ for $X\neq\emptyset$ is $\leq\norm{\vec\lambda}_1$]
For any Shannon flow inequality $\inner{\vec\lambda, \mv h} \leq \inner{\vec\delta,\mv h}$
with a witness $(\vec\sigma,\vec\mu)$,
there exists a Shannon flow inequality $\inner{\vec\lambda, \mv h} \leq \inner{\vec\delta',\mv h}$
with a tight witness $(\vec\sigma',\vec\mu')$ such that $\vec\delta'$ is dominated by $\vec\delta$ and
\[\forall \emptyset\neq X\subseteq \calV, \quad \sum_{Y\supset X} \mu'_{X,Y} \leq \lambda_{X}.\]
Hence
\begin{equation}
\sum_{\emptyset\neq X\subset Y\subseteq \calV} \mu'_{X,Y}\leq \norm{\vec\lambda}_1.
\end{equation}
Moreover, $(\vec\lambda,\vec\delta',\vec\sigma',\vec\mu')$ is rationally-compatible with
$(\vec\lambda,\vec\delta,\vec\sigma,\vec\mu)$.
\label{cor:bound-conditional-mu}
\ecor

\subsubsection{Bounding the total of $\delta_{Y|\emptyset}$}
\label{subsec:bounding-whole-delta}

\blmm[The total of $\delta_{Y|\emptyset}$ is $\leq n\cdot \norm{\vec\lambda}_1$]
For any Shannon flow inequality $\inner{\vec\lambda, \mv h} \leq \inner{\vec\delta,\mv h}$
with a witness $(\vec\sigma,\vec\mu)$,
there exists a Shannon flow inequality $\inner{\vec\lambda, \mv h} \leq \inner{\vec\delta',\mv h}$
with a witness $(\vec\sigma',\vec\mu')$
such that $\vec\delta'$ is dominated by $\vec\delta$ and
\[\forall v\in \calV, \quad \sum_{Y\ni v}\delta'_{Y|\emptyset} \leq \norm{\vec\lambda}_1.\]
Hence
\begin{equation}
\sum_{Y\subseteq \calV}\delta'_{Y|\emptyset} \leq n\cdot \norm{\vec\lambda}_1.
\end{equation}
Moreover, $(\vec\lambda,\vec\delta',\vec\sigma',\vec\mu')$ is rationally-compatible with
$(\vec\lambda,\vec\delta,\vec\sigma,\vec\mu)$.
\label{lmm:bound-whole-delta}
\elmm

\bp
Fix an arbitrary variable $v\in \calV$ such that
$\sum_{Y\ni v}\delta_{Y|\emptyset} > \norm{\vec\lambda}_1$.
Define a function $\calF:\Q_+^\calP\rightarrow\Q_+^\calP$
such that for any vector $\mv t\in\Q_+^\calP$,
$\calF(\mv t)$ is a vector $\mv t'\in\Q_+^\calP$  defined as
\[
t'_{Y|X}\defeq
\begin{cases}
t_{Y|X} + t_{Y|X-\{v\}} + t_{Y-\{v\}|X-\{v\}}& \text{if $v\in X$ (hence $v \in Y$)}\\
0 & \text{otherwise (i.e. if $v\notin Y$ or $v \notin X$)}
\end{cases}
\]
Note that $\mv t'$ is dominated by $\mv t$ where the domination function $f$ is defined as
\begin{eqnarray*}
f((X,Y),(X\cup\{v\},Y\cup\{v\})) &=& t_{Y|X},\\
f((X,Y),(X',Y')) &=& 0 \quad\quad\text{otherwise.}
\end{eqnarray*}

\begin{claim}
The inequality $\inner{\calF(\vec\lambda), \mv h} \leq \inner{\calF(\vec\delta),\mv h}$
holds for all conditional polymatroids $\mv h$.
\label{clm:bound-whole-delta:1}
\end{claim}

(Notice that the inequality $\inner{\calF(\vec\lambda), \mv h} \leq \inner{\calF(\vec\delta),\mv h}$
does not necessarily have the form \eqref{eqn:sfi} of a Shannon flow inequality,
because it can have $\lambda'_{B\cup\{v\}|\{v\}}>0$.)
We will prove Claim~\ref{clm:bound-whole-delta:1} based on the following claim.

\begin{claim}
For any $\mv f\in \{
   \mv s_{I,J},
   \mv m_{X,Y},
   \mv c_{X,Y},
\mv d_{Y,X} \}$ and any conditional polymatroid $\mv h$,
the following inequality holds
\[\inner{\calF(\mv f),\mv h} \leq 0.\]
\label{clm:bound-whole-delta:2}
\end{claim}
Claim~\ref{clm:bound-whole-delta:2} holds because the following inequalities hold for all conditional polymatroids.
\begin{eqnarray*}
h(I\cup J\cup\{v\}|J\cup\{v\})-h(I\cup\{v\}|(I\cap J)\cup\{v\})&\leq&0, I \incomp J\\
-h(Y\cup\{v\}|\{v\})+h(X\cup\{v\}|\{v\})&\leq&0, X\subset Y\\
h(Y\cup\{v\}|\{v\})-h(Y\cup\{v\}|X\cup\{v\})-h(X\cup\{v\}|\{v\})&\leq&0, X\subset Y\\
-h(Y\cup\{v\}|\{v\})+h(Y\cup\{v\}|X\cup\{v\})+h(X\cup\{v\}|\{v\})&\leq&0, X\subset Y
\end{eqnarray*}

To prove Claim~\ref{clm:bound-whole-delta:1}, we construct a proof sequence for the original inequality $\inner{\vec\lambda, \mv h} \leq \inner{\vec\delta,\mv h}$,
which is possible thanks to Theorem~\ref{thm:ps:construction:1}.
Consider the corresponding inequality sequence
$\vec\delta=\vec\delta_0, \vec\delta_1, \dots, \vec\delta_\ell\geq\vec\lambda$ such that, for
every $i \in [\ell]$, $\vec\delta_{i} = \vec\delta_{i-1} + w_i \cdot \mv f_i$
for some $w_i > 0$ and $\mv f_i \in \{
   \mv s_{I,J},
   \mv m_{X,Y},
   \mv c_{X,Y},
\mv d_{Y,X} \}$.
\begin{eqnarray*}
\inner{\calF(\mv f_i),\mv h} &\leq& 0
\quad\quad\quad \text{(Claim~\ref{clm:bound-whole-delta:2})}\\
\inner{\calF(1/w_i(\vec\delta_i-\vec\delta_{i-1})),\mv h} &\leq& 0\\
\inner{1/w_i(\calF(\vec\delta_i)-\calF(\vec\delta_{i-1})),\mv h} &\leq& 0
\quad\quad\quad \text{(by linearity of $\calF$)}\\
\inner{\calF(\vec\delta_i),\mv h} &\leq& \inner{\calF(\vec\delta_{i-1}),\mv h}
\end{eqnarray*}
which proves Claim~\ref{clm:bound-whole-delta:1}.

Finally, let $\vec\delta'=\calF(\vec\delta)$, $\vec\lambda'=\calF(\vec\lambda)$,
$f$  be the domination function of $(\vec\delta,\vec\delta')$,
$D$ be the common denominator of $(\vec\delta,\vec\delta',\vec\lambda,\vec\lambda')$, and $w=1/D$.
Initially, we have $\sum_{Y\ni v}\delta'_{Y|\emptyset}=0$.
While there is $B\subseteq \calV$ such that $\lambda'_{B\cup\{v\}|\{v\}}>0$,
find $Y\ni v$ such that $f((\emptyset,Y),(\{v\},Y))>0$,
and reduce $\lambda'_{B\cup\{v\}|\{v\}}$, $f((\emptyset,Y),(\{v\},Y))$, and $\delta'_{Y|\{v\}}$ by $w$,
and increase $\lambda'_{B\cup\{v\}|\emptyset}$, $f((\emptyset,Y),(\emptyset,Y))$, and $\delta'_{Y|\emptyset}$ by $w$.
Note that this update maintains that $\inner{\vec\lambda',\mv h}\leq\inner{\vec\delta',\mv h}$ holds for all polymatroids $\mv h$
and that $\vec\delta'$ is dominated by $\vec\delta$.
At the end, $\vec\lambda'$ will dominate $\vec\lambda$ and we will have $\sum_{Y\ni v}\delta'_{Y|\emptyset}\leq\norm{\vec\lambda}_1$.
\ep

\subsubsection{Bounding $\norm{\vec\mu}_1$}
\label{subsec:bounding-mu}

\bcor[$\norm{\vec\mu}_1\leq n\cdot \norm{\vec\lambda}_1$]
For any Shannon flow inequality $\inner{\vec\lambda, \mv h} \leq \inner{\vec\delta,\mv h}$
with a witness $(\vec\sigma,\vec\mu)$,
there exists a Shannon flow inequality $\inner{\vec\lambda, \mv h} \leq \inner{\vec\delta'',\mv h}$
with a tight witness $(\vec\sigma'',\vec\mu'')$ such that $\vec\delta''$ is dominated by $\vec\delta$ and
\begin{equation}
\norm{\vec\mu''}_1\leq n\cdot \norm{\vec\lambda}_1.
\end{equation}
In particular
\begin{equation}
\forall \emptyset\neq X\subseteq \calV, \quad \sum_{Y\supset X} \mu''_{X,Y} \leq \lambda_{X},
\quad \text{hence}\quad \sum_{\emptyset\neq X\subset Y\subseteq \calV} \mu''_{X,Y}\leq \norm{\vec\lambda}_1
\label{eq:bound-mu-conditional}
\end{equation}
and
\begin{equation}
\sum_{Y\subseteq \calV} \mu''_{\emptyset,Y} \leq (n-1)\cdot\norm{\vec\lambda}_1.
\label{eq:bound-mu-whole}
\end{equation}
Moreover, $(\vec\lambda,\vec\delta'',\vec\sigma'',\vec\mu'')$ is rationally-compatible with
$(\vec\lambda,\vec\delta,\vec\sigma,\vec\mu)$.
\label{cor:bound-mu}
\ecor

\bp
First, we will apply Lemma~\ref{lmm:bound-whole-delta} to get a Shannon flow inequality
$\inner{\vec\lambda, \mv h} \leq \inner{\vec\delta',\mv h}$ where $\vec\delta'$ is dominated by $\vec\delta$
and $\sum_{Y\subseteq \calV}\delta'_{Y|\emptyset} \leq n\cdot \norm{\vec\lambda}_1$.
Then, we apply Corollary~\ref{cor:bound-conditional-mu} on $\inner{\vec\lambda, \mv h} \leq \inner{\vec\delta',\mv h}$
to get another Shannon flow inequality $\inner{\vec\lambda, \mv h} \leq \inner{\vec\delta'',\mv h}$
with a tight witness $(\vec\sigma'',\vec\mu'')$ such that $\vec\delta''$ is dominated by $\vec\delta'$
and \eqref{eq:bound-mu-conditional} holds.
Because $\vec\delta''$ is dominated by $\vec\delta'$, we have
\[
\sum_{Y\subseteq \calV}\delta''_{Y|\emptyset} \leq\sum_{Y\subseteq \calV}\delta'_{Y|\emptyset}\leq n\cdot \norm{\vec\lambda}_1.
\]
Let $\flow''(Z)$ denote the quantity $\flow(Z)$ measured on the vector
$(\vec\delta'',\vec\sigma'',\vec\mu'')$.
\begin{eqnarray*}
\sum_{\emptyset\neq Z\subseteq \calV} \flow''(Z) &=& \norm{\vec\lambda}_1\\
\sum_{Z\neq\emptyset} \delta''_{Z|\emptyset}
-\sum_{Z\neq\emptyset} \mu''_{\emptyset,Z}
-\sum_{\substack{I\incomp J\\ I\cap J=\emptyset}} \sigma''_{I,J} &=& \norm{\vec\lambda}_1
\end{eqnarray*}
\begin{eqnarray*}
\sum_{Z\neq\emptyset} \mu''_{\emptyset,Z}
\leq \sum_{Z\neq\emptyset} \delta''_{Z|\emptyset} - \norm{\vec\lambda}_1
\leq n\cdot \norm{\vec\lambda}_1-\norm{\vec\lambda}_1.
\end{eqnarray*}
\begin{eqnarray*}
\norm{\vec\mu''}_1 &=&
\sum_{Z\neq\emptyset} \mu''_{\emptyset,Z}
+\sum_{\emptyset\neq X\subset Y\subseteq \calV} \mu''_{X,Y}
\leq n\cdot \norm{\vec\lambda}_1.
\end{eqnarray*}
\ep

\subsubsection{Bounding $\norm{\vec\sigma}_1$ and $\norm{\vec\delta}_1$}
\label{subsec:bounding-delta-and-sigma}

\bcor[$2\norm{\vec\sigma}_1+\norm{\vec\delta}_1 \leq n^3 \cdot \norm{\vec\lambda}_1$]
For any Shannon flow inequality $\inner{\vec\lambda, \mv h} \leq \inner{\vec\delta,\mv h}$
with a witness $(\vec\sigma,\vec\mu)$,
there exists a Shannon flow inequality $\inner{\vec\lambda, \mv h} \leq \inner{\vec\delta'',\mv h}$
with a tight witness $(\vec\sigma'',\vec\mu'')$ such that $\vec\delta''$ is dominated by $\vec\delta$ and
\begin{eqnarray}
2\norm{\vec\sigma''}_1+\norm{\vec\delta''}_1 &\leq& n^3 \cdot \norm{\vec\lambda}_1,\label{eq:bound-sigma-and-delta}\\
\norm{\vec\mu''}_1 &\leq& n\cdot \norm{\vec\lambda}_1.
\end{eqnarray}
Moreover, $(\vec\lambda,\vec\delta'',\vec\sigma'',\vec\mu'')$ is rationally-compatible with
$(\vec\lambda,\vec\delta,\vec\sigma,\vec\mu)$.
\label{cor:2sigma<=n-cube-lambda}
\ecor

\bp
We apply Corollary~\ref{cor:bound-mu} and obtain $(\delta'',\sigma'',\mu'')$.
Let $\flow''(Z)$ denote the quantity $\flow(Z)$ measured on the vector
$(\vec\delta'',\vec\sigma'',\vec\mu'')$.
\begin{equation}
\sum_{Z\subseteq \calV} \flow''(Z)\cdot|Z|^2 \leq \sum_{B\neq\emptyset} \lambda_B \cdot |B|^2
\label{eq:sum-flow(Z)|Z|^2}
\end{equation}
Define the following quantities.
\begin{eqnarray*}
T_{\sigma''} &\defeq& \sum_{I\incomp J} \sigma''_{I,J}(|I\cup J|^2+|I\cap J|^2-|I|^2-|J|^2)\\
T_{\delta''} &\defeq& \sum_{X\subset Y} \delta''_{Y|X}(|Y|^2-|X|^2)\\
T_{\mu''} &\defeq& \sum_{X\subset Y} \mu''_{X,Y} (|X|^2-|Y|^2)
\end{eqnarray*}
Now \eqref{eq:sum-flow(Z)|Z|^2} becomes
\[T_{\sigma''}+T_{\delta''}+T_{\mu''} \leq \sum_{B\neq\emptyset} \lambda_B \cdot |B|^2\]
\[T_{\sigma''}+T_{\delta''}
-\sum_{Y\neq\emptyset}\mu''_{\emptyset,Y}\cdot |Y|^2
+\sum_{\emptyset\neq X\subset Y}\mu''_{X,Y} (|X|^2-|Y|^2)
\leq \sum_{B\neq\emptyset} \lambda_B \cdot |B|^2\]
\begin{eqnarray*}
T_{\sigma''}+T_{\delta''} &\leq&
\sum_{Y\neq\emptyset}\mu''_{\emptyset,Y}\cdot |Y|^2
+\sum_{B\neq\emptyset}|B|^2 \biggl(
\underbrace{\lambda_B-\sum_{Y\supset B}\mu''_{B,Y}}_{\text{$\geq 0$ by \eqref{eq:bound-mu-conditional}}}
+\sum_{\emptyset\neq X\subset B} \mu''_{X,B}\biggr)\\
&\leq&
n^2 \cdot \underbrace{\left(\sum_{Y\neq\emptyset}\mu''_{\emptyset,Y}\right)}_{\text{bounded by \eqref{eq:bound-mu-whole}}}
+n^2\cdot
\underbrace{\sum_{B\neq\emptyset} \biggl(
\lambda_B-\sum_{Y\supset B}\mu''_{B,Y}
+\sum_{\emptyset\neq X\subset B} \mu''_{X,B}\biggr)}_{=\norm{\vec\lambda}_1}\\
&\leq& n^3\cdot \norm{\vec\lambda}_1
\end{eqnarray*}

\begin{claim}
For any $I\incomp J$, we have $|I\cup J|^2+|I\cap J|^2-|I|^2-|J|^2\geq 2$.
\label{clm:2sigma<=n-cube-lambda:1}
\end{claim}

From Claim~\ref{clm:2sigma<=n-cube-lambda:1}, $T_{\sigma''}\geq 2\norm{\vec\sigma}_1$.
Moreover, for any $X\subset Y$, $|Y|^2-|X|^2\geq 1$.
Therefore, $T_{\delta''}\geq \norm{\vec\delta''}_1$, which proves \eqref{eq:bound-sigma-and-delta}.

To prove Claim~\ref{clm:2sigma<=n-cube-lambda:1}, let $a\defeq |I\setminus J|$, $b\defeq|J\setminus I|$, and $c=|I\cap J|$.
Because $I\incomp J$, we have $a\geq1$ and $b\geq 1$.
\[|I\cup J|^2+|I\cap J|^2-|I|^2-|J|^2=(a+b+c)^2+c^2-(a+c)^2-(b+c)^2=2ab\geq 2.\]
\ep

\subsection{Construction of a poly-sized proof sequence}
\label{sec:poly-sized-proofseq}
Given a fixed Shannon flow inequality $\inner{\vec\lambda,\mv h}\leq\inner{\vec\delta,\mv h}$,
Section~\ref{sec:panda} shows that the runtime of $\panda$ algorithm depends on
the length of the proof sequence that is being used for that inequality,
thus motivating the need to minimize that length as much as possible.
Section~\ref{subsec:sfi:motivations} shows that the inequality
$\inner{\vec\lambda,\mv h}\leq\inner{\vec\delta,\mv h}$
basically corresponds to a feasible solution to the linear program \eqref{eqn:dual:ddl:target},
whose size is $O(2^{2n})$. Our aim in this section is to construct a proof sequence for
$\inner{\vec\lambda,\mv h}\leq\inner{\vec\delta,\mv h}$ whose length is polynomial in the size of that linear program.
Hence, we are looking for a proof sequence whose length is $\poly(2^n)$.

We start by introducing a connection to flow networks in Section~\ref{subset:network-flows},
and then we present the actual construction of a poly-sized proof sequence in Section~\ref{subsec:poly-sized-proofseq}. The construction relies on many technical tools
among which is Edmond-Karp maximum flow algorithm~\cite{DBLP:journals/jacm/EdmondsK72}.
It also relies on the bounds developed earlier in Section~\ref{app:subsec:bounding-delta-sigma-mu}.

\subsubsection{Introduction: Connection to flow networks}
\label{subset:network-flows}
We will make use of the following definition. Given a Shannon flow inequality $\inner{\vec\lambda, \mv h} \leq \inner{\vec\delta,\mv h}$, define
\begin{equation}
   \filter(\vec\lambda) \defeq \left\{ F \suchthat \exists B \in 2^\calV \text{ where }
   \lambda_B > 0 \text{ and } B \subseteq F \right\}.
   \label{eqn:filter}
\end{equation}
\bthm[Construction of a proof sequence using a flow network]
Let $\inner{\vec\lambda, \mv h} \leq \inner{\vec\delta,\mv h}$ be a Shannon
flow inequality with witness $(\vec\sigma,\vec\mu)$.
Algorithm~\ref{algo:calB:target:ps} produces a proof sequence for the inequality
$\inner{\vec\lambda, \mv h} \leq \inner{\vec\delta,\mv h}$ with length at most
$2^n D(\norm{\vec\lambda}_1+\norm{\vec\sigma}_1)$, where
$D$ is the minimum common denominator
of all entries in $(\vec\lambda,\vec\delta,\vec\sigma)$.
\label{thm:ps:construction:2}
\ethm
\bp
Given $(\vec\lambda,\vec\delta,\vec\sigma,\vec\mu)$,
define a flow network $G(\vec\lambda,\vec\delta,\vec\sigma,\vec\mu)=(2^\calV, \calA)$
as follows. There is an arc (i.e. directed edge)
$(X,Y) \in\calA$ for every pair $(X,Y)$ for which
$X\subset Y$ and $\delta_{Y|X} >0$;
the capacity of $(X,Y)$ is $\delta_{Y|X}$.
These are called {\em up arcs}.
There is an arc $(Y,X)$ for every pair $(X,Y)$ such that $X \subset Y$;
The capacity of this arc is $+\infty$.
These are called {\em down arcs}.
Let $\calK$ denote the set of vertices $Z$ reachable from $\emptyset$
in $G(\vec\lambda,\vec\delta,\vec\sigma,\vec\mu)$. A pair
$(I,J) \in 2^\calV \times 2^\calV$ is {\em good} for
$\calK$ if $I\in \calK$, $J\in \calK$, $I\cup J
\notin \calK$, and $\sigma_{I,J} > 0$.

\begin{algorithm}[th]
  \caption{Constructing a proof sequence for $\inner{\vec\lambda, \mv h} \leq \inner{\vec\delta,\mv h}$}
  \label{algo:calB:target:ps}
  \begin{algorithmic}[1]
    \Require{Inequality $\inner{\vec\lambda, \mv h} \leq \inner{\vec\delta,\mv h}$ with
    witness $(\vec\sigma,\vec\mu)$}
     \State \proofseq $\gets ()$
     \For {each $B\in 2^\calV$}
     \State $t \gets \min\{\lambda_B,\delta_{B | \emptyset}\}$;
     \myskip $\lambda_B \gets \lambda_B - t$;
     \myskip $\delta_{B | \emptyset} \gets \delta_{B | \emptyset} - t$
     \EndFor
     \State $w \gets 1/D$
     \While {$\filter(\vec\lambda) \neq \emptyset$}
       \State $\calK \gets \{ Z \suchthat Z \text{ is reachable from $\emptyset$
       in $G(\vec\lambda,\vec\delta,\vec\sigma,\vec\mu)$} \}$
       \If {($\filter(\vec\lambda) \cap \calK \neq \emptyset$)}\Comment{Case 1}
          \State Find a shortest path $\emptyset = X_0, X_1, \dots,
          X_\ell = B$ in $G(\vec\lambda,\vec\delta,\vec\sigma,\vec\mu)$ from $\emptyset$
          to some $B$ where $\lambda_B>0$\label{line:find-path}
          \For {($j \gets 1$ {\bf to} $\ell$)} \Comment{Push $w$-flow to $\F$}
          \If {($(X_{j-1},X_j)$ is an up arc)}
                \State Append $w \cdot \mv c_{X_{j-1},X_j}$ to \proofseq;
                \myskip
                $\delta_{X_j|X_{j-1}} \gets \delta_{X_j|X_{j-1}}-w$
             \Else \Comment{$(X_{j-1},X_j)$ is a down arc}
                \State Append $w \cdot \mv d_{X_{j-1},X_j}$ to \proofseq;
                \myskip
                $\delta_{X_{j-1}|X_{j}} \gets \delta_{X_{j-1}|X_{j}}+w$
             \EndIf
          \EndFor
          \State $\lambda_{B} \gets \lambda_B-w$; \label{line:reduce:lambda}
          \label{line:delta:B}
        \ElsIf {(There exists a good pair $(I,J)$ for $\calK$)}\Comment{Case 2}
          \State Find a shortest path $\emptyset = X_0, X_1, \dots, X_\ell = I$
          \For {($j \gets 1$ {\bf to} $\ell$)} \Comment{Push $w$-flow to $I$}
          \If {($(X_{j-1},X_j)$ is an up arc)}
                \State Append $w \cdot \mv c_{X_{j-1},X_j}$ to \proofseq;
                \myskip
                $\delta_{X_j|X_{j-1}} \gets \delta_{X_j|X_{j-1}}-w$
             \Else \Comment{$(X_{j-1},X_j)$ is a down arc}
                \State Append $w \cdot \mv d_{X_{j-1},X_j}$ to \proofseq;
                \myskip
                $\delta_{X_{j-1}|X_{j}} \gets \delta_{X_{j-1}|X_{j}}+w$
             \EndIf
          \EndFor
          \State Append $w \cdot \mv d_{I,I \cap J}$, then
          $w \cdot \mv s_{I,J}$ to \proofseq
          \State $\sigma_{I,J} \gets \sigma_{I,J}-w$;
          \myskip $\delta_{I\cup J|J} \gets \delta_{I\cup J|J}+w$;
          \myskip $\delta_{I\cap J|\emptyset} \gets \delta_{I\cap
          J|\emptyset}+w$. \label{line:2nd:branch}
       \EndIf
    \EndWhile
     \State \Return \proofseq
  \end{algorithmic}
\end{algorithm}

For any set $\calK \subset 2^\calV$ such that $\filter(\vec\lambda) \cap \calK =
\emptyset$, define $\outflow(\calK) \defeq \sum_{Z \notin \calK} \flow(Z)$.
We claim that the algorithm maintains the following two invariants:
\bi
 \item {\bf Invariant 1:} all the quantities $\flow(Z)-\lambda_Z$ are unchanged from
    one iteration to the next iteration of the algorithm.
 \item {\bf Invariant 2:} for every $\calK$ where $\calK \cap
\filter(\vec\lambda)=\emptyset$,
at the beginning of each iteration we have
\[ \sum_{B\in\filter(\vec\lambda)} \lambda_B \leq \outflow(\calK). \]
\ei
The first invariant is simple to verify: we specifically modified all
the entries $(\vec\lambda,\vec\delta,\vec\sigma,\vec\mu)$ to keep the flows
unchanged, except for the very last element $B$ in line \ref{line:find-path}, where the value $\flow(B)$ for that
element $B$ is reduced by $w$. However, $\lambda_B$ is also reduced by $w$
at line~\ref{line:reduce:lambda}, keeping $\flow(B)-\lambda_B$ constant.
The second invariant is satisfied at the beginning of the very first iteration,
and hence it is satisfied the the beginning of every later iteration due to
the first invariant.

We now show that the algorithm produces a valid proof sequence.
At the beginning of every iteration,
the inequality $\inner{\vec\lambda, \mv h} \leq \inner{\vec\delta,\mv h}$
is a Shannon flow inequality witnessed by $(\vec\sigma,\vec\mu)$,
due to invariant 1 and Proposition~\ref{prop:sfi:ddl:target}.
(Note that we modify $\vec\lambda$ and the witness inside each iteration.)
Hence, if we can show that the algorithm does terminate, then the proof
sequence it produces is valid.
To show termination, we show that as long as there exists a $B \in
\filter(\vec\lambda)$
for which $\lambda_B > 0$, then either Case 1 or Case 2 applies in the algorithm.

Assume that at the beginning of some iteration,
we have $\filter(\vec\lambda) \cap \calK \neq \emptyset$.
Then from~\eqref{eqn:filter}, there is some $B\subseteq F$ where $\lambda_B>0$ and $F\in\calK$,
and thanks to the down arc $(F, B)$, we have $B\in\calK$.
Now, assume that at the beginning of some iteration (where $\lambda_B>0$ for some $B$),
we have $\filter(\vec\lambda) \cap \calK = \emptyset$, yet there is no good
pair $(I,J)$ for $\calK$.
Then, $\outflow(\calK) \geq \lambda_B > 0$ due to Invariant 2.
However, all variables $\delta_{Y|X}, \mu_{X,Y}, \sigma_{I,J}$ contribute a
non-positive amount to $\outflow(\calK)$, which is a contradiction.

Next, we bound the proof sequence's length.
When any one of the two cases applies,
$\norm{\vec\lambda}_1+\norm{\vec\sigma}_1$ was reduced by $w=1/D$,
and at most $2^n$ steps are added to the proof sequence.
Hence the overall length of the proof sequence is at most
$2^n D(\norm{\vec\lambda}_1+\norm{\vec\sigma}_1)$.
\ep

\subsubsection{Final construction of a poly-sized proof sequence}
\label{subsec:poly-sized-proofseq}
We start by extending the flow network that we constructed earlier in Theorem~\ref{thm:ps:construction:2}.
Then, we prove a lower bound on the maximum flow of the extended network in Lemma~\ref{lmm:flow>lambda_1}.
Theorem~\ref{thm:almost-poly-proofseq} presents an advanced construction of a proof sequence
that is ``almost'' polynomial. Among other technical tools, the construction uses Edmond-Karp maximum flow Algorithm as a black-box~\cite{DBLP:journals/jacm/EdmondsK72}.
Finally, Corollary~\ref{cor:poly-proofseq} makes the final step and bounds the length of the proof sequence
to be truly polynomial in $2^n$.

\bdefn[Extended flow network]
Given $(\vec\lambda,\vec\delta,\vec\sigma,\vec\mu)$,
we define a flow network $\bar G(\vec\lambda,\vec\delta,\vec\sigma,\vec\mu)$
(which is an extended version of the network $G(\vec\lambda,\vec\delta,\vec\sigma,\vec\mu)$ defined in the proof of Theorem~\ref{thm:ps:construction:2}) as follows.
The set of vertices of the network $\bar G$ is ${\bar{\mcalV}}\defeq2^{\calV}\cup \calT\cup\{\bar T\}$,
where $\calT$ is a set that contains a new vertex $T_{I,J}$ for every pair $I,J\subseteq \calV$ such that $I\incomp J$, i.e.
\[\calT\defeq\left\{T_{I,J}\suchthat I,J\subseteq \calV \wedge I\incomp J\right\},\]
and $\bar T$ is yet another vertex that represents the \emph{sink} of the network.
The \emph{source} of the network $\bar G$ is $\emptyset$.
The set of arcs $\bar \calA$ of $\bar G$ consists of the following subsets:
\bi
\item There is an arc
$(X,Y)$ for every pair $(X,Y)$ for which
$X\subset Y\subseteq \calV$ and $\delta_{Y|X} >0$;
the capacity of $(X,Y)$ is $\delta_{Y|X}$.
These are called {\em up arcs} (as in the network $G$ from Theorem~\ref{thm:ps:construction:2}).
\item There is an arc $(Y,X)$ for every pair $(X,Y)$ such that $X \subset Y\subseteq \calV$;
The capacity of this arc is $+\infty$.
These are called {\em down arcs} (as in the network $G$).
\item For ever pair $I, J\subseteq \calV$ where $I\incomp J, \sigma_{I, J}>0$,
there are two arcs $(I, T_{I,J})$ and $(J, T_{I,J})$ each of which has an infinite capacity ($+\infty$), and there is a third arc $(T_{I,J}, \bar T)$ whose capacity is $\sigma_{I,J}$.
\item For every $\emptyset\neq B\subseteq \calV$ where $\lambda_B>0$,
there is an arc $(B,\bar T)$ whose capacity is $\lambda_B$.
\ei
\label{defn:flow-net-barG}
\edefn

\blmm[Maximum flow is $\geq \norm{\vec\lambda}_1$]
For any $(\vec\lambda,\vec\delta,\vec\sigma,\vec\mu)$,
the maximum flow of the network $\bar G(\vec\lambda,\vec\delta,\vec\sigma,\vec\mu)$ given by Definition~\ref{defn:flow-net-barG} is $\geq \norm{\vec\lambda}_1$.
\label{lmm:flow>lambda_1}
\elmm
\bp
We will use the min-cut max-flow theorem,
and show that every cut has a capacity $\geq \norm{\vec\lambda}_1$
(where the capacity of a cut is the total capacity of arcs crossing that cut).
In particular, for any $\calC\subset{\bar{\mcalV}}$ such that $\emptyset\in\calC, \bar T\notin \calC$,
and $\calC'\defeq{\bar{\mcalV}}\setminus \calC$,
we will show that the capacity of the $\calC$-$\calC'$ cut is $\geq\norm{\vec\lambda}_1$.
\bi
\item If there are $X\subset Y\subseteq \calV$ where $Y\in\calC$ and $X\notin\calC$,
then the down arc $(Y,X)$ crosses the cut, hence the cut capacity is $+\infty$.
\item Otherwise (i.e. if for any $X\subset Y\subseteq \calV$ where $Y\in\calC$, we have $X\in\calC$),
if there is a pair $I,J\subseteq \calV$ where $I\incomp J$, $\sigma_{I,J}>0$, $I\in\calC$ and $T_{I,J}\not\in\calC$,
then the cut capacity is also $+\infty$.
\item Otherwise, let $\calK\defeq\calC\cap 2^\calV$, $\calK'\defeq 2^\calV\setminus \calK$.
\begin{eqnarray*}
\sum_{Z \in \calK'} \flow(Z) &\geq& \sum_{B\in\calK'} \lambda_B\\
\sum_{\substack{I,J\in\calK\\ I\incomp J,\; I\cup J\in \calK'}} \sigma_{I,J}
+\sum_{\substack{X\in\calK,\;Y\in\calK'\\X\subset Y}}\delta_{Y|X}
&\geq& \sum_{B\in\calK'} \lambda_B
\end{eqnarray*}
However, we have 3 types of arcs crossing the cut:
$(T_{I,J},\bar T)$ where $I\in\calK$,
up arcs $(X, Y)$ where $X\in\calK, Y\in\calK'$,
and $(B,\bar T)$ where $B\in\calK$.
Hence the cut capacity is at least
\begin{eqnarray*}
\sum_{\substack{I\in \calK\\I\incomp J}}\sigma_{I,J}
+\sum_{\substack{X\in\calK,\;Y\in\calK'\\X\subset Y}}\delta_{Y|X}
+\sum_{B\in\calK} \lambda_B\geq \norm{\vec\lambda}_1.
\end{eqnarray*}
\ei
\ep

\blmm[Getting rid of the domination assumption]
Let $\inner{\vec\delta'_{\ell'}, \mv h}\leq\inner{\vec\delta'_0, \mv h}$ be an equality
that has a proof sequence $\proofseq'$ of length $\ell'$,
and let $\vec\delta_0$ be a vector that dominates $\vec\delta'_0$.
Then, there exists a vector $\vec\delta_\ell$ that dominates $\vec\delta'_{\ell'}$
and an inequality $\inner{\vec\delta_\ell, \mv h}\leq\inner{\vec\delta_0, \mv h}$
that has a proof sequence $\proofseq$ of length $\ell$ that satisfies
\[\ell=O(3^{n}\cdot\ell').\]
\label{lmm:remove-domination}
\elmm
\bp
Let $f$ be the domination function of $(\vec\delta_0,\vec\delta'_0)$.
Let $w'\cdot \mv f'=\vec\delta'_1-\vec\delta'_0$ be the first proof step in $\proofseq'$,
where $w'>0$ and $\mv f' \in \{\mv s_{I,J},\mv m_{X,Y},\mv c_{X,Y},\mv d_{Y,X}\}$.
Initially, let $\vec\delta_1$ be identical to $\vec\delta_0$.
\bi
\item If $\mv f'=\mv s_{I, J}$ for some $I\incomp J$,
then $(I\cap J, I)$ is dominated by a total of $w'$
(i.e. $\sum_{(X, Y)}f((X, Y),(I \cap J, I))\geq w'$).
We keep looking for pairs $(X,Y)\in\calP$ that dominate $(I\cap J, I)$
(i.e. such that $f((X, Y), (I\cap J, I))>0$),
and we make those pairs $(X, Y)$ dominate $(J,I\cup J)$ instead
(i.e. we reduce $f((X, Y), (I\cap J, I))$ and increase $f((X, Y), (J, I\cup J))$ by the same amount).
We keep doing so until $(J, I\cup J)$ is dominated by a total of $w'$.
Now, $f$ is a domination function for $(\vec\delta_1, \vec\delta'_1)$.
\item If $\mv f'=\mv m_{X', Y'}$ for some $X'\subset Y'$, then
we look for $(\emptyset, Y)$ that dominate $(\emptyset, Y')$
and we make them dominate $(\emptyset, X')$ instead
until $(\emptyset, X')$ is dominated by a total of $w'$.
Now, $f$ is a domination function for $(\vec\delta_1, \vec\delta'_1)$.
\item If $\mv f'=\mv d_{Y',X'}$ for some $X'\subset Y'$, then
we look for $(\emptyset, Y)$ that dominate $(\emptyset, Y')$,
and we append $w\cdot \mv m_{Y', Y}$ to $\proofseq$ (if $Y'\neq Y$),
and reduce $f((\emptyset, Y), (\emptyset, Y'))$ by the same amount $w$.
We keep doing so until the total of $w$ is equal to $w'$.
Now, we append $w'\cdot \mv d_{Y',X'}$ to $\proofseq$,
and increase $f((\emptyset, X'), (\emptyset, X'))$ and $f((X', Y'), (X', Y'))$ by $w'$.
\item If $\mv f'=\mv c_{X', Y'}$ for some $X'\subset Y'$,
then we look for $(\emptyset, X)$ that dominate $(\emptyset, X')$,
and we append $w\cdot \mv m_{X', X}$ to $\proofseq$,
and reduce $f((\emptyset, X), (\emptyset, X'))$ by the same amount $w$,
until the total of $w$ is equal to $w'$.
Now, we look for $(X, Y)$ that dominates $(X', Y')$
\footnote{From Definition~\ref{defn:delta-domination},
if $f((X, Y), (X', Y'))>0$, then $Y'\setminus X'\subseteq Y\setminus X$ and $X'\supseteq X$.},
and append
$w_2\cdot \mv m_{A, X'}$
(where $A\defeq X'\setminus (Y\setminus X)$),
$w_2\cdot \mv s_{Y,A}$,
$w_2\cdot \mv c_{A,Y\cup A}$,
$w_2\cdot \mv m_{Y',Y\cup A}$ to $\proofseq$,
and we reduce $f((X, Y), (X', Y'))$ by the same amount $w_2$.
We keep doing so until the total of $w_2$ is equal to $w'$.
Finally, we increase $f((\emptyset, Y'),(\emptyset, Y'))$ by $w'$.
The total number of pairs $(X, Y)$ where $X\subset Y\subseteq \calV$ is $\leq 3^n$.
\ei
\ep

\begin{algorithm}[th!]
\caption{Constructing a poly-sized proof sequence for $\inner{\vec\lambda, \mv h} \leq \inner{\vec\delta,\mv h}$}
\label{algo:poly-proofSeq}
\begin{algorithmic}[1]
\Require{Inequality $\inner{\vec\lambda, \mv h} \leq \inner{\vec\delta,\mv h}$
with witness $(\vec\sigma,\vec\mu)$, where $\norm{\vec\lambda}_1=1$}
\State \proofseq $\gets ()$
\State $D\gets \Call{least-common-denominator}{(\vec\lambda,\vec\delta,\vec\sigma,\vec\mu)}$
\Comment{$D$ is the least common denominator of all entries}
\State $\vec\lambda\gets D \cdot\vec\lambda$;
\myskip $\vec\delta\gets D \cdot\vec\delta$;
\myskip $\vec\sigma\gets D\cdot\vec\sigma$;
\myskip $\vec\mu\gets D\cdot\vec\mu$
\Comment{Now, all entries are integers}
\State Find $\inner{\vec\lambda, \mv h} \leq \inner{\vec\delta',\mv h}$
with witness $(\vec\sigma',\vec\mu')$
where $\vec\delta'$ is dominated by $\vec\delta$
and $\norm{\vec\sigma'}_1\leq \frac{1}{2}n^3\norm{\vec\lambda}_1$\\
\Comment{Corollary~\ref{cor:2sigma<=n-cube-lambda}}
\State $\Delta\gets \max\left\{2^i \suchthat 2^i\leq \norm{\vec\lambda}_1 \text{ and $i\in \N$}\right\}$\label{ln:poly-initial-Delta}
\Comment{Now, we have $\Delta\leq\norm{\vec\lambda}_1< 2\Delta$ and $\norm{\vec\sigma'}_1<n^3\Delta$}
\While {$\norm{\vec\lambda}_1>0$} \label{ln:poly-while}
\State Construct $\bar G(\vec\lambda,\vec\delta',\vec\sigma',\vec\mu')$ \Comment{Definition~\ref{defn:flow-net-barG}}
\label{ln:poly-net-construction}
\State $P\gets$\Call{Edmond-Karp}{$\bar G(\vec\lambda,\vec\delta',\vec\sigma',\vec\mu')$}.
\Comment{Apply Edmond-Karp algorithm for maximum flow}\\
\Comment{P is a collection of augmenting paths}
\For {each augmenting path $\emptyset=X_0,X_1,\ldots, X_{\ell}=\bar T$ with capacity $w'$ in $P$}
\label{ln:poly-for-all-path}
\For {$j\gets 1$ to $\ell$}
\If{$(X_{j-1},X_j)=(B,\bar T)$ for some $\emptyset\neq B\subseteq \calV, \lambda_B>0$}
\State $\lambda_B \gets \lambda_B - w'$;
\myskip $\delta'_{B|\emptyset}\gets \delta'_{B|\emptyset}-w'$
\label{ln:poly-reduce-lambda}
\State $t\gets w'$
\While {$t > 0$}
\State Find $Y \supseteq B$ such that $\delta_{Y|\emptyset}>0$
\Comment{Possible because $\vec\delta$ dominates $\vec\delta'$}
\State $w\gets\min\{t, \delta_{Y|\emptyset}\}$;
\myskip $t\gets t-w$;
\myskip $\delta_{Y|\emptyset}\gets\delta_{Y|\emptyset}-w$
\If {$Y\supset B$}
\State Append $w\cdot \mv d_{Y, B}$ to $\proofseq$;
\myskip $\delta_{Y|B}\gets\delta_{Y|B}+w$
\EndIf
\EndWhile
\Else
\State $\vec\delta'_0\gets\vec\delta'$;
\myskip $\proofseq_0'\gets ()$;
\myskip $\vec\delta_0\gets\vec\delta$
\If {$(X_{j-1},X_j)$ is an up arc}
\State Append $w' \cdot \mv c_{X_{j-1},X_j}$ to $\proofseq_0'$;
\myskip $\vec\delta'\gets\vec\delta'+w' \cdot \mv c_{X_{j-1},X_j}$
\ElsIf {$(X_{j-1},X_j)$ is a down arc}
\State Append $w' \cdot \mv d_{X_{j-1},X_j}$ to $\proofseq_0'$;
\myskip $\vec\delta'\gets\vec\delta'+w' \cdot \mv d_{X_{j-1},X_j}$
\ElsIf {$(X_{j-1},X_j)=(I,T_{I,J})$ for some $I,J\subseteq \calV, I\incomp J, \sigma_{I,J}>0$}
         \State Append $w' \cdot \mv d_{I,I \cap J}$, then
          $w' \cdot \mv s_{I,J}$ to $\proofseq_0'$;
          \myskip \myskip $\vec\delta'\gets\vec\delta'+w' \cdot \mv d_{I,I \cap J}+w' \cdot \mv s_{I,J}$
          \State $\sigma'_{I,J} \gets \sigma'_{I,J}-w'$ \label{ln:poly-reduce-sigma}
\EndIf
\Comment{Now, $\proofseq_0'$ is a proof sequence for $\inner{\vec\delta', \mv h}\leq\inner{\vec\delta'_0, \mv h}$,
and $\vec\delta_0$ dominates $\vec\delta'_0$}
\State Find $\vec\delta$ dominating $\vec\delta'$
and a proof sequence $\proofseq_0$ for $\inner{\vec\delta, \mv h}\leq\inner{\vec\delta_0, \mv h}$
\Comment{Lemma~\ref{lmm:remove-domination}}
\State Append $\proofseq_0$ to $\proofseq$
\EndIf
\EndFor
\EndFor
\If {$\norm{\vec\lambda}_1<\Delta$} \label{ln:poly-lambda<Delta}
\State Find $\inner{\vec\lambda, \mv h} \leq \inner{\vec\delta',\mv h}$
with witness $(\vec\sigma',\vec\mu')$
where \State \quad\quad $\vec\delta'$ is dominated by $\vec\delta$
and $\norm{\vec\sigma'}_1\leq \frac{1}{2}n^3\norm{\vec\lambda}_1$
\Comment{Corollary~\ref{cor:2sigma<=n-cube-lambda}}
\State $\Delta\gets \max\left\{2^i \suchthat 2^i\leq \norm{\vec\lambda}_1 \text{ and $i\in \N$}\right\}$\label{ln:poly-update-Delta}
\Comment{Now, we have $\Delta\leq\norm{\vec\lambda}_1< 2\Delta$ and $\norm{\vec\sigma'}_1<n^3\Delta$}
\EndIf
\EndWhile
\State \Return \proofseq
\end{algorithmic}
\end{algorithm}

\bthm[Construction of a poly-sized proof sequence, modulo $\log D$]
For any Shannon flow inequality $\inner{\vec\lambda, \mv h}\leq \inner{\vec\delta, \mv h}$
with witness $(\vec\sigma,\vec\mu)$ where $\norm{\vec\lambda}_1=1$,
Algorithm~\ref{algo:poly-proofSeq} produces a proof sequence
whose length is $O(\log D \cdot n^3\cdot 2^{6n}\cdot 3^n)=O(\log D\cdot \poly(2^n))$,
where $D$ is the common denominator of all entries in $(\vec\lambda,\vec\delta,\vec\sigma,\vec\mu)$.
\label{thm:almost-poly-proofseq}
\ethm

\bp
Given a flow network $G$ with vertex set $\mcalV$ and arc set $\calA$,
Edmond-Karp algorithm~\cite{DBLP:journals/jacm/EdmondsK72, Cormen:2009:IAT:1614191} finds the maximum flow in time
$O(|{\mcalV}|\cdot |\calA|^2)$.
Edmond-Karp is based on augmenting paths.
In particular, it finds $O(|\mcalV|\cdot|\calA|)$ such paths,
each of which has length $\leq|\mcalV|$ and can be found in time $O(|\calA|)$.
Each path also has a \emph{capacity} $c$, which is the minimum capacity among arcs of that path.
The total capacity of all augmenting paths is equal to the maximum flow.

The flow network $\bar G(\vec\lambda,\vec\delta',\vec\sigma',\vec\mu')$
constructed in Line~\ref{ln:poly-net-construction} of Algorithm~\ref{algo:poly-proofSeq}
has a vertex set $\mcalV$ of size $|\bar \mcalV|=O(2^{2n})$
and an arc set $\calA$ of size $|\bar\calA|=O(2^{2n})$.
By Lemma~\ref{lmm:flow>lambda_1}, it has a maximum flow $\geq \norm{\vec\lambda}_1\geq \Delta$ .
The loop in Line~\ref{ln:poly-for-all-path} maintains
the quantities $\flow'(Z)-\lambda_Z$ for all $\emptyset\neq Z\subseteq \calV$ unchanged
(where $\flow'(Z)$ is the quantity $\flow(Z)$ measured on the vector
$(\vec\delta',\vec\sigma',\vec\mu')$.).
For each augmenting path in the loop, we either reduce some $\lambda_B$ (in Line~\ref{ln:poly-reduce-lambda})
or $\sigma'_{I,J}$ (in Line~\ref{ln:poly-reduce-sigma}).
After we are done with all paths in the loop,
the quantity $\norm{\vec\sigma'}_1+\norm{\vec\lambda}_1$ is reduced by at least $\Delta$,
and a total of $O(2^{6n})$ proof steps where appended to $\proofseq'_0$.
By Lemma~\ref{lmm:remove-domination}, the original proof sequence $\proofseq$
is extended by $O(2^{6n}\cdot 3^n)$ proof steps.

Before the loop in line~\ref{ln:poly-while}, we had
$\Delta\leq\norm{\vec\lambda}_1< 2\Delta$ and $\norm{\vec\sigma'}_1<n^3\Delta$,
which both continue to hold and $\Delta$ remains fixed as long as
$\norm{\vec\lambda}_1$ remains $\geq\Delta$ in line~\ref{ln:poly-lambda<Delta}.
The maximum number of iterations the loop in line~\ref{ln:poly-while} can make before $\norm{\vec\lambda}_1$ drops below $\Delta$ can be bounded as follows.
At each iteration, the quantity $\norm{\vec\sigma'}_1+\norm{\vec\lambda}_1$ is reduced by at least $\Delta$.
However, $\norm{\vec\lambda}_1$ cannot be reduced by more than $\Delta$ in total (since $\Delta\leq\norm{\vec\lambda}_1< 2\Delta$).
Also, $\norm{\vec\sigma'}_1$ cannot be reduced by more than $n^3\Delta$ in total.
Therefore the maximum number of iterations before $\norm{\vec\lambda}_1$ drops below $\Delta$ is bounded by $n^3+1$.

Every time $\norm{\vec\lambda}_1$ drops below $\Delta$, $\Delta$ gets divided by 2 (or a higher power of 2).
The initial $\Delta$ in line ~\ref{ln:poly-initial-Delta} was $\leq D$.
Therefore, the maximum number of times $\norm{\vec\lambda}_1$ can drop below $\Delta$ is $O(\log D)$.
\ep

\bprop[$\log D$ is polynomial]
Given an optimization problem of the form \eqref{eqn:ddl:target:first:version},
if the objective value is positive and bounded,
then there exists a vector $\vec\lambda=(\lambda_B)_{B\in\calB}$ satisfying the following conditions:
\bi
\item[(a)] $\norm{\vec\lambda}_1=1$.
\item[(b)] The optimization problem~\eqref{eqn:ddl:target:first:version}
has the same optimal objective value as the linear program~\eqref{eqn:ddl:target} (using this $\vec\lambda$).
\item[(c)] The linear program~\eqref{eqn:ddl:target} has an optimal dual solution $(\vec\delta^*,\vec\sigma^*,\vec\mu^*)$ such that the common denominator $D$ of all entries in $(\vec\lambda,\vec\delta^*,\vec\sigma^*,\vec\mu^*)$ satisfies
\[D\leq (2^{n})!\]
\ei
\label{prop:logD-is-poly}
\eprop
\bp
This proposition is a minor extension to Lemma~\ref{lmm:rewrite},
specialized to the optimization problem~\eqref{eqn:ddl:target:first:version}.
In particular, when we pick an optimal dual solution $(\mv z^*, \mv y^*)$ to \eqref{eqn:LP4},
we choose $(\mv z^*, \mv y^*)$ to be an extreme point of the following polyhedron
\[\{(\mv z, \mv y)\suchthat \mv A^T\mv y \geq \mv C \mv z, \mv 1_p^T\mv z \geq 1, (\mv z, \mv y)\geq \mv 0\}.\]
In \eqref{eqn:ddl:target:first:version}, all entries of $\mv A$ and $\mv C$ are in $\{1, 0, -1\}$.
By Cramer's rule, the common denominator $D^*$ of all entries in $(\mv z^*, \mv y^*)$ is
(the absolute value of) the determinant of an $m\times m$ matrix whose entries are all in $\{1,0,-1\}$.
Hence, $D\leq m!$.
\ep

\bcor[Construction of a poly-sized proof sequence]
Given any optimization problem of the form \eqref{eqn:ddl:target:first:version}
where the optimal objective value $\obj$ is positive and bounded,
there exists a Shannon flow inequality $\inner{\vec\lambda, \mv h}\leq\inner{\vec\delta,\mv h}$
satisfying the following conditions:
\bi
\item $\sum_{(X, Y)}\delta_{Y|X}n_{Y|X}\leq\obj.$
\item $\lambda_{Y|X}=0$ for all $(X, Y)\in\calP$ where $X\neq\emptyset$ or $Y\notin\calB$.
\item $\inner{\vec\lambda, \mv h}\leq\inner{\vec\delta,\mv h}$ has a proof sequence of length
$O(n^4\cdot 2^{7n}\cdot 3^n)=O(\poly(2^n))$.
\ei
\label{cor:poly-proofseq}
\ecor

\end{document}